\documentclass[12pt,a4paper,reqno]{article}

\usepackage{amssymb,amsmath,amsfonts,geometry,graphicx,caption,color,setspace,natbib,subcaption,array,hyperref,amsthm,float,tikz,url}

\usetikzlibrary{decorations.pathreplacing}
\usepackage[T1]{fontenc}
\usepackage{newpxmath}
\usepackage{newpxtext}
 \usepackage{threeparttable}
\usetikzlibrary{patterns}

\usepackage[titletoc]{appendix}

\allowdisplaybreaks

\theoremstyle{plain}\newtheorem{theorem}{Theorem}
\theoremstyle{plain}\newtheorem{proposition}{Proposition}
\theoremstyle{plain}\newtheorem{corollary}{Corollary}
\theoremstyle{plain}\newtheorem{lemma}{Lemma}
\theoremstyle{plain}

\theoremstyle{definition}
\newtheorem{definition}{Definition}

\newtheorem{remark}{Remark}

\newcommand{\mR}{\mathbb{R}}
\newcommand{\mE}{\mathbb{E}}
\newcommand{\mN}{\mathbb{N}}
\newcommand{\mP}{\mathbb{P}}
\newcommand{\mF}{\mathbb{F}}

\newcommand{\mI}{\mathbb{I}}

\newcommand{\lr}[1]{\texttt{LR}(#1)}

\usetikzlibrary{calc}
\usetikzlibrary{positioning}
\usetikzlibrary{matrix,fit}
\usetikzlibrary{intersections}

\makeatletter
\let\@makefntextOring\@makefntext
\def\@makefntext#1{\@makefntextOring{\baselineskip=15pt#1}}
\makeatother

\definecolor{navyblue}{rgb}{0.0, 0.0, 0.5}

\hypersetup{
    colorlinks=true,
    citecolor=navyblue,
    linkcolor=navyblue,
    urlcolor=navyblue,
}

\def\citeapos#1{\citeauthor{#1}\textcolor{navyblue}{'s} (\citeyear{#1})}

\usepackage{etoolbox}
\patchcmd{\thanks}{#1}{\protect\doublespacing}{}{}

\DeclareFontShape{OT1}{cmr}{m}{n}{<->cmr10}{}
\usepackage{fix-cm}

\usepackage{titlesec}

\renewcommand{\epsilon}{\varepsilon}

\begin{document}

\begin{titlepage}
\title{Allocating Common-Value Goods\thanks{\protect First draft: December 2025. This draft is a substantially revised version of "Managing Learning Structures." We thank Ian Ball, Parag Pathak, Stephen Morris, and Alex Wolitzky, for their support and guidance. We are also grateful to Yu Awaya, Louis Becker, Eric Gao, Jack Hirsch, Michihiro Kandori, Andrew Koh, Yukio Koriyama, Mikael M\"{a}kimattila, Shinpei Noguchi, Kao Nomura, Sivakorn Sanguanmoo, Yucheng Shang, and Konan Shimizu, as well as audiences at Boston-Harvard-MIT Theory student group, MIT lunches, Musashi University, Tokyo University of Science, Keio University, Game Theory Workshop 2026, Summer Workshop on Economic Theory 2026, Contract Theory Workshop, and The University of Tokyo for insightful comments. 
Sato acknowledges the financial support from the JSPS KAKENHI Grant 24KJ0100. 
Shirakawa acknowledges the financial support from the Funai Foundation for Information Technology. 
All remaining errors are our own.
}}
\author{
\Large Hiroto Sato\thanks{Nagoya University. Email: \url{sato.hiroto.s9@f.mail.nagoya-u.ac.jp}.}
\and 
\Large Ryo Shirakawa\thanks{Massachusetts Institute of Technology. Email: \url{shira723@mit.edu}.}
} 
\date{\today}
\maketitle
\begin{abstract} 
We study a simple problem of allocating common-value goods. The designer seeks to allocate the goods to as many unit-demand agents as possible without monetary transfers, while agents, who possess partial private information about the goods, are willing to receive them only when the goods are of high value.
Mechanisms screen each agent’s private information using the information of other agents, and in doing so shape what agents learn from other agents about the value of the goods. 
The optimal mechanism can be summarized by two parameters: one adjusts the allocation probability, while the other governs the amount of learning induced by allocation. Although the designer prefers to allocate the goods, the optimal mechanism excludes some agents and, as a result, may withhold allocation even when all agents would be willing to receive them. The optimal mechanism has the same structure even when payments are available, but it does not exclude any agent and involves strictly positive payments that are not increasing in allocation probability. 

\vspace{0in} 
\end{abstract}
\setcounter{page}{0}
\thispagestyle{empty}
\end{titlepage}


\newcolumntype{C}[1]{>{\centering\arraybackslash}p{#1}}
\onehalfspacing

\setlength{\abovedisplayskip}{5pt}
\setlength{\belowdisplayskip}{5pt}

\section{Introduction} \label{sec: introduction}
What agents know about an option can influence other agents’ decisions. In many environments, individuals are only partially informed about the value of taking an opportunity, and therefore, whether or not they can observe other agents' private information naturally affects their beliefs. Such interactions can, in turn, affect whether socially valuable opportunities are ultimately accepted or rejected. The implications of these information externalities for decision making and market allocations have been explored in various contexts in the literature on information economics.

In the context of allocation design, such information externalities need not be taken as fixed; rather, they can themselves be objects of design. This perspective is relevant because many allocation problems in practice naturally share two features. First, participants are unlikely to be fully informed about the value of the objects to be allocated. Second, they nevertheless possess partial private information about that value, which is also valuable to other participants because it provides an additional signal about the objects' value. In such environments, an allocation mechanism can serve not only to assign objects but also to reveal participants’ private information to one another through allocation outcomes, thereby shaping their ex post payoffs conditional on receiving an object and, in turn, their ex ante incentives to manipulate the mechanism. 

We study how information externalities can be managed and exploited in allocation design. To this end, we develop a simple model of a design problem with the features described above. There are many homogeneous goods, whose value is binary: either high or low.\footnote{We focus on the case in which the value is common across agents, but the framework extends to values with idiosyncratic components; see Remark \ref{rem: interdependent value}.} The designer aims to allocate these goods as widely as possible without monetary transfers, whereas agents are willing to accept them only when they believe the value is high.\footnote{As noted in Remark \ref{rem: objective function}, our main results continue to hold under more general objective functions, such as convex combinations of the allocation probability and agents' welfare.} Agents have partial information about the value of the goods, but are not fully informed.

We aim to develop a conceptual theoretical framework, but our model captures features common to many allocation problems. One example is vaccine distribution, where public health authorities seek broad uptake, while individuals do not internalize the positive externalities. These individuals have own perceptions about vaccine safety based on the experiences of people around them. 
Another example is organ allocation, where transplant organizations seek to promote the use of available organs while patients, who possess heterogeneous assessment about organ quality based on information from their physicians, may refuse them unless the expected benefits are sufficiently high. 
Similar issues arise in capital raising. A firm seeks broad investment, while potential investors are willing to invest only when their predictions about its future performance are sufficiently favorable. 

In our model, mechanisms can influence an agent's allocation only by using the private information of the other agents. In particular, because monetary transfers are unavailable, if a mechanism cannot condition an agent's allocation on information held by others, then the only feasible allocation rule is such that the agent receives the good only if her private belief exceeds a given cutoff. In this sense, informational externalities are the main instrument for shaping allocations, and any non-trivial mechanism must exploit it by screening each agent's private information through the information of others. 

Our first result, Theorem \ref{thm: optimal mechanism}, shows that the optimal mechanism has a monotone threshold structure. The mechanism first partitions each agent's type space into ordered intervals and then assigns the same allocation rule to all types within a given interval. The allocation rule for each interval is characterized by two parameters. One parameter does not affect the information externality and determines how aggressively the designer allocates within that interval. The other specifies the threshold level of evidence for high value, inferred from others' information, that is required for allocation. In particular, whenever an agent receives the good, the agent's conditional belief that it is of high quality is weakly higher than the agent's initial belief. 

In the proof, we work with indirect interim expected utility. Because there are no monetary transfers, feasibility of a mechanism will require infinitely many linear constraints rather than a single monotonicity condition, making the set of feasible indirect utility functions hard to characterize. Hence, we begin by identifying a set of simple necessary conditions for feasibility and then narrow this set by appealing to optimality. We next show that the designer's objective function can be written as a linear functional of indirect utility as in standard mechanism design problems, and that any utility function within the reduced set corresponds to a mechanism in Theorem \ref{thm: optimal mechanism}. Further details appear in Subsection \ref{subsec: proof idea}.  

One economic implication, formalized in Proposition \ref{prop: exclusion at the bottom}, is that the optimal mechanism excludes low-belief agents from the market even when every agent is willing to receive the good. More specifically, even a designer who seeks to minimize the waste of, for example, organ-transplant opportunities or vaccine stocks may optimally exclude those who initially assign a low value to them, even when the market's average belief about the quality is high and sharing that information would induce anyone to accept the good. In classical revenue-maximization problems in mechanism design, such exclusion is standard because it helps reduce information rents. In our setting, however, this result is not obvious, since the designer always prefers to allocate the good. As we discuss later, when the designer can also choose payments, the optimal mechanism excludes no types. 

The intuition behind Proposition \ref{prop: exclusion at the bottom} lies in a trade-off in the management of incentives for manipulating information externalities. Under the participation constraint, allocating the good to low-belief types is feasible only when the private information of other agents provides sufficiently strong evidence of high quality to offset their pessimistic beliefs. Such a scheme, however, would induce more optimistic types to misreport downward in order to receive the good only when they are almost certain that its quality is high. To maintain a high allocation probability for high types, the designer therefore eliminates such manipulations. In the capital-raising environment discussed earlier, this result and its intuition suggest that a firm may optimally screen out low-confidence investors not because it wishes to limit participation, but because any scheme that incentivizes them to invest would undermine the information elicitation. 


Theorem \ref{thm: optimal mechanism} provides a general implication for how information externalities can be structured through allocation itself, while Theorems \ref{thm: optimal mechanism log-concave} and \ref{prop: large market} offer sharper characterizations of the structure of the optimal mechanism. When the signal distribution is log-concave, Theorem \ref{thm: optimal mechanism log-concave} shows that the optimal mechanism collapses to a two-threshold structure: the lowest types are excluded, the middle types receive a constant allocation, and the highest types face no distortion and obtain the efficient allocation. Theorem \ref{prop: large market} shows that, as the number of participants becomes large, a family of similar two-threshold mechanisms is asymptotically optimal under general distributions. Within this class, however, the highest types are always allocated the good. 

Proposition \ref{prop: laissez-faire outcome} also provides a sufficient condition under which the \textit{laissez-faire outcome}, in which each agent independently decides whether to accept the good based solely on her private information, is optimal among all mechanisms. This result is of particular interest because it has an application to a social learning model as in \citet{banerjee1992simple}, \citet{bikhchandani1992theory}, and \citet{acemoglu2011bayesian}, even though our model concerns a static allocation problem and may therefore appear unrelated. In a social learning environment, agents make accept-or-reject decisions sequentially after observing their private beliefs and a subset of their predecessors' decisions. We show that any equilibrium of this social-learning game induces a direct mechanism in our model, so our mechanism-design problem provides an upper bound on what can be achieved by designing observation structures in the social learning model. Proposition \ref{prop: laissez-faire outcome} therefore identifies a sufficient condition under which completely concealing all predecessors' actions, which induces the laissez-faire outcome, is optimal among all possible observation structures. 

In the discussion section, we extend the model and consider payment design. For a given allocation rule, payments may appear only to tighten participation constraints and therefore merely reduce the allocation probability. Payments, however, can serve as a qualitatively new screening device, which is helpful for deterring downward deviations. If a payment is imposed only when the other agents' beliefs are sufficiently high, then higher types perceive a higher expected payment, which in turn makes high-belief types less willing to misreport as low-belief types. Through this channel, the optimal mechanism generally involves positive payments. Moreover, unlike in standard mechanism design, the optimal payment may be decreasing in beliefs, consistent with the idea of deterring downward manipulations. We elaborate on these points by extending Theorems \ref{thm: optimal mechanism} and \ref{prop: large market} to this environment. 

Taken together, this paper introduces a simple model of managing information externalities into mechanism design. To summarize, the paper has three main implications. First, even when monetary transfers are infeasible or inappropriate, the designer can manipulate allocation by screening each participant's private information about the common-value goods through the information of others. Second, in such environments, a threshold rule can increase the allocation probability by raising agents' valuations conditional on being allocated the good; however, in order to control incentives for downward manipulations, it may withhold allocation even when all agents prefer to receive the good. Third, although payments tighten participation constraints, higher types perceive higher expected payments from the same payment function, which can be used to deter downward manipulations and therefore improve the allocation probability. 

\subsection{Related literature} \label{subsec: literature} 

Broadly speaking, this study contributes to the mechanism design literature by bringing a simple problem of managing what participants learn from other participants. Specifically, our study contributes to the following strands of the literature.  

This study builds on methodologies developed in the literatures on mechanism design \citep{mussa1978monopoly, myerson1981optimal, toikka2011ironing} and delegation problems \citep{Holmstrom1984, melumad1991communication}. \citet{rochet1987necessary} characterizes indirect utility functions induced by incentive compatible mechanisms, and \citet{daskalakis2017strong} and \citet{kleiner2022optimal} respectively reformulate multidimensional mechanism design problems with and without monetary transfers in terms of indirect utility functions.\footnote{See also \citet{rochet1998ironing}, \citet{manelli2007multidimensional}, and \citet{kovavc2009stochastic}, among many others.} We also adopt this approach in our setting. Among others, the largest difference from this literature arises when the designer is also allowed to choose payments. Specifically, the \textit{revenue equivalence} theorem, namely, the property that the allocation rule uniquely pins down payments, fails in our setting. As a result, the indirect utility function alone is no longer sufficient to characterize the optimization problem. Consequently, the optimal payment may be decreasing in allocation probability, a feature that cannot arise in many standard mechanism design problems. 

We also build on a recent literature that studies the extreme points of functional spaces arising in mechanism design \citep{kleiner2021extreme, yang2024monotone, kleiner2024extreme,yang2025multidimensional}. In our setting, joint restrictions on the slopes and intercepts of feasible indirect utility functions make it difficult to represent the extreme points of the feasible set in a simple form. We therefore use optimality to construct a simple subset and show that the optimum is attained within this subset. This subset is called a convex function interval, whose extreme points are characterized by \citet{augias2025economics}. We note that the technical implication of Theorem \ref{thm: optimal mechanism} lies in reducing our problem to an optimization problem over a convex function interval, and this step is independent of their analysis. When payments are allowed, as discussed above, the objective function depends not only on indirect utility functions in such a convex function interval, but also on payment functions. We therefore need to establish several additional results on the structure of the optimal payment function in order to bring the problem back to a similar form.

More specifically, our model is related to the literature on mechanism design with correlated valuations. A well-known result in this literature is the \textit{full-surplus-extraction} result, which shows that under a sufficiently rich correlation structure, there exists a payment rule that extracts the entire surplus and, more generally, that makes almost any allocation rule implementable \citep{myerson1981optimal,cremer1985optimal,cremer1988full,mcafee1992correlated,lopomo2022detectability}. Such payments generally involve negative transfers, however, which are ruled out in our model. This is because we interpret payments not as literal payments, but rather as payment-like instruments, such as waiting time, that can only impose a burden on the agents.\footnote{If the designer is allowed to use negative payments, we prove Proposition \ref{prop: full surplus extraction}, which implies that the designer can always allocate the good using an ex ante budget-balanced payment. In our model, correlation is not sufficiently rich in the sense of that literature, and the proof therefore requires a separate argument. See Section \ref{sec: discussion} for details.} 

Relatively few papers work on mechanism design problems with correlated values without monetary transfers. \citet{kattwinkel2024optimal} analyze collective decision making between two alternatives among agents. \citet{kattwinkel2020allocation} examines how to allocate a good to an agent who always demands it and shows that the principal may withhold allocation when the agent's valuation is very high and the principal's allocation cost is low, because such realizations appear too good to be true under positive correlation. \citet{niemeyer2022simple} studies a similar environment but with many participants who have correlated private types, and show that what they call "jury mechanisms" are optimal among a class of dominant-strategy incentive compatible mechanisms.\footnote{In our environment, requiring dominant-strategy implementation results in a trivial solution, which we discuss in Proposition \ref{prop: BIC-EPIC non-equivalence}.} One of the central ideas in this line of work is that when signals are correlated across agents, the designer can use cross-checks of reports to detect manipulations. By contrast, our insight is that correlation can also be exploited through allocations themselves: by shaping information externalities, the designer manipulates beliefs and thereby disciplines incentives for misreports. 

In our model, we interpret payments as socially wasteful costs such as waiting time, and in this respect our study is also related to the literature on money-burning problems. For example,  \citet{mcafee1992bidding}, \citet{hartline2008optimal}, \citet{yoon2011optimal}, \citet{condorelli2012money}, and \citet{chakravarty2013optimal} study efficient allocation with costly signals and emphasize the trade-off between allocative efficiency and the deadweight cost of screening.\footnote{For more recent developments of this literature, see, e.g., \citet{noda2024no}, \citet{tokarski2024equitable}, \citet{tokarski2025screening}, \citet{yang2025costly}, \citet{yang2025comparison},   \citet{dworczak2026allocate}, and \citet{salamat2026optimal}.} The role of payments in our model is similar to these studies, but differs in at least two ways. First, in our environment the designer can already screen agents even when payments are unavailable. In this sense, the potential benefit of payments is limited. Second, the designer may nonetheless find payments valuable because their feature is distinct from that of conventional payments. In our setting, due to correlation, contingent payments are perceived differently by different types, and this property can be exploited to deter downward misreporting. 

Finally, although our setting may appear irrelevant, the motivation of this paper resonates with the literature on queue design, particularly the strand that emphasizes informational externalities arising during the allocation process.\footnote{Pioneered by \citet{naor1969regulation}, there is a large literature on optimal queue design without quality uncertainty but with a focus on other important aspects. See \citet{hassin2003queue} and \citet{hassin2016rational} for comprehensive surveys.} \citet{zhang2010sound} and \citet{doval2024social} study organ transplantation in U.S. queue-based systems and provide empirical evidence that what agents learn from others affects their decisions and, in turn, the allocation probability.\footnote{\citet{su2004patient,su2005patient,su2006recipient} also study the assignment of transplant organs through waiting lists and propose mechanisms that induce agents to accept marginal kidneys in order to reduce organ wastage.} \citet{leshno2022dynamic} establishes a similar implication theoretically in waiting-list-with-decline systems commonly used in practice. \citet{kremer2014implementing} and \citet{che2018recommender} study mechanism design problems in social learning models as in \citet{banerjee1992simple}, \citet{bikhchandani1992theory}, and \citet{smith2000pathological} and show that agents do not internalize the feedback effects of exploration, implying that withholding information about past actions may improve social welfare.\footnote{In traditional social learning models, many studies explore network structures which induce efficient learning. See, e.g., \citet{acemoglu2011bayesian}, \citet{lobel2015information}, and \citet{kartik2024beyond}.} In Subsection \ref{subsec: laissez-faire}, we return to this literature and discuss a theoretical connection between our model and models of social learning. 


The remainder of the paper is organized as follows. Section \ref{sec: model} presents the model and formulates the mechanism design problem. In Section \ref{sec: optimal mechanism}, we present the main results and provide a sketch of the proof. Section \ref{sec: discussion} introduces a model with payment design and discusses its implications. Section \ref{sec: conclusion} concludes. All proofs are collected in the Appendix.

\section{Model} \label{sec: model}

The state space is binary, $\omega\in \{-1,+1\}$, with each state occurring with probability $1/2$. There are finitely many agents, indexed by $i\in \{1,\dots,n\}$, each with unit demand. Each agent who receives a good obtains payoff $\omega$, and $0$ otherwise. In Remark \ref{rem: interdependent value}, we discuss how we can extend our results into a more general setting where the payoff is given by $\omega+\epsilon_{i}$ with some idiosyncratic private shock $\epsilon_{i}$. 

Each agent $i$ has a \textit{private signal}, which is a random variable whose distribution may depend on the state $\omega$. Assume that each agent draws her private signal from a common state-contingent distribution $\mF_{\omega}\in \Delta(S)$ on some signal space $S$. Private signals are conditionally independent given the state, but are ex ante correlated. Let $\mF=[\mF_{-1}+\mF_{+1}]/2$ be the unconditional distribution. 

Given a realization of private signal $s_{i}$, agent $i$ forms a \textit{private belief} about high state $\mP[\omega = +1 \mid s_{i}]$. Since private signals are conditionally independent across agents, it is without loss of generality to relabel each private signal by its induced private belief, that is, $s_{i} = \mP[\omega = +1 \mid s_{i}]$. Then, note that $s_{i}$ must have mean equal to the prior $1/2$. Under this normalization, $\mF_{\omega} \in \Delta([0,1])$ is the conditional distribution of private beliefs. 

We impose two standard assumptions from the mechanism design literature. First, we assume that each conditional distribution $\mF_{\omega}$ is mutually absolutely continuous and admits a density $f_{\omega}$. Then, the unconditional distribution $\mF$ also admits a density, which we denote by $f$. Second, we assume that the support of $f$ is a non-singleton interval $[\underline{s},\overline{s}] \subset [0,1]$, and that $f$ is differentiable on its support with a continuous and bounded derivative $df/ds_{i}$.\footnote{Formally, we require that the derivative $df/ds_{i}$ be continuous and that there exist a constant $c\in \mR$ such that $|df(s_{i})/ds_{i}|\leq c$ for all $s_{i}$ in the support.} Whether the interval is closed does not play an essential role. 

In our study, the likelihood ratio is often more convenient to work with than posterior beliefs directly. For any event $E$, typically a subset of the set of all private-signal profiles $[0,1]^{n}$, the likelihood ratio is defined by
\begin{align*}
    \lr{E} = \frac{\mP[E\mid \omega=+1]}{\mP[E\mid \omega=-1]}.
\end{align*}
It follows from Bayes' rule that the posterior belief $\mP[\omega=+1\mid E]$ weakly exceeds $1/2$ if and only if $\lr{E}\geq 1$. Note that the normalization $s_{i}=\mP[\omega=+1\mid s_{i}]$ together with Bayes' rule implies $\lr{s_{i}}=s_{i}/(1-s_{i})$.

\subsection{Mechanism design problem} \label{sec: mechanism design} 

Here, we set up the mechanism design problem of allocating homogeneous goods. A \textit{(direct) mechanism} is a function $x:[\underline{s},\overline{s}]^{n}\rightarrow [0,1]^{n}$. For each profile of private beliefs $s=(s_{i})_{i}$, the mechanism $x$ allocates the good to agent $i$ with probability $x_{i}(s)\in[0,1]$. No monetary transfers are allowed. In Section \ref{sec: discussion}, we extend the model to allow the designer to choose payments as well. 

For a given mechanism $x$, each agent $i$ with private belief $s_{i}$ obtains interim payoff $\mE[\omega\cdot x_{i}(s_{i},s_{-i})\mid s_{i}]$. We say that $x$ is \textit{feasible} if it satisfies both the interim participation constraint and the Bayesian incentive compatibility condition: 
\begin{align*}
    \mE[\omega\cdot x_{i}(s_{i},s_{-i})\mid s_{i}] &\geq 0, \label{eq: P}\tag{P} \\
    \mE[\omega\cdot x_{i}(s_{i},s_{-i})\mid s_{i}] &\geq \mE[\omega\cdot x_{i}(\hat{s}_{i},s_{-i})\mid s_{i}], \label{eq: IC}\tag{IC}
\end{align*}
for each $s_{i},\hat{s}_{i}\in [\underline{s},\overline{s}]$ and $s_{-i}\in [\underline{s},\overline{s}]^{n-1}$. Note that the distribution of $s_{-i}$ depends on $s_{i}$ because they are ex ante correlated. 

We assume that the designer seeks to maximize the expected number of allocated goods. In other words, the optimization problem is given by
\begin{align*}
    \max_{x} \mE\left[\sum_{i}x_{i}(s)\right] 
    \text{ s.t. } \eqref{eq: P} \text{ and } \eqref{eq: IC},
\end{align*}
where the expectation is taken with respect to the unconditional distribution of private signals. However, our first main theorem does not rely on this particular specification and extends to objective functions that are convex combinations of this objective and the utilitarian sum, which we discuss in Remark \ref{rem: objective function}. 

Although \eqref{eq: IC} and \eqref{eq: P} are written in terms of interim expected payoffs, the law of iterated expectations allows each payoff to be decomposed according to whether the agent is allocated the good. Since the agent obtains zero otherwise, her interim payoff equals the probability of allocation multiplied by her expected value of the good conditional on being allocated it.\footnote{Therefore, the optimization problem remains unchanged even if we additionally require that each agent be willing to receive the good conditional on being allocated it: this requirement is implied by \eqref{eq: P}.} This conditional value is shaped by what the agent infers about the other agents’ private information from being allocated the good. Hence, by changing how allocation for each agent depends on the others' reports, the mechanism manages information externalities and thereby affects both participation and reporting incentives.

Our problem differs from classic mechanism design problems in two respects. First, monetary transfers are not allowed. When monetary transfer is available, as in \citet{mussa1978monopoly} and \citet{myerson1981optimal}, incentive compatibility is eventually characterized by monotonicity of interim allocations. In our setting, by contrast, we must contend directly with infinitely many constraints. Second, values are ex ante correlated across agents. With monetary transfer, as we discuss in Section \ref{sec: discussion}, \citet{cremer1988full} shows that a rich correlation structure renders essentially any allocation rule implementable. Here, by contrast, the absence of monetary transfers imposes substantive restrictions on feasible mechanisms. 

Our problem is also closely related to the delegation problems by \citet{Holmstrom1984} and \citet{melumad1991communication}. This literature studies environments in which an uninformed principal delegates a decision to an informed agent whose preferences are misaligned with her own. There is a technical parallel in that the principal in delegation models also does not rely on monetary transfers to provide incentives. There are, however, important differences too. One common theme in this literature is to identify conditions under which an optimal contract can be implemented through delegation of a simple set of actions, and to address this question the literature typically assumes quadratic payoffs, a continuous action and state space, and a single agent.\footnote{Recent papers such as \citet{alonso2008optimal}, \citet{kovavc2009stochastic}, and \citet{kleiner2022optimal} relax the quadratic structure, but still impose concave payoffs, which are central to their analysis. Also related is \citet{gan2023optimal}, which studies a multi-agent delegation problem in which the state space consists of a profile of independently distributed private information.}

\section{Optimal mechanism} \label{sec: optimal mechanism} 

In this section, we explore the structures of the optimal mechanism. 

\begin{definition}\label{def: monotone threshold mechanism}
    Call $x_{i}$ a \textit{monotone threshold mechanism} if (i) there is a partition $\mathcal{S}_{i}$ of the signal space $[\underline{s},\overline{s}]$ such that for each element $S_{i}\in \mathcal{S}_{i}$, there exist two thresholds $\kappa\in [0,1]$ and $\tau\geq 0$ such that 
    \begin{align*} 
        x_{i}(s_{i},s_{-i}) = \kappa \cdot \mI\{\lr{s_{-i}}\geq \tau\}
    \end{align*}
    for all $s_{i}\in S_{i}$ and $s_{-i}\in [\underline{s},\overline{s}]^{n-1}$, and (ii) for the coarsest such partition $\mathcal{S}_{i}$, all elements of $\mathcal{S}_{i}$ are intervals. 
\end{definition}

The monotone threshold mechanism partitions the signal space into monotone intervals, within each of which all types are treated symmetrically. Any type in a given interval then faces a common, signal-independent allocation rule of the form $\kappa \cdot \mI\{\lr{s_{-i}}\geq \tau\}$. The first parameter, $\kappa\in [0,1]$, scales the overall intensity of allocation and thus controls the "volume" of trade, while the second parameter, $\tau\geq 0$, governs the extent to which allocation is conditioned on informative events, that is, the "degree of learning" embedded in the mechanism. Conditional on obtaining the object, agent $i$ learns that the likelihood ratio of the other agents' signals is at least $\tau$. Hence, under any monotone threshold mechanism, the posterior belief conditional on allocation is always weakly more favorable than the private belief.

A simple example of a monotone threshold mechanism is an \textit{efficient allocation}
\begin{align*}
    x_{i}(s_{i},s_{-i}) = \mI\{\lr{s_{i},s_{-i}}\geq 1\},
\end{align*}
which allocates to the agent if and only if her posterior belief conditional on observing the entire signal profile exceeds $1/2$. Since conditional independence implies $\lr{s_{i},s_{-i}}=\lr{s_{i}}\cdot \lr{s_{-i}}$, the efficient allocation is a monotone threshold mechanism with a partition that consists of singletons $\mathcal{S}_{i}=\{\{s_{i}\}\mid s_{i}\in [\underline{s},\overline{s}]\}$, and with threshold $\tau=\lr{s_{i}}^{-1}$ on each interval $\{s_{i}\}$. 

Not every monotone threshold mechanism is feasible. For example, consider the constant mechanism that always allocates the good. Although this mechanism is clearly in the class of monotone threshold mechanisms with the coarsest partition $\mathcal{S}_{i}=\{[\underline{s},\overline{s}]\}$ and the lowest threshold $\tau=0$, it violates the participation constraint \eqref{eq: P} for low types $s_{i} < 1/2$. The participation constraint is substantially more restrictive here than in standard auction problems. 

The following is the first main technical result of this paper. 

\begin{theorem}\label{thm: optimal mechanism}
    There exists a monotone threshold mechanism that is optimal. 
\end{theorem}

Under any smooth distribution of private beliefs, an optimal mechanism exists within the class of monotone threshold mechanisms. Although the proof is somewhat involved, the idea is relatively simple and also useful in discussing the other results. Accordingly, the formal proof is deferred to the Appendix, where we establish a stronger statement, and we discuss the proof strategy in the following subsection. 

There is another generic property of the optimal mechanism. The optimal mechanism excludes low-belief types. 

\begin{proposition}\label{prop: exclusion at the bottom} 
    Suppose $\text{supp } \mF = (0,1)$. 
    Any optimal mechanism excludes a bottom type, i.e., there exists $\epsilon_{i}>0$ such that $x_{i}(s_{i},s_{-i})=0$ for all $s_{i}\leq \epsilon_{i}$ and $s_{-i}$. 
    Consequently, efficient allocation is always suboptimal. 
\end{proposition} 

In particular, even when allocating the good would be socially efficient, namely when $\lr{s_{i},s_{-i}} \ge 1$, the mechanism does not allocate the good to such types. Such exclusion at the bottom is a standard implication in conventional revenue-maximization problems, where it serves to reduce information rents and thereby increase the seller's surplus. In our setting, however, the notion of information rent is not clear. Moreover, when $\lr{s_{i},s_{-i}} \ge 1$, the preferences of the designer and the participants are aligned, and the designer wishes to allocate the good whenever possible. As a corollary, efficient allocation is never optimal. 

The intuition is that the designer must balance incentives for \textit{manipulating information externalities}. At a high level, the participation constraint implies that the designer can allocate the good to the lowest-belief types only when the other agents hold sufficiently positive signals. But doing so may invite deviations by higher-belief types, who may mimic those low types in order to receive the good only when they are nearly certain that its quality is high. Thus, the designer faces a trade-off between providing appropriate incentives and maintaining a high allocation probability. Proposition \ref{prop: exclusion at the bottom} shows that the designer resolves this tension by excluding bottom types. 

Figure \ref{fig: uniform} illustrates the optimal mechanism for agent $1$ in the uniform case $\mF(s_{i})=s_{i}$ with two agents, where we can represent the space of signal profiles over the plane. In this environment, the optimal mechanism is deterministic and allocates the good to agent $1$ if and only if the reported belief profile lies in the shaded region. We have $s_{1}^{\texttt{min}}=3/8$ and $s_{1}^{\texttt{max}}=3/4$, which implies that types in $[0,3/8]$ are excluded. Since $\lr{s_{1},s_{2}}\ge 1$ is equivalent to $s_{1}+s_{2}\ge 1$, the 45-degree line identifies the efficient allocation in this case. In the Appendix, we formally derive this figure.\footnote{Note that the area of the shaded region does not represent the objective value because the signals are ex ante correlated and the joint distribution of $(s_{1},s_{2})$ is not uniform. Consequently, although the shaded area, which equals $1/2$, coincides with that under efficient allocation, our mechanism strictly outperforms the efficient allocation.} 

\begin{figure}[ht]
    \centering
\begin{tikzpicture}[scale=0.5]
\fill[gray!30] (3.75,10) -- (3.75, 2.5) -- (7.5,2.5) -- (10,0) -- (10,10) -- (3.75,10);
\draw[dashed] (0,10) -- (10,0);
\draw[dashed] (3.75,0) -- (3.75,10);
\node at (3.75, 0) [below,font=\scriptsize]{$s^{\texttt{min}}_{1}$};
\draw[dashed] (7.5,0) -- (7.5,10);
\node at (7.5, 0) [below,font=\scriptsize]{$s^{\texttt{max}}_{1}$};
\draw[dashed] (0,2.5) -- (10,2.5);
\node at (0, 2.5) [left,font=\scriptsize]{$1-s_{1}^{\texttt{max}}$};
\draw[thick] (0,0) -- (10,0) -- (10,10) -- (0,10) -- (0,0);
\draw[->, thick] (0,0) -- (11,0);
\draw[->, thick] (0,0) -- (0,11);
\node at (0,0) [below left,font=\scriptsize]{$0$};
\node at (10,0) [below,font=\scriptsize]{$1$};
\node at (0,10) [left,font=\scriptsize]{$1$};
\node at (11,0) [right,font=\scriptsize]{$s_{1}$};
\node at (0,11) [above,font=\scriptsize]{$s_{2}$};
\node at (5,10) [above,font=\scriptsize]{$x_{1}(s_{1},s_{2})$};
\end{tikzpicture}
    \caption{}
    \par
  \makeatletter\def\TPT@hsize{}\makeatletter
  \begin{tablenotes}
       \footnotesize
       \item[] \textit{Notes:}
       The optimal mechanism under uniform distribution allocates the good to agent $1$ if and only if the report $(s_{1},s_{2})$ is in the shaded region. 
       It excludes the bottom types as in Proposition \ref{prop: exclusion at the bottom}. 
       Note that the middle types receive the good even if the signal profile is below the 45 degree line $s_{1}+s_{2}=1$, at which we have $\lr{s_{1},s_{2}}<1$ and it is not socially optimal to allocate the good. 
       Therefore, further allocating to the bottom types in the upper triangle region $s_{2}\geq 1-s_{1}$ violates the incentive compatibility of middle types.
       \end{tablenotes}
    \label{fig: uniform}
\end{figure}

In the uniform case, the optimal mechanism is even simpler than the structure described in Theorem \ref{thm: optimal mechanism}, and exhibits a two-threshold form. It excludes the bottom types $[0,3/8]$, assigns the constant allocation rule $\mI\{s_{2}\geq 1/4\}$ to the middle types in $[3/8,3/4]$, and leaves the top types $[3/4,1]$ undistorted. As we show in Subsection \ref{subsec: two-threshold}, this structure extends to a broader class of private-belief distributions. 

Before moving on to the subsections, we provide one remark. Our results depend on adopting Bayesian incentive compatibility, rather than \textit{ex-post incentive compatibility}, which requires for any $s_{i}$ and $s_{-i}$,
\begin{align*}
    \mE[\omega\cdot x_{i}(s_{i},s_{-i})\mid s_{i},s_{-i}] &\geq \mE[\omega\cdot x_{i}(\hat{s}_{i},s_{-i})\mid s_{i},s_{-i}]. \label{eq: EPIC}\tag{EPIC}
\end{align*} 
In certain classes of mechanism design problems, it is well known that any interim incentive compatible mechanism admits an outcome-equivalent ex-post incentive compatible counterpart: See, e.g., \citet{gershkov2013equivalence}. In our setting, however, imposing these two requirements leads to a markedly different conclusion. 

\begin{proposition}\label{prop: BIC-EPIC non-equivalence}
    Suppose $\text{supp }\mF = (0,1)$. 
    Then, the efficient mechanism defined as $x_{i}(s_{i},s_{-i})=\mI\{\lr{s_{i},s_{-i}}\geq 1\}$ uniquely maximizes the designer's payoff among all mechanisms that satisfy \eqref{eq: P} and \eqref{eq: EPIC} up to a measure-zero set. 
\end{proposition} 

Intuition is simple. Since monetary transfers are unavailable, the designer can influence the agent's incentives only by shaping beliefs. This manipulation, however, operates solely through interim incentives. Ex-post incentives depend directly on the realized signals and therefore cannot be manipulated. Accordingly, the designer can allocate the object only when the agent is willing to have it, which implies the optimality of an efficient allocation. 

\subsection{Proof sketch of Theorem \ref{thm: optimal mechanism}} \label{subsec: proof idea} 

The formal proof of Theorem \ref{thm: optimal mechanism} is in the Appendix. Here, we highlight a key idea behind the proof. For expositional convenience, suppose throughout this subsection that the distribution of private signals has maximal support $\text{supp }\mF=(0,1)$. Relaxing this assumption complicates the proof but does not affect the main steps of the argument. 

We mainly work with the space of indirect utility functions $U_{i}$ induced by mechanisms. By definition, we have $U_{i}(s_{i})=\max_{t_{i}\in [0,1]} U_{i}(t_{i};s_{i})$, where
\begin{align*}
    U_{i}(t_{i};s_{i})
    &= \mE[\omega\cdot x_{i}(t_{i},s_{-i})\mid s_{i}] \\
    &= s_{i}\sum_{\omega\in \{-1,+1\}}\mE[x_{i}(t_{i},s_{-i})\mid \omega] - \mE[x_{i}(t_{i},s_{-i})\mid \omega=-1].
\end{align*} 
Notice that the normalization $s_{i}=\mP[\omega=+1\mid s_{i}]$ plays a crucial role here, as it implies that the indirect utility functions still exhibit several structures familiar from standard mechanism design problems. 

The above expression immediately yields several properties of the indirect utility functions. First, it is the upper envelope of linear functions and is therefore convex. Second, the envelope theorem implies that the slope of $U_{i}$ at each point $s_{i}$ is given by $\sum_{\omega}\mE[x_{i}(s_{i},s_{-i})\mid \omega]$. Since $x_{i}(s)\in [0,1]$, it follows immediately that the subgradients of $U_{i}$ lie in the interval $[0,2]$. Third, the participation constraint requires that $U_{i}\geq 0$. Finally, because monetary transfers are not available, $U_{i}$ must lie pointwise below the interim payoff function induced by the efficient allocation, which we denote by $\overline{U}_{i}$.

Therefore, any feasible indirect utility is a convex, increasing, and $2$-Lipschitz continuous function that belongs to a shaded area in Figure \ref{fig: necessity triangle}.\footnote{For a general case where the support of private signals is a strict subset of the unit interval $[0,1]$, we prove that every feasible indirect utility function has an appropriate \textit{extension} that has all these properties over the extended domain $[0,1]$.} 
In the Appendix, we also show that the envelope formula enables us to rewrite the objective function as a linear functional of indirect utility. 
Therefore, if Figure \ref{fig: necessity triangle} characterizes the set of all feasible indirect utility functions, the problem is reduced to maximizing a linear functional subject to this shape constraint.\footnote{\citet{kleiner2022optimal} provides an analogous characterization in multidimensional delegation problems, which leverages the strict concavity of payoff functions.} 
It should be noted, however, that this argument holds only because payments are not allowed. In the next section, we augment the design problem by allowing payments and show that the designer's objective is no longer measurable with respect to agents' indirect utilities. 

\begin{figure}[ht]
    \centering
\begin{tikzpicture}[scale=0.4]
\fill[gray!30] (10,10) -- (10,0) -- (0,0) plot[domain=0:10, smooth, variable=\x] (\x, {10*(\x/10)^2}) -- cycle; 
\draw[->, thick] (0,0) -- (11,0) node[right,font=\scriptsize] {$s_{i}$};
\node at (10,0) [below,font=\scriptsize]{$1$};
\node at (0,0) [below,font=\scriptsize]{$0$};
\draw[->, thick] (0,0) -- (0,11);
\draw[ultra thick, smooth, domain=0:10, variable=\x]
    plot (\x, {10*(\x/10)^2}); 
\node at (10,10) [above,font=\scriptsize]{$\overline{U}_{i}$};
\draw[dashed,ultra thick] (10,0) -- (10,10);
\end{tikzpicture}
    \caption{}
    \par
  \makeatletter\def\TPT@hsize{}\makeatletter
  \begin{tablenotes}
       \footnotesize
       \item[] \textit{Notes:}
       Every feasible indirect utility function is convex, increasing, $2$-Lipschitz, and lies in the shaded area.
       \end{tablenotes}
    \label{fig: necessity triangle}
\end{figure}

However, not all functions in the shaded area are implementable.  
In particular, consider a function
\begin{align*}
    U_{i}(s_{i})
    = 
    \begin{cases}
        0 \quad &\text{if } s_{i}\leq 3/4, \\
        2s_{i}-3/2 \quad &\text{if } s_{i}\geq 3/4.
    \end{cases}
\end{align*} 
Although this function satisfies all necessary shape restrictions we have listed above, it cannot arise from any feasible mechanism. Intuition is simple: having a maximum slope of $2$ requires a mechanism to always allocate the good, which however forces the intercept $\mE[x_{i}(s)\mid \omega=-1]$ to equal $1$. 

This example also illustrates that obtaining a precise characterization of the set of feasible indirect utility functions is difficult. In particular, a feasible utility function can attain slope $2$ only along the line $2s_{i}-1$, which corresponds to the indirect utility under the mechanism that always allocates the good. This observation further shows that the set of admissible slopes at any point $s_{i}$ is constrained by the height of the function itself.

\begin{figure}[ht]
\centering
\begin{subfigure}{0.48\textwidth}
    \centering
    \begin{tikzpicture}[scale=0.4]
        \draw[->, thick] (0,0) -- (11,0) node[right,font=\scriptsize] {$s_{i}$};
        \node at (10,0) [below,font=\scriptsize]{$1$};
        \node at (0,0) [below,font=\scriptsize]{$0$};
        \draw[->, thick] (0,0) -- (0,11);
        \draw[dashed, ultra thick, smooth, domain=0:10, variable=\x]
            plot (\x, {10*(\x/10)^2});
        \draw[dashed, ultra thick] (5,0) -- (10,10);
        \node at (5,0) [below,font=\scriptsize]{$\frac{1}{2}$};
        \draw[line width=2pt, smooth, domain=0:10, variable=\x] plot (\x, {5*(\x/10)^2});
        \node at (10,5) [right,font=\scriptsize]{$U_{i}(s_{i})$};
    \draw[dashed,ultra thick] (10,0) -- (10,10);
    \end{tikzpicture}
    \caption{}
    \label{fig:2s-1_a}
\end{subfigure}
\hfill
\begin{subfigure}{0.48\textwidth}
    \centering
    \begin{tikzpicture}[scale=0.4]
        \draw[->, thick] (0,0) -- (11,0) node[right,font=\scriptsize] {$s_{i}$};
        \node at (10,0) [below,font=\scriptsize]{$1$};
        \node at (0,0) [below,font=\scriptsize]{$0$};
        \draw[->, thick] (0,0) -- (0,11);
        \draw[dotted, ultra thick, smooth, domain=0:10, variable=\x]
            plot (\x, {10*(\x/10)^2});
        \node at (10,10) [above,font=\scriptsize]{$\max\{U_{i}(s_{i}),2s_{i}-1\}$};
        \draw[dotted, ultra thick] (5,0) -- (10,10);
        \node at (5,0) [below,font=\scriptsize]{$\frac{1}{2}$};
        \draw[dashed, line width=2pt, smooth, domain=0:10, variable=\x] plot (\x, {5*(\x/10)^2});
            \pgfmathsetmacro{\sI}{2 - sqrt(2)}          
            \pgfmathsetmacro{\xI}{10*(2 - sqrt(2))}      
            \pgfmathsetmacro{\yI}{10*(3 - 2*sqrt(2))}    
        \draw[line width=2pt, smooth, domain=0:\xI, variable=\x] plot (\x, {5*(\x/10)^2}); 
        \draw[line width=2pt] (\xI,\yI) -- (10,10);
        \draw[dotted,ultra thick] (10,0) -- (10,10);
    \end{tikzpicture}
    \caption{}
    \label{fig:2s-1_b}
\end{subfigure}
\caption{}
\par
  \makeatletter\def\TPT@hsize{}\makeatletter
  \begin{tablenotes}
       \footnotesize
       \item[] \textit{Notes:}
       For every indirect utility function that crosses the line $2s_{i}-1$, the designer prefers to shift the indirect utility upward up to the line $2s_{i}-1$.
  \end{tablenotes}
\end{figure} 

A key step in the argument is to rule out a class of indirect utility functions by appealing to optimality. Suppose that a feasible utility function intersects the line $2s_{i}-1$ at some interior point, as depicted in Figure \ref{fig:2s-1_a}. We then argue that the designer should replace $U_{i}$ with its pointwise maximum and $2s_{i}-1$, as in Figure \ref{fig:2s-1_b}. The resulting indirect utility can be implemented by augmenting the original menu mechanism with an option that allocates the good with probability one. The agent then chooses whichever option yields the higher interim payoff, thereby generating the upper envelope. Since this modification increases the probability that the good is allocated, it is preferred by the designer. 

We can then conclude that the optimal mechanism must induce an indirect utility function that lies in the triangular region depicted in Figure \ref{fig:optimal indirect}. Note, again, that this triangular region does not characterize the set of feasible indirect utility functions. For example, $U_{i}(s_{i})=0$ for all $s_{i}$ is clearly implementable by the mechanism $x_{i}(s_{i},s_{-i})=0$, which never allocates the good, but it is ruled out by optimality.

\begin{figure}[ht]
    \begin{subfigure}{0.48\textwidth}
    \centering
\begin{tikzpicture}[scale=0.4]
\fill[gray!30]
    (10,10) -- (5,0) -- (0,0)
    plot[domain=0:10, smooth, variable=\x]
        (\x, {10*(\x/10)^2})
    -- cycle;
\draw[->, thick] (0,0) -- (11,0) node[right,font=\scriptsize] {$s_{i}$};
\node at (10,0) [below,font=\scriptsize]{$1$};
\node at (0,0) [below,font=\scriptsize]{$0$};
\draw[->, thick] (0,0) -- (0,11);
\draw[ultra thick, smooth, domain=0:10, variable=\x]
    plot (\x, {10*(\x/10)^2});
\node at (10,10) [above,font=\scriptsize]{$\overline{U}_{i}$};
\draw[ultra thick] (5,0) -- (10,10);
\node at (5,0) [below,font=\scriptsize]{$\frac{1}{2}$};
\draw[ultra thick] (0,0) -- (5,0);
\draw[dashed,ultra thick] (10,0) -- (10,10);
\end{tikzpicture}
\caption{}
    \label{fig:optimal indirect}
  \end{subfigure}
\hfill
\begin{subfigure}{0.48\textwidth}
    \centering
    \begin{tikzpicture}[scale=0.4]
\fill[gray!30]
    (10,10) -- (5,0) -- (0,0)
    plot[domain=0:10, smooth, variable=\x]
        (\x, {10*(\x/10)^2})
    -- cycle;
\draw[->, thick] (0,0) -- (11,0) node[right,font=\scriptsize] {$s_{i}$};
\node at (10,0) [below,font=\scriptsize]{$1$};
\node at (0,0) [below,font=\scriptsize]{$0$};
\draw[->, thick] (0,0) -- (0,11);
\draw[thick, smooth, domain=0:10, variable=\x]
    plot (\x, {10*(\x/10)^2});
\node at (10,10) [above,font=\scriptsize]{$\overline{U}_{i}$};
\draw[thick] (5,0) -- (10,10);
\node at (5,0) [below,font=\scriptsize]{$\frac{1}{2}$};
\draw[dashed,ultra thick] (0.65,-0.3) -- (10,4.375); 
\node at (10,4.5) [right,font=\scriptsize]{$\kappa_{i}=1$};
\draw[ultra thick] (2.4,-0.3) -- (10,3.5); 
\node at (10,3.5) [right,font=\scriptsize]{$a_{i}s_{i}-b_{i}$};
\draw[dashed,ultra thick] (4.4,-0.3) -- (10,2.5); 
\node at (10,2.5) [right,font=\scriptsize]{$\kappa_{i}=a_{i}/2$};
\draw[dashed,ultra thick] (10,0) -- (10,10);
\end{tikzpicture}
\caption{}
    \label{fig:slope intermediate}
\end{subfigure}
    \caption{}
  \begin{tablenotes}
       \footnotesize
       \item[] \textit{Notes:}
       (a) An optimal indirect utility lies in a truncated triangle. 
       (b) In the Appendix, we show that if $a_{i}\in [0,2]$, each $\kappa \in [a_{i}/2,1]$ admits an associated threshold $\tau$ that generates slope $a_{i}$. 
       Varying $\kappa\in [a_{i}/2,1]$ while holding the slope fixed then traces out precisely the set of intercepts for which the resulting line intersects the triangle region. 
  \end{tablenotes}
  \label{fig:3}
\end{figure}

The final step is to show that every convex function lying within the triangular region is implementable. For each pair of parameters $\kappa\in[0,1]$ and $\tau\geq 0$, a mechanism $\kappa\cdot \mI\{\lr{s_{-i}}\geq \tau\}$ is independent of the report $s_{i}$ and therefore generates a linear indirect utility, with slope and intercept jointly determined by $(\kappa,\tau)$; see Figure \ref{fig:slope intermediate} for an illustration. We show that, for any line segment with slope in $[0,2]$, the intermediate value theorem can be used to find a parameter pair $(\kappa,\tau)$ that implements it if and only if the segment intersects the triangular region. Since any convex function within the triangle can be expressed as the upper envelope of such linear functions, it is implementable via an appropriate monotone threshold mechanism. 

A by-product of this proof that is as important as the theorem itself is a characterization of the class of optimal indirect utility functions by the triangular region. Given that the objective is a linear functional of the indirect utility function, this characterization reduces our optimization problem to a linear program. As explained below, the results in the following subsections are established on the basis of this observation. 

A few remarks that follow from this proof are in order.

\begin{remark}[Objective function]\label{rem: objective function}
    As one can infer from the discussion thus far, our proof uses little beyond a basic monotonicity property of the objective. In fact, Theorem \ref{thm: optimal mechanism} goes through with essentially the same argument as long as the objective weakly increases when we replace any utility function with its pointwise maximum with $2s_{i}-1$. For example, if the designer also values participants' welfare and maximizes
    \begin{align*}
        \mE\left[\alpha \sum_{i}x_{i}(s) + (1-\alpha)\sum_{i}U_{i}(s_{i})\right] 
    \end{align*}
    for some weight $\alpha \in [0,1]$, we can still conclude that a monotone threshold mechanism is optimal.
    \qed
\end{remark}

\begin{remark}[Signal distributions]\label{rem: signal distribution}
    We do not rely on any symmetry of the signal distributions $\mF$ across agents. The result continues to hold even when each agent $i$ draws an ex-ante heterogeneous private signal distributed according to a CDF $\mF^{i}$. However, we must retain conditional independence, which, though implicit, is essential for the normalization $s_{i}=\mP[\omega=+1\mid s_{i}]$ to be valid. Similarly, we also need to assume that each agent's distribution admits a density with a continuous and bounded derivative. 
    \qed
\end{remark} 

\begin{remark}[Idiosyncratic shocks]\label{rem: interdependent value}
    A similar result holds when the value of the good has both common and idiosyncratic components. Suppose that agent $i$ obtains payoff $\omega+\epsilon_{i}$ from receiving the good, where $\epsilon_{i}\in [-1,+1]$ is additional private information that is iid across agents and follows a CDF $\mF_{\epsilon}$ with a convex-support density $f_{\epsilon}$. A mechanism is then a mapping $x_{i}:[0,1]^{n}\times [-1,+1]^{n}\to [0,1]$, where the two components of the input are the profiles of private beliefs and idiosyncratic shocks. 

    To define a class of optimal mechanisms in this generalized problem, let 
    \begin{align*}
        t_{i}(s_{i},\epsilon_{i}) = \frac{(1+\epsilon_{i})\cdot s_{i}}{(1+\epsilon_{i})\cdot s_{i} + (1-\epsilon_{i})\cdot (1-s_{i})} \in [0,1]
    \end{align*}
    be the \textit{effective type} of agent type $(s_{i},\epsilon_{i})$. Then, an optimal mechanism is given as follows: there exists a partition $\mathcal{T}_{i}$ of the space of effective types $[0,1]$ into intervals such that, for each $T_{i}\in \mathcal{T}_{i}$, there exists $(\kappa,\tau)$ such that $x_{i}(s,\epsilon) = \kappa \cdot \mI\{\lr{s_{-i}}\geq \tau\}$ for every $(s_{i},\epsilon_{i})$ with $t_{i}(s_{i},\epsilon_{i})\in T_{i}$. 

    An intuition is the following. The effective type basically summarizes the utility level from receiving the good in the high state relative to the low state. In particular, two types with the same effective type have the same preferences over allocation rules.\footnote{In the other models of multidimensional screening without monetary transfer, some papers use similar observations; that is, mechanisms cannot screen agent types who have identical preferences. See, e.g., \citet{lahr2024extreme} and \citet{tokarski2026targeting}.} An incentive-compatible mechanism therefore cannot offer them different outcomes. Hence, without loss of generality, agent $i$'s allocation can be taken to depend on $(s_{i},\epsilon_{i})$ only through $t_{i}$. The problem thus reduces to a similar single-dimensional screening problem studied above, and the proof of Theorem \ref{thm: optimal mechanism} applies with minor modifications. Further details are provided in the Appendix. 
    \qed
\end{remark}

\begin{remark}[Queue-based allocation]\label{rem: queue-based allocation}
    We can also extend our result to a scenario where agents arrive sequentially. 
    A mechanism $x:[\underline{s},\overline{s}]^{n}\to [0,1]^{n}$ is a \textit{queue-based mechanism} if, for each $i$, $x_{i}(s)$ depends only on $(s_{1},\dots,s_{i})$, not on information of later agents. Within this class, maximizing $\mE[x_{i}(s_{1},\dots,s_{i})]$ reduces to the corresponding problem with only $i$ agents $\{1,\dots,i\}$. Therefore, we can conclude that there exists an optimal \textit{monotone threshold queue-based mechanism} $x$ among all queue-based mechanisms: for each agent $i$, there is a partition $\mathcal{S}_{i}$ of the signal space $[\underline{s},\overline{s}]$ into disjoint intervals such that, for each $S_{i}\in\mathcal{S}_{i}$, there exist $\kappa\in[0,1]$ and $\tau\ge 0$ satisfying
    \begin{align*}
        x_{i}(s_i,s_{-i})=\kappa\cdot \mI\{\lr{s_{1},\dots,s_{i-1}}\ge \tau\}
    \end{align*}
    for all $s_{i}\in S_{i}$ and all $s_{-i}\in[\underline{s},\overline{s}]^{n-1}$. In subsection \ref{subsec: laissez-faire}, we further discuss the connection between our model and queue-based allocation design, focusing on \citeapos{acemoglu2011bayesian} social learning models with general observation structures. 
    \qed
\end{remark}

\begin{remark}[Extreme point approach]\label{rem: extreme point}
    The proof of Theorem \ref{thm: optimal mechanism} yields a more detailed characterization of the partition $\mathcal{S}_{i}$ in the optimal monotone threshold mechanism. The discussion so far implies that the optimal indirect utility function maximizes a linear functional over what \citet{augias2025economics} call a \textit{convex function interval}, namely, the set of all convex functions sandwiched between two convex functions:
    \begin{align*}
        \mathcal{U}_{i}^{*} = \left\{U_{i}\ \middle|\ U_{i} \text{ is convex and } \underline{U}_{i}\leq U_{i}\leq \overline{U}_{i} \right\},
    \end{align*}
    where $\underline{U}_{i}(s_{i})=0$ for all $s_{i}$. Bauer's maximum principle implies that the optimal indirect utility function is an extreme point of this set.

    \citet{augias2025economics} characterize the extreme points of convex function intervals. In our setting, their result implies that any extreme point $U_{i}$ admits a decomposition into two disjoint regions $I_{i}, J_{i} \subset [0,1]$ such that $U_{i}$ coincides with $\overline{U}_{i}$ on $I_{i}$, while on $J_{i}$ the function $U_{i}$ is piecewise linear with at most countably many kinks. This implies that an optimal monotone threshold mechanism coincides with the efficient allocation for a subset $I_{i} \subset [0,1]$ of types, and that for the remaining types $J_{i}=[0,1]\setminus I_{i}$ the associated partition has at most countably many intervals. Their result also implies that, for any two adjacent intervals $S_{i}$ and $T_{i}$ in $\mathcal{S}_{i}$, the mechanism must be deterministic on one of the two intervals, that is, the volume-controlling parameter $\kappa \in [0,1]$ associated with one of them must be either $0$ or $1$. We formalize and prove this result in the Appendix.
    \qed
\end{remark}

\subsection{Two-threshold structures}\label{subsec: two-threshold}

Theorem \ref{thm: optimal mechanism} provides a general implication for how the allocation rate can be improved by designing information externalities. We now further elaborate on the optimal mechanism under additional assumptions on the distribution of private beliefs and on market size, respectively.

First, the following theorem shows that, under log-concavity, the optimal partition has a two-threshold structure, as illustrated in Figure \ref{fig: uniform} under uniform distribution.\footnote{The full-support assumption is imposed for simplicity. In the Appendix we prove the same statement in a more general environment where the support is symmetric around the prior; that is, $\overline{s}=1-\underline{s}$.} 

\begin{theorem}\label{thm: optimal mechanism log-concave}
    Assume $\text{supp }\mF=(0,1)$. 
    If $f$ is log-concave, then the optimal mechanism takes the form 
    \begin{align*}
        x_{i}(s_{i},s_{-i})
        = 
        \begin{cases}
            0 \quad & \text{if} \quad s_{i} \leq s^{\texttt{min}}_{i}(\tau), \\
            \mI\{\texttt{LR}(s_{-i})\geq \tau\} \quad & \text{if} \quad s^{\texttt{min}}_{i}(\tau) \leq s_{i} \leq s^{\texttt{max}}_{i}(\tau),\\ 
            \mI\{\texttt{LR}(s_{i},s_{-i})\geq 1\} \quad & \text{if} \quad s^{\texttt{max}}_{i}(\tau)\leq s_{i}, 
        \end{cases}
    \end{align*} 
    for some threshold $\tau\geq 0$, where $s^{\texttt{min}}_{i}(\tau)<1/2$ and $s^{\texttt{max}}_{i}(\tau)>1/2$ are decreasing functions of $\tau$ and uniquely defined by incentive compatibility. 
\end{theorem} 

The two thresholds are chosen to ensure incentive compatibility. Types in the upper region $s_{i}\geq s^{\texttt{max}}_{i}(\tau)$ have no incentive to misreport, since they receive the efficient allocation. The upper cutoff $s^{\texttt{max}}_{i}(\tau)$ is defined so that type $s^{\texttt{max}}_{i}(\tau)$ is indifferent between being assigned to the middle group and to the top group.\footnote{Specifically, we have $s^{\texttt{max}}_{i}(\tau)=1/(1+\tau)$, so that the mechanism is continuous at $s^{\texttt{max}}_{i}(\tau)$.} Consequently, no middle type wishes to deviate upward: under our mechanism, higher types are allocated the good with higher probability, and interim utility from allocation is increasing in type. The lower cutoff $s^{\texttt{min}}_{i}(\tau)$ is defined so that type $s^{\texttt{min}}_{i}(\tau)$ obtains exactly zero utility from the middle option. This makes the middle option individually rational for all middle types and rules out downward deviations to the bottom region. 

Many commonly used distributions are log-concave, and therefore admit a simple optimal mechanism with three regions. For example, the uniform distribution, (truncated) normal distributions with mean $1/2$, and beta distributions of the form $f(s_{i})=s_{i}^{\alpha-1}(1-s_{i})^{\alpha-1}/B$ with $\alpha\geq 1$ and a constant $B>0$ are all log-concave. See, for example, \citet{an1998logconcavity}, \citet{bagnoli2005log}, and \citet{zou2025log} for characterizations, properties, and further examples of log-concave densities.

\begin{figure}[ht]
    \centering
\begin{tikzpicture}[scale=0.4]
\fill[gray!30]
    (10,10) -- (5,0) -- (0,0)
    plot[domain=0:10, smooth, variable=\x]
        (\x, {10*(\x/10)^2})
    -- cycle;
\draw[->, thick] (0,0) -- (11,0) node[right,font=\scriptsize] {$s_{i}$};
\node at (10,0) [below,font=\scriptsize]{$1$};
\node at (0,0) [below,font=\scriptsize]{$0$};
\draw[->, thick] (0,0) -- (0,11);
\draw[dashed,ultra thick] (0,10) -- (10,10);
\node at (0,10) [left,font=\scriptsize]{$1$};
\draw[thick, smooth, domain=0:10, variable=\x]
    plot (\x, {10*(\x/10)^2});
\node at (10,10) [above,font=\scriptsize]{$U_{i}$};
\draw[thick] (5,0) -- (10,10);
\node at (5,0) [below,font=\scriptsize]{$\frac{1}{2}$};
\draw[thick] (0,0) -- (5,0);
\draw[dashed,ultra thick] (10,0) -- (10,10);
\node at (3.25,0) [below,font=\scriptsize]{$s^{\texttt{min}}_{i}$};
\node at (6.25,0) [below,font=\scriptsize]{$s^{\texttt{max}}_{i}$};
\draw[dashed,ultra thick] (3,0) -- (3,10);
\draw[dashed,ultra thick] (6,0) -- (6,10);
\draw[ultra thick] (0,0) -- (3,0);
\draw[ultra thick] (3,0) -- (6,3.6);
\draw[ultra thick, smooth, domain=6:10, variable=\x]
    plot (\x, {10*(\x/10)^2});
\end{tikzpicture}
\caption{}
    \label{fig:optimal indirect log concave}
  \begin{tablenotes}
       \footnotesize
       \item[] \textit{Notes:}
    Two-threshold mechanisms in Theorem \ref{thm: optimal mechanism log-concave} induce indirect utility functions that coincide with $0$ up to $s_{i}^{\texttt{min}}$, are linear over $[s_{i}^{\texttt{min}}, s_{i}^{\texttt{max}}]$, and then coincide with the first-best payoff $\overline{U}_{i}$ above $s_{i}^{\texttt{max}}$. 
    The thresholds $s_{i}^{\texttt{min}}$ and $s_{i}^{\texttt{max}}$ are chosen so that the indirect utility function is continuous and convex. 
  \end{tablenotes}
\end{figure}

We prove Theorem \ref{thm: optimal mechanism log-concave} by developing a weak duality for the primal problem. Our construction builds on \citet{kleiner2022optimal} and \citet{augias2025economics}. They provide duality results, for the problem of maximizing a linear functional over convex functions constrained to lie between two convex bounds. Then, under a log-concave density, we show we can construct a certificate for the optimality of an indirect utility function of the form shown in Figure \ref{fig:optimal indirect log concave}. The class of two-threshold mechanisms in Theorem \ref{thm: optimal mechanism log-concave} induces this shape. 

Second, we study the optimal mechanism in the large-market limit $n\to\infty$ and show that simple mechanisms are asymptotically optimal. Fix any agent $i$, and for each market size $n$, let $V_{n}\in [0,1]$ denote the maximum possible probability that the agent receives the good. Note that $V_{n}$ is increasing in $n$ and bounded above by $1$. Therefore, the limit exists, and we denote it by $V_{\infty}$. 

\begin{theorem}\label{prop: large market}
    Fix any agent $i$.
    There exist a sequence of mechanisms for agent $i$ of a form
    \begin{align*}
        x_{i}(s_{i},s_{-i};n) 
        = 
        \begin{cases}
            0 \quad & \text{if} \quad s_{i} \leq s^{\texttt{min}}_{i}(n) \\
            \kappa(n)\cdot \mI\{\lr{s_{-i}}\geq \tau(n)\} \quad & \text{if} \quad s^{\texttt{min}}_{i}(n) \leq s_{i}\leq s^{\texttt{max}}_{i}(n), \\
            1 \quad & \text{if} \quad s^{\texttt{max}}_{i}(n)\leq s_{i}, 
        \end{cases}
    \end{align*} 
    for each market size $n\in \mN$, such that the difference $|V_{n}-\mE[x_{i}(s;n)]|$ converges to zero in the limit $n\rightarrow \infty$. 
\end{theorem} 

The result shows that a simple class of mechanisms is asymptotically optimal. It consists of monotone threshold mechanisms with a simple two-threshold structure. These mechanisms exclude bottom types and assign a common allocation rule to middle types; unlike the mechanisms in Theorem \ref{thm: optimal mechanism log-concave}, however, they always allocate the good to the remaining types. As the number of agents grows, the designer can achieve a higher expected payoff from each agent because the mechanism has more inputs. Therefore, the large-market result also provides an upper bound on the designer's payoff. 

The proof uses the same machinery as that used in Theorems \ref{thm: optimal mechanism} and \ref{thm: optimal mechanism log-concave}. Recall from Subsection \ref{subsec: proof idea} that for each market size, the optimal indirect utility must maximize a linear functional over the set of convex functions $U_{i}$ that are pointwise above $\underline{U}_{i}(s_{i})=\max\{0,2s_{i}-1\}$ and below $\overline{U}_{i}$. Now, in the limit $n\rightarrow \infty$, the law of large numbers applies to the profile of the others' reports $s_{-i}$, which reveals the state $\omega\in \{-1,+1\}$. Therefore, in the limit, the first-best payoff for the agent is attained by the mechanism of allocating the good if and only if the state is $+1$. Hence, the upper bound 
\begin{align*}
    \overline{U}_{i}^{\infty}(s_{i}) = s_{i}\cdot 1 - (1-s_{i})\cdot 0 = s_{i}
\end{align*}
is now linear in interim beliefs. 

Therefore, the set of indirect utility functions over which we optimize in the limit $n\rightarrow \infty$ is now represented by the triangle region in Figure \ref{fig: large market extreme point}. Since the objective function is linear, Bauer's maximum principle shows that a solution is an extreme point $U_{i}$ of this set, which has a very simple form: $U_{i}$ coincides with $\underline{U}_{i}$ except on an interval $[s^{\texttt{min}}_{i},s^{\texttt{max}}_{i}]$, and $U_{i}$ is linear on that interval. This is an implication of Theorem 1 in \citet{augias2025economics}, which characterizes the extreme points of convex function intervals. See Remark \ref{rem: extreme point} for a related discussion. 

\begin{figure}[ht]
    \begin{subfigure}{0.48\textwidth}
    \centering
\begin{tikzpicture}[scale=0.4]
\fill[gray!30]
    (10,10) -- (5,0) -- (0,0)
    plot[domain=0:10, smooth, variable=\x]
        (\x, \x)
    -- cycle;
\draw[->, thick] (0,0) -- (11,0) node[right,font=\scriptsize] {$s_{i}$};
\node at (10,0) [below,font=\scriptsize]{$1$};
\node at (0,0) [below,font=\scriptsize]{$0$};
\draw[->, thick] (0,0) -- (0,11);
\draw[ultra thick, smooth, domain=0:10, variable=\x]
    plot (\x, \x);
\node at (10,10) [above,font=\scriptsize]{$\overline{U}_{i}^{\infty}$};
\draw[ultra thick] (5,0) -- (10,10);
\node at (5,0) [below,font=\scriptsize]{$\frac{1}{2}$};
\draw[ultra thick] (0,0) -- (5,0);
\draw[dashed,ultra thick] (10,0) -- (10,10);
\end{tikzpicture}
\caption{}
  \end{subfigure}
\hfill
\begin{subfigure}{0.48\textwidth}
    \centering
    \begin{tikzpicture}[scale=0.4]
\fill[gray!30]
    (10,10) -- (5,0) -- (0,0)
    plot[domain=0:10, smooth, variable=\x]
        (\x, \x)
    -- cycle;
\draw[->, thick] (0,0) -- (11,0) node[right,font=\scriptsize] {$s_{i}$};
\node at (10,0) [below,font=\scriptsize]{$1$};
\node at (0,0) [below,font=\scriptsize]{$0$};
\draw[->, thick] (0,0) -- (0,11);
\draw[thick, smooth, domain=0:10, variable=\x]
    plot (\x, \x);
\draw[thick] (5,0) -- (10,10);
\node at (5,0) [below,font=\scriptsize]{$\frac{1}{2}$};
\draw[ultra thick] (0,0) -- (2,0) -- (7,4) -- (10,10); 
\node at (10,10) [above,font=\scriptsize]{$U_{i}$};
\draw[dashed,ultra thick] (10,0) -- (10,10);
\end{tikzpicture}
\caption{}
\end{subfigure}
    \caption{}
  \begin{tablenotes}
       \footnotesize
       \item[] \textit{Notes:}
       (a) The space of indirect utility functions over which we optimize in the limit.
       (b) An extreme point of this triangle is a piecewise linear function with at most two kinks. Except for the middle region, it coincides with $\underline{U}_{i}$. 
  \end{tablenotes}
  \label{fig: large market extreme point}
\end{figure}

We then show that there exist sequences of parameters $\kappa(n)$ and $\tau(n)$ that asymptotically implement the line segment in the middle region $[s^{\texttt{min}}_{i},s^{\texttt{max}}_{i}]$. If the middle region is in the interior of $[0,1]$, we set $s^{\texttt{min}}_{i}(n)=s^{\texttt{min}}_{i}$ and $s^{\texttt{max}}_{i}(n)=s^{\texttt{max}}_{i}$ for all sufficiently large $n$. In this case, there exists $N$ such that for every market size $n\geq N$, the extreme point $U_{i}$ lies pointwise below $\overline{U}_{i}$; hence, $U_{i}$ is exactly optimal for every $n\geq N$, and $V_{n} = V_{\infty}$. Otherwise, we instead construct a sequence of intervals $[s^{\texttt{min}}_{i}(n),s^{\texttt{max}}_{i}(n)]$ whose endpoints converge to the target values $s^{\texttt{min}}_{i}$ and $s^{\texttt{max}}_{i}$ in the limit. 

\begin{remark}[Comparison with Theorem \ref{thm: optimal mechanism log-concave}]\label{rem: large-market limit log-concave and symmetric case}
    As the discussion above suggests, whether $s_{i}^{\texttt{max}}(n)$ converges to $1$ has important implications. In the Appendix, we show that if $f$ is log-concave and symmetric around the prior, then $s_{i}^{\texttt{max}}(n) \to 1$ as $n \to \infty$. This is consistent with Theorem \ref{thm: optimal mechanism log-concave}, which shows that for the top group, the optimal indirect utility coincides with its first-best payoff $\overline{U}_{i}$ for every market size. Note that in any finite market, the piecewise linear function with parameter $s_{i}^{\texttt{max}}=1$ lies above the first-best payoff $\overline{U}_{i}$ in a neighborhood of $s_{i}=1$ and is therefore infeasible. Consequently, both the designer and a high-type agent $i$ strictly benefit from having infinitely many market participants. 
    \qed 
\end{remark}

\subsection{The laissez-faire outcome and social-learning design}\label{subsec: laissez-faire} 

Among the mechanisms we study, a particularly interesting benchmark is the \textit{laissez-faire outcome}, i.e., $x_{i}(s)=\mI\{s_{i}\geq 1/2\}$ for each agent $i$: each agent receives the good if and only if her private belief exceeds $1/2$. This outcome can be implemented without eliciting any private information; it suffices to let each agent independently decide whether to accept the good. 

It would be worthwhile to note that the laissez-faire outcome is a special case of the class of mechanisms characterized in Theorem \ref{thm: optimal mechanism log-concave} and is obtained when $s^{\texttt{max}}=1$ and $s^{\texttt{min}}=1/2$. As Theorem \ref{thm: optimal mechanism log-concave} implies, however, the laissez-faire outcome is never optimal under log-concave densities. In particular, the optimal lower cutoff satisfies $s^{\texttt{min}}<1/2$, so the designer benefits from making agent $i$'s allocation depend on the information held by the other agents, which creates information externality. 

We note, however, that the weak-duality argument underlying Theorem \ref{thm: optimal mechanism log-concave} does not itself rely on the log-concavity of densities. We can therefore use the same argument to identify distributions for which the laissez-faire outcome is optimal. The next result provides such a condition. 

\begin{proposition}\label{prop: laissez-faire outcome}
    Assume $\text{supp }\mF = (0,1)$. If $\log f$ is weakly increasing on $[0,1/2]$, is weakly concave over $[1/2,1]$, and 
    \begin{align*}
        \int_{1/2}^{1}(1-t_{i})f(t_{i})dt_{i}\geq \frac{1}{4}f\left(\frac{1}{2}\right),
    \end{align*} 
    then, the laissez-faire outcome $x_{i}(s)=\mI\{s_{i}\geq 1/2\}$ is optimal.
\end{proposition}

At a high level, these conditions require that a relatively large probability mass be placed on private beliefs that are above but close to the prior. Figure \ref{fig: laissez-faire} shows one such distribution. Proposition \ref{prop: laissez-faire outcome} claims that the laissez-faire outcome is optimal among all mechanisms under this distribution. Note that the log-density is not globally concave over $(0,1)$. 

\begin{figure}[H]
    \centering
    \begin{subfigure}[t]{0.48\textwidth}
    \centering
\begin{tikzpicture}[x=5cm,y=0.8cm]

\def\aval{30}
\def\bval{0.04}
\def\kval{148}
\def\cval{0.093} 

\draw[->, thick] (0,0) -- (1.08,0) node[right,font=\scriptsize] {$s_i$};

\draw[->, thick] (0,0) -- (0,6.3) node[above,font=\scriptsize] {$f(s_i)$};

\node at (0,0) [below,font=\scriptsize] {$0$};
\node at (0.5,0) [below,font=\scriptsize] {$\frac{1}{2}$};
\node at (1,0) [below,font=\scriptsize] {$1$};

\draw[dashed, thick]
    (0.5,0) -- (0.5,{exp(\cval)});
\draw[dashed, thick] (1,0) -- (1,{exp(\aval*(1-0.5)-\kval*(1-0.5)^2+\cval)});

\draw[ultra thick] (0,{exp(-\aval*\bval/2+\cval)}) -- ({0.5-\bval},{exp(-\aval*\bval/2+\cval)});

\draw[ultra thick, smooth, domain=0.5-\bval:0.5, variable=\x, samples=8]
    plot (\x,{exp(\aval*(\x-0.5)+\aval*(\x-0.5)^2/(2*\bval)+\cval)});

\draw[ultra thick, smooth, domain=0.5:1, variable=\x, samples=25]
    plot (\x,{exp(\aval*(\x-0.5)-\kval*(\x-0.5)^2+\cval)});

\end{tikzpicture}
\caption{}
\label{fig:f_density}
\end{subfigure}
\hfill
\begin{subfigure}[t]{0.48\textwidth}
\centering
\begin{tikzpicture}[x=5cm,y=0.22cm]

\def\aval{30}
\def\bval{0.04}
\def\kval{148}
\def\cval{0.093} 

\draw[->, thick] (0,-20) -- (1.08,-20)
    node[right,font=\scriptsize] {$s_i$};

\draw[->, thick] (0,-20) -- (0,2.5)
    node[above,font=\scriptsize] {$\log f(s_i)$};

\draw[dashed, thick] (0,0) -- (1,0);

\node at (0,0) [left,font=\scriptsize] {$0$};

\node at (0,-20) [below,font=\scriptsize] {$0$};
\node at (0.5,-20) [below,font=\scriptsize] {$\frac{1}{2}$}; 

\draw[dashed, thick]
    (0.5,-20) -- (0.5,{\cval});
\draw[dashed, thick] (1,-20) -- (1,{\aval*(1-0.5)-\kval*(1-0.5)^2+\cval});

\draw[ultra thick]
    (0,{-\aval*\bval/2+\cval}) -- ({0.5-\bval},{-\aval*\bval/2+\cval});

\draw[ultra thick, smooth, domain=0.5-\bval:0.5, variable=\x, samples=6]
    plot (\x,{\aval*(\x-0.5)+\aval*(\x-0.5)^2/(2*\bval)+\cval});

\draw[ultra thick, smooth, domain=0.5:1, variable=\x, samples=12]
    plot (\x,{\aval*(\x-0.5)-\kval*(\x-0.5)^2+\cval});

\end{tikzpicture}
\caption{}
\label{fig:logf_density}
\end{subfigure}
\caption{}
\label{fig: laissez-faire}
\begin{tablenotes}
    \footnotesize
    \item[] \textit{Notes:}
    A distribution of private beliefs under which the laissez-faire
    outcome is optimal. Note that $\log f$ is increasing over
    $[0,1/2]$ and concave over $[1/2,1]$. However, it is not globally concave over $[0,1]$. 
    Specifically, the figure shows a parametrized distribution given by $f(s_{i})\propto \exp{\{\varphi(s_{i}-1/2)}\}$, where $\varphi(x_{i})$ equals $-(\alpha\beta)/2$ over $[-1/2,-\beta]$, $\alpha x_{i}+\alpha x_{i}^{2}/(2\beta)$ over $[-\beta, 0]$, and then $\alpha x_{i}-\kappa x_{i}^{2}$ over the remaining region. $\kappa$ is chosen so that $f$ has mean $1/2$. In the figure, we take $\alpha=30$ and $\beta=0.04$. 
\end{tablenotes}
\end{figure}
Compared to Theorem \ref{thm: optimal mechanism log-concave}, the sufficient condition does not require $f$ to be log-concave over the entire domain $(0,1)$. It is worthwhile to note that the same argument can be proved under the alternative assumption that $f$ is log-concave globally over $(0,1)$ and the inequality condition in the proposition. However, as we can infer from Theorem \ref{thm: optimal mechanism log-concave}, no such distribution of private beliefs exists: these two requirements are not consistent with the mean-preserving constraint $\mE[s_{i}]=1/2$. In the proof of Theorem \ref{thm: optimal mechanism log-concave}, we prove and use this observation.   

This result is of particular interest because it has a less obvious application to a \textit{social learning problem} \`a la \citet{banerjee1992simple}, \citet{bikhchandani1992theory}, and \citet{acemoglu2011bayesian}, in which agents make decisions sequentially. Although our model concerns static allocation and may appear unrelated at first glance, it can be used to study a design problem in this dynamic environment. Consider a social learning environment with an \textit{observation structure} $\mathcal{B}=(B_{i})_{i=1}^{n}$, where $B_{i}\subseteq\{1,\dots,i-1\}$ is the set of predecessors whose actions agent $i$ observes. 

For each given network structure $\mathcal{B}=(B_{i})_{i=1}^{n}$, the social learning game proceeds as follows. At date $0$, the state $\omega\in\{-1,+1\}$ is realized with equal probability and is not observed by any agent. At each date $i\in \{1,\dots,n\}$, agent $i$ observes her private belief $s_{i}$ and the actions taken by agents in $B_{i}$, and then decides whether to accept the good. An agent who accepts obtains payoff $\omega$, whereas an agent who rejects obtains zero. We say that an allocation rule $x:[0,1]^{n}\to [0,1]^{n}$ is a \textit{social-learning outcome} under $\mathcal{B}$ if there exists a Bayesian equilibrium of this game under which $x_{i}(s)$ equals the probability that agent $i$ accepts the good for every agent $i$ and signal profile $s$. 

Consider now \textit{social-learning design}, in which the designer chooses the observation structure $\mathcal{B}$ to maximize the expected number of agents who accept the good. When viewed as a direct allocation mechanism, we can show that any social-learning outcome satisfies both \eqref{eq: IC} and \eqref{eq: P}. Therefore, the optimal value of our mechanism-design problem provides an upper bound on what can be achieved by designing the observation structure. When the laissez-faire outcome is optimal among all mechanisms, this upper bound is attained by the empty observation structure $\mathcal{B}^{\varnothing}$, under which no agent observes any previous action. 

\begin{corollary}\label{cor: social learning}
    Assume $\mF$ satisfies the conditions in Proposition \ref{prop: laissez-faire outcome}. Then, among all network structures $\mathcal{B}$, the empty observation structure $\mathcal{B}^{\varnothing}$ maximizes the expected number of agents who accept the good, which induces the laissez-faire outcome. 
\end{corollary} 

The literature on social learning mainly asks which network and information structures lead to efficient learning. For example, \citet{acemoglu2011bayesian}, \citet{lobel2015information}, and \citet{kartik2024beyond} provide necessary conditions for the agents' beliefs to converge to the correct state in the limit $n\to\infty$, and also give a class of simple observation structures that indeed induces the correct learning. More recent work such as \citet{arieli2023herd} considers a design problem; that is, they take the underlying observation structure as given and study how the private information disclosed to agents can be designed to induce desirable behavior. Our result considers the complementary problem: the agents' private information structure is fixed, while the observation structure is the object of design. Although observation networks form a discrete and combinatorial design space and are therefore difficult to analyze directly, our analysis suggests that this difficulty can sometimes be overcome by embedding the network-design problem in a broader mechanism-design problem. The latter provides an upper bound that may be attained by a simple observation structure under certain conditions over the distribution of private beliefs. 

\section{Extension: payment design}\label{sec: discussion}

Partly motivated by applications such as vaccine allocation and lung transplantation, we study a mechanism design problem in which the designer cannot use monetary transfers. Even in those applications, however, the designer may still be able to use payment-like instruments, either by directly offering a non-linear pricing scheme or by imposing implicit, wasteful costs such as waiting time or other burdens. Here, we explore implications of designing such payments. 

Consider the problem of maximizing the same objective function by choosing both allocation and payment $(x_{i},t_{i}):[0,1]^{n}\rightarrow [0,1]\times [0,T]$ for each agent $i$. Here, $T\geq 1$ is a finite number so as to ensure that an optimal mechanism exists.
Participation and incentive constraints are now rephrased as follows.
\begin{align*}
    & \mE[\omega\cdot x_{i}(s_{i},s_{-i})-t_{i}(s_{i},s_{-i})\mid s_{i}] \geq 0, \label{eq: P with money} \tag{$\text{P}^{*}$} \\
    & \mE[\omega\cdot x_{i}(s_{i},s_{-i})-t_{i}(s_{i},s_{-i})\mid s_{i}] \geq \mE[\omega\cdot x_{i}(\hat{s}_{i},s_{-i})-t_{i}(\hat{s}_{i},s_{-i})\mid s_{i}], \label{eq: IC with money} \tag{$\text{IC}^{*}$}
\end{align*} 
for all $s_{i}$ and $\hat{s}_{i}$, where we assume agents' preferences are quasi-linear in payments.
Note that we also assume $t_{i}(s)\geq 0$ for all agents $i$ and signal profiles $s$, i.e, the designer cannot subsidize. 
At the end of this subsection, we show that the designer can always allocate the good when the designer can also use negative payments. 
A mechanism $(x_{i},t_{i})$ for agent $i$ is \textit{feasible} if it satisfies \eqref{eq: P with money} and \eqref{eq: IC with money}. 

\begin{remark}[Revenue]\label{rem: revenue and nonnegative payments}
    Alternatively, depending on the application, it may be natural to interpret nonnegative payments $t_{i}$ as revenue for the designer and consider $\mE[\sum_{i}t_{i}(s)]$ as the designer's objective. Our approach can also be used for this alternative objective function to establish the same results as those in Theorem \ref{thm: monotone threshold and payment} and Proposition \ref{prop: positive payment}. In particular, this implies that the optimal mechanism does not involve exclusion at the bottom, in sharp contrast to standard auction settings with independent private values. 
\end{remark}

Allowing for payment design further complicates the analysis. The reason is that the revenue equivalence theorem, which underlies standard mechanism design, fails in our environment. In particular, for a given allocation rule, there generally exist multiple payment rules that render it incentive compatible. Relatedly, objective function can no longer be represented solely in terms of the indirect utility function, so the approach used in the previous section cannot be applied as is in this extension. Therefore, even conditional on an allocation rule, the payment rule must be designed carefully. 

The first result of this subsection shows that the main theorem of this study is robust to introducing payment design. 
Abusing terminology, we call a mechanism $(x_{i},t_{i})$ a \textit{monotone threshold mechanism} if (i) there exist a partition $\mathcal{S}_{i}$ of the signal space $[\underline{s},\overline{s}]$ such that for each interval $S_{i}\in \mathcal{S}_{i}$, there exist thresholds $(\kappa_{x},\tau_{x})\in \mR^{2}$ and $(\kappa_{t},\tau_{t})\in \mR^{2}$ such that
\begin{align*} 
    x_{i}(s_{i},s_{-i}) &= \kappa_{x} \cdot \mI\{\lr{s_{-i}}\geq \tau_{x}\}, \\
    t_{i}(s_{i},s_{-i}) &= \kappa_{t} \cdot \mI\{\lr{s_{-i}}\geq \tau_{t}\},
\end{align*}
for all $s_{i}\in S_{i}$ and $s_{-i}\in [\underline{s},\overline{s}]^{n-1}$, and (ii) the coarsest such partition consists of intervals. That is, both allocation and payment have a monotone threshold structure. 

\begin{theorem}\label{thm: monotone threshold and payment} 
    Suppose that the designer can also design non-negative payments. 
    Then, there exists a monotone threshold mechanism $(x_{i},t_{i})$ that is optimal. 
\end{theorem} 

The logic builds on that of Theorem \ref{thm: optimal mechanism}. The first key observation is that any positive payment can only reduce expected payoffs. Together with the participation constraint, this implies that any optimal indirect utility function for each agent $i$ must belong to the same set $\mathcal{U}_{i}^{*}$ as the one constructed in the proof of Theorem \ref{thm: optimal mechanism}. We use this observation to show that, for any optimal mechanism $(x_{i},t_{i})$, its allocation rule $x_{i}$ can be replaced with a monotone threshold allocation while keeping the same payment rule $t_{i}$. The second key lemma then shows that, if $t_{i}$ cannot be replaced with a monotone threshold payment without changing the objective function, the mechanism must involve an excessively high payment when the state is low. We then argue that, by perturbing the allocation and payment jointly in a specific way, one can strictly increase the allocation probability, yielding a contradiction. 

In Remark \ref{rem: objective function}, we discussed that Theorem \ref{thm: optimal mechanism} continues to hold even when the designer’s objective function is not the expected allocation probability. The same logic also applies to Theorem \ref{thm: monotone threshold and payment}. Specifically, for any objective function such that, for each type $s_{i} \geq 1/2$, replacing any option $(x_{i}(s_{i},s_{-i}), t_{i}(s_{i},s_{-i}))$ with $(1,0)$ increases the designer’s objective value, there exists an optimal monotone threshold mechanism $(x,t)$. For example, as discussed in Remark \ref{rem: objective function}, one such objective function is a convex combination of the allocation probability and the utilitarian sum.

The designer in our setting cares only about the allocation probability, and therefore, any positive payments are socially wasteful. At the same time, payments always tighten the participation constraint. It may therefore be tempting to conjecture that the optimal payment rule imposes no payments at all. 
Indeed, when there is only one agent on the market, it is easy to see that this conjecture is correct.\footnote{The reason is that the agent's interim expected payoff is now written as $(2s_{i}-1)x_{i}(s_{i})-t_{i}(s_{i})$, where $t_{i}(s_{i})$ is any nonnegative payment. Then, the participation constraint requires that the designer can allocate the good with probability at most $1-\mF(1/2)$. This can be achieved by a non-payment mechanism $x_{i}(s_{i})=\mI\{s_{i}\geq 1/2\}$.} 
However, more generally, when there are multiple agents, positive payments are strictly optimal. 

\begin{proposition}\label{prop: positive payment}
    Suppose that the designer can also design non-negative payments. Assume $n\geq 2$ and $\text{supp }\mF=(0,1)$. Then, any optimal monotone threshold mechanism does not exclude any positive mass of types. Consequently, it is strictly optimal to impose strictly positive payments. 
\end{proposition} 

The intuition is as follows. 
By Proposition \ref{prop: exclusion at the bottom}, under an optimal mechanism when payments are unavailable, there exists a set of low types that are excluded from the market. Now consider adding a new option to the original menu, such that the good is allocated only when the other agents' private beliefs are sufficiently high, and a small positive payment is charged in that event. Because the excluded low types obtain zero utility from the original menu, if the payment is sufficiently small, some of them will choose this new option, thereby increasing the total allocation probability. At the same time, the new option may also attract some higher types away from their original options. The proof shows, however, that the resulting reduction in the allocation probability is smaller. A similar logic suggests that, whenever a set of types is excluded and receives payoff zero, we can increase the allocation probability by an analogous manipulation of the original mechanism.

The second result shows that, as in Theorem \ref{prop: large market}, a family of simple mechanisms are asymptotically optimal. 
Fixing any agent $i$, let $V_{n}^{*}\in [0,1]$ be the maximum possible probability of allocation to the agent, when the designer can also design payments.

\begin{theorem}\label{thm: large market with payment}
    Suppose that the designer can also design non-negative payments. 
    Fix any $i$.    
    Then, there exist a sequence of mechanisms for agent $i$ of a form
    \begin{align*}
        x_{i}(s_{i},s_{-i};n) 
        & = 
        \begin{cases}
            \mI\{\lr{s_{i},s_{-i}}\geq \eta(n)\} \quad & \text{if} \quad s_{i} \leq s^{\texttt{min}}_{i}(n) \\
            \kappa_{x}+(1-\kappa_{x})\cdot \mI\{\lr{s_{i},s_{-i}}\geq \tau(n)\} \quad & \text{if} \quad s^{\texttt{min}}_{i}(n) \leq s_{i}\leq s^{\texttt{max}}_{i}(n), \\
            1 \quad & \text{if} \quad s^{\texttt{max}}_{i}(n)\leq s_{i},
        \end{cases} \\
        t_{i}(s_{i},s_{-i};n) 
        & = 
        \begin{cases}
            \delta(n)\cdot \mI\{\lr{s_{i},s_{-i}}\geq \eta(n)\} \quad & \text{if} \quad s_{i} \leq s^{\texttt{min}}_{i}(n) \\
            \kappa_{t}(n)\cdot \mI\{\lr{s_{i},s_{-i}}\geq \tau(n)\} \quad & \text{if} \quad s^{\texttt{min}}_{i}(n) \leq s_{i}\leq s^{\texttt{max}}_{i}(n), \\
            0 \quad & \text{if} \quad s^{\texttt{max}}_{i}(n)\leq s_{i}.
        \end{cases}
    \end{align*} 
    for each market size $n\in \mN$, such that $|V^{*}_{n}-\mE[x_{i}(s;n)]|$ converges to zero in the limit $n\rightarrow \infty$, where $\delta(n)\rightarrow 1$ as $n\rightarrow \infty$, and $\kappa_{x}\in [0,1]$ and $\kappa_{t}(n)\in [0,1]$. 
\end{theorem} 

The proof proceeds in two steps. First, as in the proof of Theorem \ref{prop: large market}, we consider the limit market as $n \to \infty$. We show that, even when payments are allowed, the problem can ultimately be formulated as one of maximizing a linear functional of indirect utility subject to the same triangular constraint. We use this structure to derive an optimal mechanism in the limit market. 

Second, we construct a sequence of mechanisms that converges to this limit-market mechanism. This part, however, is different from the proof of Theorem \ref{prop: large market}. As discussed above, in finite markets the indirect utility function alone does not pin down the objective function, so it is necessary to construct the sequence of mechanisms directly. As a result, the family of mechanisms in Theorem \ref{thm: large market with payment} takes a somewhat more complex form. For example, whereas the family of mechanisms considered in Theorem \ref{prop: large market} induces piecewise linear utility functions, the mechanisms constructed here induce strictly convex indirect utility functions. 

Figure \ref{fig: limit market with payment} provides a high-level illustration of the mechanisms in Theorem \ref{thm: large market with payment} near the limit as $n \to \infty$. The convex boundary represents an iso-likelihood-ratio curve defined by $\lr{s_{i}}\cdot \lr{s_{-i}}=1$. The mechanism partitions the type space into three intervals. In the bottom interval, the good is allocated if and only if $\lr{s_{i}}\cdot \lr{s_{-i}}\geq 1$, in which case the agent makes a payment of $1$. In the middle interval, the good is also allocated with an interior probability even when $\lr{s_i}\cdot \lr{s_{-i}}<1$, while the payment is lower. In the top interval, the good is always allocated and no payment is required. 

\begin{figure}[ht]
    \begin{subfigure}{0.48\textwidth}
    \centering
\begin{tikzpicture}[scale=0.43]
\fill[black, pattern = north east lines]
     (7.5,10) -- (3.75,10) -- (3.75,0) -- (7.5,0)
    -- cycle;
\fill[gray!60]
    (10,10) -- (7.5,10) -- (7.5,0) -- (10,0)
    -- cycle;
\fill[gray!60]
    (10,10) -- (10,0)
    .. controls (4,0) and (0,2.5) .. (0,10)
    -- cycle;
\draw[->, thick] (0,0) -- (11,0) node[right,font=\scriptsize] {$\lr{s_{i}}$};
\node at (0,0) [left,font=\scriptsize]{$0$};
\draw[->, thick] (0,0) -- (0,11);
\draw[->, thick] (0,0) -- (0,11) node[left,font=\scriptsize] {$\lr{s_{-i}}$};
\draw[dashed,ultra thick] (0.1,10) .. controls (0,2.5) and (4,0) .. (10,0.1);
\draw[dashed,ultra thick] (3.75,0) -- (3.75,10);
\draw[dashed,ultra thick] (7.5,0) -- (7.5,10);
\draw[dashed,ultra thick] (10,0) -- (10,10);
\draw[dashed,ultra thick] (0,10) -- (10,10);
\end{tikzpicture}

\begin{tikzpicture}[scale=0.6]
\draw (0,0) rectangle (0.3,0.3);
\node[right,font=\scriptsize] at (0.3,0.15) {$x_i(s)=0$};

\fill[pattern=north east lines] (3.0,0) rectangle (3.3,0.3);
\draw (3.0,0) rectangle (3.3,0.3);
\node[right,font=\scriptsize] at (3.3,0.15) {$x_i(s)\in(0,1)$};

\fill[gray!60] (6.6,0) rectangle (6.9,0.3);
\draw (6.6,0) rectangle (6.9,0.3);
\node[right,font=\scriptsize] at (6.9,0.15) {$x_i(s)=1$};
\end{tikzpicture}

\caption{}
  \end{subfigure}
\begin{subfigure}{0.48\textwidth}
    \centering
    \begin{tikzpicture}[scale=0.43]
\begin{scope}
    \clip (3.75,0) rectangle (7.5,10);

    \fill[black, pattern = north east lines]
        (0,10)
        .. controls (0,2.5) and (4,0) .. (10,0)
        -- (10,10)
        -- cycle;
\end{scope}
\begin{scope}
    \clip (0,0) rectangle (3.75,10);

    \fill[gray!60]
        (0,10)
        .. controls (0,2.5) and (4,0) .. (10,0)
        -- (10,10)
        -- cycle;
\end{scope}
\draw[->, thick] (0,0) -- (11,0) node[right,font=\scriptsize] {$\lr{s_{i}}$};
\node at (0,0) [left,font=\scriptsize]{$0$};
\draw[->, thick] (0,0) -- (0,11);
\draw[->, thick] (0,0) -- (0,11) node[left,font=\scriptsize] {$\lr{s_{-i}}$};
\draw[dashed,ultra thick] (0.1,10) .. controls (0,2.5) and (4,0) .. (10,0.1);
\draw[dashed,ultra thick] (3.75,0) -- (3.75,10);
\draw[dashed,ultra thick] (7.5,0) -- (7.5,10);
\draw[dashed,ultra thick] (10,0) -- (10,10);
\draw[dashed,ultra thick] (0,10) -- (10,10);
\end{tikzpicture}

\begin{tikzpicture}[scale=0.6]
\draw (0,0) rectangle (0.3,0.3);
\node[right,font=\scriptsize] at (0.3,0.15) {$t_i(s)=0$};

\fill[pattern=north east lines] (3.0,0) rectangle (3.3,0.3);
\draw (3.0,0) rectangle (3.3,0.3);
\node[right,font=\scriptsize] at (3.3,0.15) {$t_i(s)\in(0,1)$};

\fill[gray!60] (6.6,0) rectangle (6.9,0.3);
\draw (6.6,0) rectangle (6.9,0.3);
\node[right,font=\scriptsize] at (6.9,0.15) {$t_i(s)=1$};
\end{tikzpicture}

\caption{}
\end{subfigure}
   \caption{}
    \label{fig: limit market with payment}
  \begin{tablenotes}
       \footnotesize
       \item[] \textit{Notes:}
       Asymptotically optimal mechanisms for a fixed agent $i$ near the limit. Since we consider large markets, the space of private beliefs is compressed and represented in a two-dimensional space using $\lr{s_{i}}$ and $\lr{s_{-i}}$. These are sufficient statistics for describing the mechanisms. 
  \end{tablenotes}
\end{figure}

Theorem \ref{thm: large market with payment} has several implications for payment design. First, the designer may strictly prefer to use payments. The intuition is that, although payments tighten the participation constraint, they also provide an additional instrument for shaping incentives. The relevant force here, however, differs from that in standard auction problems with private values. What plays the essential role is not simply that incentives can be maintained by assigning different payment levels to different options. More importantly, participants with different beliefs perceive different expected payments from the same payment function. 

Specifically, in this environment, payments are used to deter \textit{downward manipulations}. The reason is that, under a monotone threshold mechanism, for any given option on the menu, higher types perceive the expected payment associated with that option to be larger. Hence, a monotone threshold payment rule makes downward misreporting less attractive. 
Consistent with this intuition, the optimal payment rule is typically nonincreasing in type. 
This again highlights that the role of payments in our environment is very different from their role in standard mechanism design. In the standard environment, the envelope theorem and the revenue equivalence theorem imply that, up to an additive constant, the transfer rule is uniquely pinned down by the integral of a monotone allocation rule. That transfer rule is increasing in type. 

Second, the mechanisms no longer exclude bottom types, unlike in Proposition \ref{prop: exclusion at the bottom}. This observation also suggests that the exclusion-at-the-bottom result obtained without payment design is nontrivial. The intuition, again, can be understood in terms of managing incentives for manipulating information externalities. Without payments, the optimal mechanism excludes bottom types in order to deter downward manipulations by middle-belief types, who may otherwise misreport so as to obtain the good only when its value is high. With payments, however, such incentives can be eliminated without excluding low-belief types, by imposing positive payments. 

We conclude this section with one observation. We assume non-negative payments because this is more consistent with our intended interpretation, namely that payments represent socially wasteful costs such as waiting time. If the designer can use arbitrary pricing schemes, including negative payments, then the first-best outcome is achievable. 

\begin{proposition}\label{prop: full surplus extraction}
    If the designer can also set negative payments, there exists a feasible mechanism that is ex-ante budget balanced and always allocates the good, i.e., $\mE[t_{i}(s)]=0$ and $x_{i}(s)=1$ for all $i$ and $s$. 
\end{proposition}

Since the agents' types are ex-ante correlated, the literature on full surplus extraction suggests that, if agents' beliefs about the other agents' types are linearly independent across types, any allocation rule is Bayesian implementable \citep{cremer1985optimal,cremer1988full,mcafee1992correlated,lopomo2022detectability}. 
These results do not apply here, because beliefs are linearly dependent: if $\beta_{i}(s_{i})$ is the belief of type $s_{i}$ about the other agents' types, we have $\beta_{i}(s_{i})=s_{i}\beta_{i}(1)+(1-s_{i})\beta_{i}(0)$. 
In the proof, we directly construct a budget-balanced payment scheme that implements the first-best allocation $x_{i}(s)=1$ for all $s$. 

\section{Conclusion} \label{sec: conclusion} 

We develop a simple model of allocating common-value goods. The allocation mechanism screens each participant's private information using the information of others. The optimal mechanism allocates the good in such a way that participants assign a higher valuation to it upon receiving it, and may withhold allocation even when allocation would be socially desirable. Although positive payments may at first appear only to tighten participation constraints and reduce the allocation probability, they provide a qualitatively different channel for deterring downward misreporting and are therefore generally optimal. 

As emphasized in the Introduction, we view our main contribution as providing a conceptual framework for designing informational externalities. Nevertheless, to illustrate the practical implications of our optimal mechanism, consider vaccine distribution. A government seeking to maximize vaccination rates because of positive externalities could establish multiple vaccination centers, each with a doctor who applies a different criterion for recommending vaccination based on evidence collected from past participants. Individuals uncertain about the severity of side effects may then sort across centers according to their private perceptions about the vaccine's safety. The resulting equilibrium allocation might be viewed as a first-order approximation to our optimal mechanism. Studying how informational externalities can be used as instruments for allocation design in specific applications, through models tailored to their institutional details, is a promising direction for future research.

We have assumed that there are sufficiently many homogeneous goods to allocate. Of course, resources may be often scarce in practice. However, our results extend easily to the case in which the resource constraint is imposed only on average supply, that is, when it takes the form $\mE[\sum_{i}x_{i}(s)]\leq q$ for some capacity $q\in \mR$. Specifically, if the solution to the original problem does not violate this constraint, then it remains optimal; if it does violate the constraint, then it is optimal to uniformly scale down the volume-controlling parameter $\kappa$ so that the sum of allocation probabilities is exactly equal to $q$. By contrast, when the resource constraint is imposed on ex post supply, that is, when it takes the form $\sum_{i}x_{i}(s)\leq q$ for each $s$, the problem becomes substantially more complicated. In particular, our current approach of working with a single indirect utility function is no longer valid, and we must instead work with a \textit{profile} of indirect utility functions. One fundamental difficulty arising in this extension is that we must also deal with \textit{winner's curse}. For example, if the mechanism allocates the good to the agent with the highest private belief, then upon receiving the good, that agent infers that every other agent had a lower private belief. Therefore, the expected payoff conditional on receiving the good is much lower than the payoff conditional only on her private signal. In order to raise this conditional belief, the optimal mechanism may need to allocate in a non-monotonic manner. A deeper analysis of this issue would be an interesting direction for future research.

\titleformat{\section}
		{\Large\bfseries}     
         {Appendix \thesection:}
        {0.5em}
        {}
        []

\renewcommand{\thetheorem}{A.\arabic{theorem}}
\setcounter{theorem}{0}

 \appendix 


\section{Proof of Theorem \ref{thm: optimal mechanism}} \label{sec: appendixa} 

Here, we prove Theorem \ref{thm: optimal mechanism}. 
As discussed in Remark \ref{rem: extreme point}, we can establish a slightly stronger result. 
We prove the following theorem, which provides a slight generalization of Theorem \ref{thm: optimal mechanism}. 

\begin{theorem}\label{thm: optimal mechanism countable partition}
    The optimal mechanism for each agent $i$ is a monotone threshold $x_{i}$ with a partition $\mathcal{S}_{i}$ such that there exists a countable subset $\mathcal{T}_{i}\subset \mathcal{S}_{i}$ of intervals with the following properties:
    \begin{itemize}
        \item $x_{i}(s_{i},s_{-i})=\mI\{\lr{s_{i},s_{-i}}\geq 1\}$ for $s_{i}\in S_{i}\notin \mathcal{T}_{i}$, 
        \item If $S_{i},T_{i}\in \mathcal{T}_{i}$ are adjacent intervals, then $x_{i}(s_{i},s_{-i})=\kappa\cdot \mI\{\lr{s_{i},s_{-i}}\geq \tau\}$ with $\kappa\in \{0,1\}$ over one of these two intervals. 
    \end{itemize} 
    I.e., for any two consecutive intervals, the mechanism is deterministic over one of them. 
\end{theorem}

For the remainder, we prove the generalized version Theorem \ref{thm: optimal mechanism countable partition}. 
For a given mechanism $x:[\underline{s},\overline{s}]^{n}\rightarrow [0,1]^{n}$, recall 
\begin{align*}
    U_{i}(s_{i}) &= \max_{t_{i}\in \text{supp }\mF} \mE[\omega\cdot x_{i}(t_{i},s_{-i})\mid s_{i}] \label{eq: U}\tag{U} 
\end{align*}
is the indirect utility function to the agent. 
If the mechanism $x$ is incentive compatible, we have $U_{i}(s_{i})=\mE[\omega\cdot x_{i}(s_{i},s_{-i})\mid s_{i}]$. 

For convenience, we often write $\mF^{\otimes}(s_{-i})=\prod_{k\neq i}\mF(s_{k})$ to be the product measure of the marginals of beliefs. 
Note that this is not the unconditional joint distribution of private beliefs: although signals are conditionally independent conditional on the state, they are correlated through the common state. 

Let us first show Lemma \ref{lem: feasible mechanism} below. 
Here, the normalization $s_{i}=\mP[\omega=+1\mid s_{i}]$ is key, which renders the interim payoff $\mE[\omega\cdot x_{i}(\hat{s}_{i},s_{-i})\mid s_{i}]$ linear in the belief $s_{i}$.
An envelope argument then characterizes the marginal interim payoff under truth-telling in terms of the interim allocation.
A few algebraic rearrangements yield \eqref{eq: EV}. 
Monotonicity \eqref{eq: M} follows by the standard argument. 

Some of the preliminary lemmas in the proof of this section will also be useful in the proofs of the other results. 

\begin{lemma}\label{lem: feasible mechanism}
    Let $\text{supp }\mF=[\underline{s},\overline{s}]$. 
    A mechanism $x$ is feasible if and only if 
    \begin{align*} 
        X_{i}(s_{i}) &= \int_{s_{-i}} \left\{\prod_{k\neq i}s_{k}+\prod_{k\neq i}(1-s_{k})\right\}x_{i}(s) d\mF^{\otimes}(s_{-i}) \text{ is increasing in } s_{i} \label{eq: M}\tag{M} \\
        \int_{\underline{s}}^{s_{i}} X_{i}(t_{i})dt_{i} &= \int_{s_{-i}} \left\{\prod_{k}s_{k}-\prod_{k}(1-s_{k})\right\}x_{i}(s_{i},s_{-i}) d\mF^{\otimes}(s_{-i}) - 2^{1-n}U_{i}(\underline{s}), \label{eq: EV}\tag{EV} 
    \end{align*} 
    for each agent $i$ and private belief $s_{i}$ in the support of $\mF$. 
    The right-hand side of equation \eqref{eq: EV} equals $2^{1-n}[U_{i}(s_{i})-U_{i}(\underline{s})]$. 
\end{lemma} 

\begin{proof}[Proof of Lemma \ref{lem: feasible mechanism}]
    We first rewrite the agents' payoff function.
    Note that the interim expected payoff to the agent with type $s_{i}$ from reporting type $t_{i}$ can be rewritten as 
    \begin{align*}
        U_{i}(t_{i};s_{i})
        &= s_{i} \cdot \mE[x_{i}(t_{i},s_{-i})\mid \omega=+1] - (1-s_{i}) \cdot \mE[x_{i}(t_{i},s_{-i})\mid \omega=-1] \\
        &= s_{i}\sum_{\omega}\mE[x_{i}(t_{i},s_{-i})\mid \omega] - \mE[x_{i}(t_{i},s_{-i})\mid \omega=-1], 
    \end{align*} 
    where the first equation follows by normalization. 
    Here, notice that 
    \begin{align*}
        \mE[x_{i}(t_{i},s_{-i})\mid \omega] 
        &= \int_{s_{-i}} x_{i}(t_{i},s_{-i})d\mF_{\omega}(s_{-i}) \\
        &= \int_{s_{-i}} \left\{\prod_{k\neq i}\frac{f_{\omega}(s_{k})}{f(s_{k})}\right\} x_{i}(t_{i},s_{-i})d\mF^{\otimes}(s_{-i}), 
    \end{align*} 
    where the second equation follows from $\mF=[\mF_{-1}+\mF_{+1}]/2$. 
    Moreover, note that by the normalization we have $f_{+1}(s_{k})/f(s_{k})=2s_{k}$ and $f_{-1}(s_{k})/f(s_{k})=2(1-s_{k})$. 
    Therefore, we eventually have
    \begin{align*}\label{eq: U star}\tag{U2}
        U(t_{i};s_{i})
        = 2^{n-1}\left\{s_{i}X_{i}(t_{i}) - \int_{s_{-i}} x_{i}(t_{i},s_{-i})\prod_{k\neq i}(1-s_{k})d\mF^{\otimes}(s_{-i})\right\}.
    \end{align*} 
    Now, we are ready to prove the statement. 
    Although the argument mirrors those in the classic mechanism design, we provide a formal proof for completeness. 

    First, assume that a mechanism $x$ satisfies \eqref{eq: P} and \eqref{eq: IC}. 
    The envelope theorem applied to the expression \eqref{eq: U} implies 
    \begin{align*}
        U_{i}(s_{i}) - U_{i}(\underline{s}) = 2^{n-1} \int_{\underline{s}}^{s_{i}}X_{i}(t_{i})dt_{i}. 
    \end{align*}
    We therefore obtain
    \begin{align*}
        \int_{\underline{s}}^{s_{i}}X_{i}(t_{i})dt_{i}
        &= s_{i}X_{i}(s_{i}) - \int_{s_{-i}} x_{i}(s_{i},s_{-i})\prod_{k\neq i}(1-s_{k})d\mF^{\otimes}(s_{-i}) - 2^{1-n}\cdot U_{i}(\underline{s}) \\
        &= \int_{s_{-i}} x_{i}(s_{i},s_{-i})\left\{\prod_{k}s_{k}-\prod_{k}(1-s_{k})\right\}d\mF^{\otimes}(s_{-i}) - 2^{1-n}\cdot U_{i}(\underline{s}),
    \end{align*}
    hence \eqref{eq: EV}. 
    Note that \eqref{eq: P} requires $2^{1-n}\cdot U_{i}(\underline{s})\geq 0$. 
    The condition \eqref{eq: IC} also implies $U(s_{i})\geq U(t_{i};s_{i})$ and $U(t_{i})\geq U(s_{i};t_{i})$. 
    Using the expression \eqref{eq: U star}, summing these inequalities yield $(s_{i}-t_{i})(X_{i}(s_{i})-X_{i}(t_{i}))\geq 0$, hence \eqref{eq: M}. 

    Second, assume a mechanism $x$ satisfies both \eqref{eq: M} and \eqref{eq: EV}. 
    Note that \eqref{eq: U star} and \eqref{eq: EV} imply
    \begin{align*}
        U_{i}(t_{i};s_{i}) 
        &= 2^{n-1}\left\{s_{i}X_{i}(t_{i}) - \int_{s_{-i}} x_{i}(t_{i},s_{-i})\prod_{k\neq i}(1-s_{k})d\mF^{\otimes}(s_{-i})\right\} \\
        &= 2^{n-1}\left\{(s_{i}-t_{i})X_{i}(t_{i}) + U(t_{i})\right\} \\
        &= 2^{n-1}\left\{(s_{i}-t_{i})X_{i}(t_{i}) + \int_{\underline{s}}^{t_{i}} X_{i}(u_{i})du_{i}\right\} + U_{i}(\underline{s}) \\
        &= U_{i}(s_{i})+2^{n-1}\int_{s_{i}}^{t_{i}}(X_{i}(u_{i})-X_{i}(t_{i}))du_{i},
    \end{align*}
    where the last two equations use \eqref{eq: EV} noting that its right-hand side equals the agent's indirect utility. 
    For each case $s_{i}\geq t_{i}$ or $t_{i}\geq s_{i}$, \eqref{eq: M} implies that the second term is negative, and therefore, we obtain $U_{i}(s_{i})\geq U_{i}(t_{i};s_{i})$ and thus \eqref{eq: IC}. 
    Finally, \eqref{eq: EV} implies that \eqref{eq: P} follows from $U_{i}(\underline{s})\geq 0$. 
\end{proof}

Next, we consider a class of functions to which any feasible indirect utility function must belong. 
Let $\overline{U}_{i}$ be the indirect utility function induced by efficient allocation, i.e., for each profile of signals $s\in [0,1]^{n}$, we allocate the good to the agent $i$ if and only if the posterior belief conditional on $s$ is above $1/2$. 
Equivalently, 
\begin{align*}\label{eq: UB}\tag{UB}
    \overline{U}_{i}(s_{i})
    = \mE[\omega\cdot \mI\{\lr{s}\geq 1\}\mid s_{i}]. 
\end{align*} 
For any function $h:[\underline{s},\overline{s}]\rightarrow \mR$, we say that $\tilde{h}:[0,1]\rightarrow \mR$ is an \textit{extension} of $h$ if it coincides with $h$ over $[\underline{s},\overline{s}]$. 
We say that $\tilde{h}$ is a \textit{convex extension} of $h$ if it is also convex. 
Then, we get the following. 

\begin{lemma}\label{lem: feasible indirect utility necessary conditions}
    Suppose that a feasible mechanism induces an indirect utility function $U_{i}:[\underline{s},\overline{s}]\rightarrow \mR$ given by \eqref{eq: U}. 
    Then, it has an extension $\tilde{U}_{i}:[0,1]\rightarrow \mR$ that is increasing, convex, $2$-Lipschitz continuous, and bounded above by $\overline{U}_{i}$ over the domain $[0,1]$. 
\end{lemma} 

\begin{proof}[Proof of Lemma \ref{lem: feasible indirect utility necessary conditions}]
    Suppose that a mechanism $x_{i}:[\underline{s},\overline{s}]^{n}\rightarrow \mR$ for agent $i$ induces an indirect utility function $U_{i}:[\underline{s},\overline{s}]\rightarrow \mR$. 
    First, we construct an extended mechanism and an extended indirect utility function. 

    Let $s^{\texttt{min}}_{i}<\underline{s}$ be a hypothetical type who obtains the expected payoff of $0$ from consuming an option $x_{i}(\underline{s},s_{-i})$. 
    Then, consider an extended mechanism $\tilde{x}_{i}:[0,1]^{n}\rightarrow [0,1]$ such that, for all $s_{-i}\in [\underline{s},\overline{s}]^{n-1}$,
    \begin{align*}
        \tilde{x}_{i}(s_{i},s_{-i}) 
        = 
        \begin{cases}
            0 \quad &\text{if } s_{i}\leq s_{i}^{\texttt{min}}, \\
            x_{i}(\underline{s},s_{-i}) \quad &\text{if } s_{i}\in (s_{i}^{\texttt{min}},\underline{s}), \\
            x_{i}(s_{i},s_{-i}) \quad &\text{if } s_{i}\in (\underline{s},\overline{s}), \\
            x_{i}(\overline{s},s_{-i}) \quad &\text{if } \overline{s}\leq s_{i}.
        \end{cases}
    \end{align*}
    Let $\tilde{x}_{i}(s_{i},s_{-i})=0$ for all $s_{-i}\notin [\underline{s},\overline{s}]^{n-1}$. 
    Then, the extended mechanism induces an indirect utility function $\tilde{U}_{i}:[0,1]\rightarrow \mR$ that is an extension of $U_{i}$. 

    We claim that the extended mechanism $\tilde{x}_{i}$ must be feasible when the signal space is $[0,1]^{n}$. 
    Let $\tilde{X}_{i}:[0,1]\rightarrow [0,1]$ be the induced interim allocation, which is also an extension of the original interim allocation $X_{i}$. 
    Then, as in \eqref{eq: U star} in the proof of Lemma \ref{lem: feasible mechanism}, the expected payoff to type $s_{i}\in [0,1]$ from reporting $t_{i}$ can be computed as 
    \begin{align*}
        \tilde{U}_{i}(t_{i};s_{i}) = 2^{n-1}\left\{s_{i}\tilde{X}_{i}(t_{i})-\int_{s_{-i}}\tilde{x}_{i}(t_{i},s_{-i})\prod_{k\neq i}(1-s_{k})d\mF^{\otimes}(s_{-i})\right\}. 
    \end{align*} 
    By construction, the co-domains of $\tilde{x}_{i}$ and $\tilde{X}_{i}$ coincide with those of $x_{i}$ and $X_{i}$, respectively. 
    Therefore, for any $s_{i}\in [\underline{s},\overline{s}]$,
    \begin{align*} 
        \tilde{U}_{i}(s_{i})
        = U_{i}(s_{i})
        \geq \max_{t_{i}\in [\underline{s},\overline{s}]}U_{i}(t_{i};s_{i}) 
        = \max_{t_{i}\in [0,1]}\tilde{U}_{i}(t_{i};s_{i}), 
    \end{align*}
    that is, no type $s_{i}\in [\underline{s},\overline{s}]$ has an incentive to misreport her type. 
    This implies that, for any type $s_{i}\leq \underline{s}$ and her misreport $t_{i}\in [0,1]$, we also have
    \begin{align*} 
        \tilde{U}_{i}(t_{i};s_{i})
        &= 2^{n-1}\left\{s_{i}\tilde{X}_{i}(t_{i})-\int_{s_{-i}}\tilde{x}_{i}(t_{i},s_{-i})\prod_{k\neq i}(1-s_{k})d\mF^{\otimes}(s_{-i})\right\} \\
        &= 2^{n-1}(s_{i}-\underline{s})\tilde{X}_{i}(t_{i}) + \tilde{U}_{i}(t_{i};\underline{s}) \\
        &\leq 2^{n-1}(s_{i}-\underline{s})\tilde{X}_{i}(t_{i}) + \tilde{U}_{i}(\underline{s};\underline{s}) \\
        &\leq 2^{n-1}(s_{i}-\underline{s})X_{i}(\underline{s}) + \tilde{U}_{i}(\underline{s};\underline{s}) \\
        &= 2^{n-1}\left\{s_{i}X_{i}(\underline{s})-\int_{s_{-i}}x_{i}(\underline{s},s_{-i})\prod_{k\neq i}(1-s_{k})d\mF^{\otimes}(s_{-i})\right\} 
        = \tilde{U}_{i}(s_{i}),
    \end{align*}
    where the first inequality follows since type $\underline{s}$ has no incentive to misreport, and the second inequality follows from $s_{i}-\underline{s}\leq 0$ because Lemma \ref{lem: feasible mechanism} shows that $X_{i}$ is increasing and $\tilde{X}_{i}$ has the same co-domain as $X_{i}$. 
    Therefore, no type $s_{i}\leq \underline{s}$ has a profitable misreport. 
    A symmetric argument shows that types $s_{i}\geq \overline{s}$ have no profitable misreports as well. 
    
    Therefore, the extended mechanism satisfies incentive compatibility \eqref{eq: IC} over the extended type space $[0,1]$. 
    By construction, $\tilde{U}_{i}\geq 0$ must also hold, hence \eqref{eq: P}. 
    That is, the extended mechanism is feasible. 
    Since $\tilde{U}_{i}$ is an extension of $U_{i}$, it remains to show that the extended indirect utility function $\tilde{U}_{i}$ satisfies the conditions listed in the statement. 

    For any incentive compatible mechanism, Lemma \ref{lem: feasible mechanism} implies the envelope formula \eqref{eq: EV}, which in particular implies 
    \begin{align*}
        \tilde{U}_{i}(s_{i})-\tilde{U}_{i}(\underline{s}) = 2^{n-1}\int_{\underline{s}}^{s_{i}}\tilde{X}_{i}(t_{i})dt_{i}.
    \end{align*}
    Since $\tilde{X}_{i}(t_{i})\geq 0$ for each $t_{i}$, the indirect utility $U_{i}$ must be increasing. 
    Also, \eqref{eq: M} implies that $\tilde{X}_{i}$ is increasing, and therefore, $\tilde{U}_{i}$ is convex. 
    Moreover, $\tilde{X}_{i}(t_{i})\leq 2^{2-n}$ because $\tilde{X}_{i}(t_{i})$ takes the highest value when $\tilde{x}_{i}(t_{i},s_{-i})=1$ for all $s_{-i}$ and $\mE[s_{k}]=1/2$.
    therefore, the slope of $\tilde{U}_{i}$ is bounded above by $2$, and thus, $\tilde{U}_{i}$ is $2$-Lipschitz continuous. 
    The last condition is obvious. 
\end{proof}

For convenience, let $\mathcal{U}_{i}$ be the class of (extended) functions that satisfy the conditions in Lemma \ref{lem: feasible indirect utility necessary conditions}. 
Formally, 
\begin{align*}
    \mathcal{U}_{i} = \left\{U_{i}:[0,1]\rightarrow \mR\ \middle|\ U_{i} \text{ is increasing, convex, $2$-Lipschitz, and } 0\leq U_{i}\leq \overline{U}_{i} \right\}.
\end{align*}
Then, the lemma states that any feasible indirect utility for agent $i$ must belong to this set $\mathcal{U}_{i}$. 
As discussed in the main section, the set $\mathcal{U}_{i}$ does not characterize the set of feasible indirect utility functions.  

Therefore, we need to further restrict the space of functions over which we will optimize. 
The next lemma does so, with recourse to optimality. 
Note that a constant mechanism such that $x_{i}(s_{i},s_{-i})=\mI\{s_{i}\geq 0.5\}$ satisfies both \eqref{eq: P} and \eqref{eq: IC} and induces an indirect utility function $\underline{U}_{i}(s_{i})=\max\{0,2s_{i}-1\}$. 
The following lemma states that any optimal mechanism must yield a pointwise higher utility than this benchmark to the agent. 

\begin{lemma}\label{lem: feasible indirect utility optimality conditions}
    Suppose that an optimal mechanism induces an indirect utility function $U_{i}$. 
    Then, we have $U_{i}(s_{i})\geq 2s_{i}-1$ for every $s_{i}$.
\end{lemma} 

\begin{proof}[Proof of Lemma \ref{lem: feasible indirect utility optimality conditions}]
    Suppose that an optimal mechanism induces an indirect utility function such that $U_{i}(s_{i})<2s_{i}-1$ for some $s_{i}$. 
    Lemma \ref{lem: feasible indirect utility necessary conditions} implies that $U_{i}$ has an extension in $\mathcal{U}_{i}$, which in particular implies that $U_{i}$ is $2$-Lipschitz. 
    Therefore, for any point $s_{i}\in [0,1]$, any element in the subgradients of $U_{i}$ is bounded above by $2$. 
    Hence, $U_{i}(t_{i})<2t_{i}-1$ implies $U_{i}(s_{i})<2s_{i}-1$ for all $s_{i}\geq t_{i}$.
    The set of points $I_{i}=\{s_{i}\mid U_{i}(s_{i})<2s_{i}-1\}$ is therefore a convex interval of a form $I_{i}=(t_{i},1]$ for some $t_{i}\geq 1/2$. 
    
    Now, consider the following mechanism:
    \begin{align*}
        x^{*}_{i}(s_{i},s_{-i})
        = 
        \begin{cases}
            x_{i}(s_{i},s_{-i}) \quad &\text{if } s_{i}\leq t_{i}, \\
            1 \quad &\text{if } s_{i}> t_{i}.
        \end{cases}
    \end{align*}
    This mechanism $x^{*}_{i}$ induces an indirect utility $U^{*}_{i}(s_{i})=\max\{U_{i}(s_{i}),2s_{i}-1\}$ and $U^{*}_{i}$ has an extension in $\mathcal{U}_{i}$. 
    Moreover, since $x^{*}_{i}$ is pointwise greater than $x_{i}$, the designer prefers $x^{*}_{i}$ over $x_{i}$. 
    Therefore, it remains to check that $x^{*}_{i}$ satisfies both \eqref{eq: M} and \eqref{eq: EV}. 

    Note that \eqref{eq: M} is obvious: if $X_{i}$ is the original interim allocation function under the mechanism $x_{i}$, the mechanism $x^{*}_{i}$ induces the interim allocation function that coincides with $X_{i}$ up to $t_{i}$ and then takes the maximum possible value, which equals $2^{2-n}$, after $t_{i}$. 
    Since $X_{i}$ is increasing, this is also an increasing function. 
    
    To see that \eqref{eq: EV} holds as well, note that the right-hand side of the equation \eqref{eq: EV} equals $2^{1-n}[U(t_{i})-U(\underline{s})]$. 
    Then, for each $s_{i}\geq t_{i}$, 
    \begin{align*}
        \int_{\underline{s}}^{s_{i}}X_{i}(r_{i})dr_{i}
        &= \int_{\underline{s}}^{t_{i}}X_{i}(r_{i})dr_{i} + \int_{t_{i}}^{s_{i}}X_{i}(r_{i})dr_{i} \\
        &= 2^{1-n}[U(t_{i})-U(\underline{s})] + \int_{t_{i}}^{s_{i}}X_{i}(r_{i})dr_{i} \\
        &= 2^{1-n}[U(t_{i})-U(\underline{s})] + \int_{t_{i}}^{s_{i}}2^{2-n}dr_{i} \\
        &= 2^{1-n}[(2t_{i}-1)-U(\underline{s})] + (s_{i}-t_{i})2^{2-n} 
        = 2^{1-n}[(2s_{i}-1)-U(\underline{s})],
    \end{align*}
    where the last term equals the type $s_{i}$'s indirect utility from the mechanism $x_{i}^{*}$. 
    Hence, \eqref{eq: EV} is also satisfied. 
\end{proof}

For later use, define the set of functions in $\mathcal{U}_{i}$ that is also pointwise above the linear function $2s_{i}-1$. 
\begin{align*}
    \mathcal{U}_{i}^{*} = \left\{U_{i}:[0,1]\rightarrow\mR\ \middle|\ U_{i} \text{ is increasing, convex, and } \underline{U}_{i}\leq U_{i}\leq \overline{U}_{i} \right\},
\end{align*}
where we set $\underline{U}_{i}(s_{i})=\max\{0,2s_{i}-1\}$. 
We remove $2$-Lipschitz continuity because it is redundant. 
To see this, note that $\underline{U}_{i}(1)=\overline{U}_{i}(1)=1$ and both functions have slope $2$ at the end point $s_{i}=1$. 
Therefore, any $U_{i}\in \mathcal{U}_{i}^{*}$ also have the same slope $2$ at $s_{i}=1$. 
Since $U_{i}$ is increasing and convex, the slope at any point is in $[0,2]$, hence $2$-Lipschitz continuous. 
As discussed in the main section, this set $\mathcal{U}^{*}_{i}$ is not a characterization for a feasible indirect utility function. 

Lemma \ref{lem: feasible indirect utility optimality conditions} implies that the optimal indirect utility function for agent $i$ can be found in the set $\mathcal{U}_{i}^{*}$. 
This is useful because, as the next lemma claims, we can also rewrite the objective function in terms of indirect utility. 
In particular, it is a linear functional of indirect utility. 

\begin{lemma}\label{lem: objective function is linear in indirect utility}
    The objective function is linear in indirect utility. 
    In particular, 
    \begin{align*}
        \mE\left[x_{i}(s)\right]
        &= - \int_{\underline{s}}^{\overline{s}}\left\{3(1-2s_{i})f(s_{i})+2s_{i}(1-s_{i})\frac{df(s_{i})}{ds_{i}}\right\}U_{i}(s_{i})ds_{i} \\
        &\quad \quad \quad \quad +2\overline{s}(1-\overline{s})f(\overline{s})U_{i}(\overline{s})-2\underline{s}(1-\underline{s})f(\underline{s})U_{i}(\underline{s}), \label{eq: OBJ}\tag{OBJ}
    \end{align*}
    for each agent $i$ and feasible mechanism $x_{i}$.  
\end{lemma} 

\begin{proof}[Proof of Lemma \ref{lem: objective function is linear in indirect utility}]
    First, analogous to the derivation in the proof of Lemma \ref{lem: feasible mechanism}, we can calculate that
    \begin{align*}
        \mE[x_{i}(s)\mid \omega]
        &= \int_{s}\left\{\prod_{k}\frac{f_{\omega}(s_{k})}{f(s_{k})}\right\} x_{i}(s)d\mF^{\otimes}(s),
    \end{align*}
    for each mechanism $x_{i}$ and state $\omega\in \{-1,+1\}$. 
    Moreover, $s_{i}=\mP[\omega=+1\mid s_{i}]$ implies $f_{+1}(s_{k})/f(s_{k})=2s_{k}$ and $f_{-1}(s_{k})/f(s_{k})=2(1-s_{k})$. 
    
    Second, for notational convenience, define 
    \begin{align*}
        A_{i}(s_{i}) &= \int_{s_{-i}} \left\{\prod_{k\neq i}s_{k}\right\}x_{i}(s) d\mF^{\otimes}(s_{-i}),\\
        B_{i}(s_{i}) &= \int_{s_{-i}} \left\{\prod_{k\neq i}(1-s_{k})\right\}x_{i}(s) d\mF^{\otimes}(s_{-i}). 
    \end{align*}
    Then, we can express these functions in terms of interim allocation $X_{i}$ using $X_{i}(t_{i})=A_{i}(s_{i})+B_{i}(s_{i})$ and the envelope formula \eqref{eq: EV}. 
    We get 
    \begin{align*}
        A_{i}(s_{i}) &= (1-s_{i})X_{i}(s_{i}) + \int_{\underline{s}}^{s_{i}}X_{i}(t_{i})dt_{i} + 2^{1-n}U_{i}(\underline{s}), \\
        B_{i}(s_{i}) &= s_{i}X_{i}(s_{i}) - \int_{\underline{s}}^{s_{i}}X_{i}(t_{i})dt_{i} - 2^{1-n}U_{i}(\underline{s}).
    \end{align*}
    The principal's payoff coming from agent $i$ is therefore
    \begin{align*}
        & 2^{1-n}\mE[x_{i}(s)] \\
        &= \int_{\underline{s}}^{\overline{s}} \left\{s_{i} A_{i}(s_{i}) + (1-s_{i})B_{i}(s_{i})\right\}d\mF(s_{i}) \\
        &= \int_{\underline{s}}^{\overline{s}} \left\{2s_{i}(1-s_{i}) X_{i}(s_{i}) + (2s_{i}-1)\left\{\int_{\underline{s}}^{s_{i}}X_{i}(t_{i})dt_{i}+2^{1-n}U_{i}(\underline{s})\right\}\right\}d\mF(s_{i}) \\
        &= \int_{\underline{s}}^{\overline{s}}2s_{i}(1-s_{i}) X_{i}(s_{i})d\mF(s_{i}) + \int_{\underline{s}}^{\overline{s}}\int_{t_{i}}^{\overline{s}}(2s_{i}-1)X_{i}(t_{i})d\mF(s_{i})dt_{i} \\
        &= \int_{\underline{s}}^{\overline{s}}2s_{i}(1-s_{i}) X_{i}(s_{i})d\mF(s_{i}) + \int_{\underline{s}}^{\overline{s}}\left\{\int_{t_{i}}^{\overline{s}}\frac{2s_{i}-1}{f(t_{i})}d\mF(s_{i})\right\}X_{i}(t_{i})d\mF(t_{i}) \\
        &= \int_{\underline{s}}^{\overline{s}} \left\{2s_{i}(1-s_{i}) + \int_{s_{i}}^{\overline{s}}\frac{2t_{i}-1}{f(s_{i})}d\mF(t_{i})\right\} X_{i}(s_{i})d\mF(s_{i}), \label{eq: OBJ-X}\tag{OBJ-X} 
    \end{align*}
    where in the second equation we change the order of integrals and in the last equation we just renames the variables. 
    Note that the term associated with $U_{i}(\underline{s})$ disappears because $s_{i}$ has mean $1/2$. 

    Finally, note that Lemma \ref{lem: feasible indirect utility necessary conditions} implies that $U_{i}$ is convex and therefore has a derivative almost everywhere. 
    Then, for almost all $s_{i}$, the envelope formula \eqref{eq: EV} implies $X_{i}(s_{i})=dU_{i}(s_{i})/ds_{i}$. 
    Therefore
    \begin{align*}
        & \int_{\underline{s}}^{\overline{s}} \left\{2s_{i}(1-s_{i}) + \int_{s_{i}}^{\overline{s}}\frac{2t_{i}-1}{f(s_{i})}d\mF(t_{i})\right\} X_{i}(s_{i})d\mF(s_{i}) \\
        &= 2^{1-n}\int_{\underline{s}}^{\overline{s}} \left\{2s_{i}(1-s_{i})f(s_{i}) + \int_{s_{i}}^{\overline{s}}[2t_{i}-1]d\mF(t_{i})\right\} \frac{dU_{i}(s_{i})}{ds_{i}} ds_{i} \\
        & = 2^{1-n}\left\{2\overline{s}(1-\overline{s})f(\overline{s})U_{i}(\overline{s})-2\underline{s}(1-\underline{s})f(\underline{s})U_{i}(\underline{s})\right\} \\
        & \quad -2^{1-n}\int_{\underline{s}}^{\overline{s}} \left\{2s_{i}(1-s_{i})\frac{df(s_{i})}{ds_{i}}+2(1-2s_{i})f(s_{i})-(2s_{i}-1)f(s_{i})\right\} U_{i}(s_{i}) ds_{i} \\
        &= 2^{1-n}\left\{2\overline{s}(1-\overline{s})f(\overline{s})U_{i}(\overline{s})-2\underline{s}(1-\underline{s})f(\underline{s})U_{i}(\underline{s})\right\} \\ 
        &\quad \quad -2^{1-n}\int_{\underline{s}}^{\overline{s}} \left\{3(1-2s_{i})f(s_{i})+2s_{i}(1-s_{i})\frac{df(s_{i})}{ds_{i}}\right\} U_{i}(s_{i}) ds_{i}, 
    \end{align*}
    where the second equation follows from integration by parts. 
    Combining this equation with \eqref{eq: OBJ-X} completes the proof. 
\end{proof}

The next lemma is a preliminary result, which observes that a threshold mechanism with a constant threshold yields a linear indirect utility function. 

\begin{lemma}\label{lem: threshold mechanism induces linear utility}
    Consider a mechanism $x_{i}(s_{i},s_{-i})=\kappa\cdot \mI\{\lr{s_{-i}}\geq \tau\}$ for each $s_{i}$ and $s_{-i}$. 
    Then, it induces a linear indirect utility $U_{i}(s_{i})=\max\{a_{i}s_{i}-b_{i}, 0\}$ such that
    \begin{align*}
        a_{i} &= 2^{n-1}\kappa \int_{s_{-i}} \left\{\prod_{k\neq i}s_{k} + \prod_{k\neq i}(1-s_{k})\right\}\mI\{\lr{s_{-i}}\geq \tau\}d\mF^{\otimes}(s_{-i}), \label{eq: Slope}\tag{Slope} \\
        b_{i} &= 2^{n-1}\kappa \int_{s_{-i}} \left\{\prod_{k\neq i}(1-s_{k})\right\}\mI\{\lr{s_{-i}}\geq \tau\}d\mF^{\otimes}(s_{-i}), \label{eq: Intercept}\tag{Intercept}
    \end{align*}
    for each $s_{i}$. 
\end{lemma} 

\begin{proof}[Proof of Lemma \ref{lem: threshold mechanism induces linear utility}]
    Note that the mechanism does not depend on agent $i$'s signal $s_{i}$, and therefore it is optimal for agent $i$ to report the true type. 
    Therefore, agent $i$ participates in the mechanism if and only if the interim expected payoff under truthful report is above $0$. 
    The formula \eqref{eq: U star} gives a interim payoff from participating in the mechanism that is exactly given by $U_{i}(s_{i})=a_{i}s_{i}-b_{i}$, with $a_{i}$ and $b_{i}$ being defined in \eqref{eq: Slope} and \eqref{eq: Intercept}. 
\end{proof}

The next lemma proves the converse of Lemma \ref{lem: threshold mechanism induces linear utility}. 
A certain class of linear functions is implementable with the class of threshold mechanisms. 

\begin{lemma}\label{lem: threshold mechanism induces any linear utility}
    Consider any function $U_{i}(s_{i})=a_{i}s_{i}-b_{i}$ with $a_{i}\in [0,2]$ and $b_{i}\geq 0$ that is pointwise below $\overline{U}_{i}$ and is not pointwise below $\max\{0,2s_{i}-1\}$ over the interval $[0,1]$. 
    Then, for some $\kappa \in [0,1]$ and $\tau\geq 0$, the mechanism $x_{i}(s_{i},s_{-i})=\kappa\cdot \mI\{\lr{s_{-i}}\geq \tau\}$ induces $U_{i}$ as the indirect utility function. 
    In particular, $\kappa=1$ if $U_{i}$ coincides with $\overline{U}_{i}$ at some point in $[0,1]$. 
\end{lemma} 

\begin{proof}[Proof of Lemma \ref{lem: threshold mechanism induces any linear utility}] 
    Fix a pair $(a_{i},b_{i})$ that satisfies the conditions in the statement. 
    First, note that we have $\mI\{\lr{s_{-i}}\geq 0\}=1$ for all profile $s\in [0,1]^{n}$. 
    Hence, Lemma \ref{lem: threshold mechanism induces linear utility} implies that, at $\tau=0$, the mechanism $x_{i}(s_{i},s_{-i})=\kappa\cdot \mI\{\lr{s_{-i}}\geq \tau\}$ induces the slope 
    \begin{align*}
        2^{n-1}\kappa \int_{s_{-i}} \left\{\prod_{k\neq i}s_{k} + \prod_{k\neq i}(1-s_{k})\right\}d\mF^{\otimes}(s_{-i}) 
        = 2^{n-1}\kappa \left\{\left(\frac{1}{2}\right)^{n-1} + \left(\frac{1}{2}\right)^{n-1}\right\}
        = 2\kappa.
    \end{align*} 
    Moreover, as $\tau\rightarrow \infty$, the slope converges to zero. 
    Therefore, if the expression \eqref{eq: Slope} is continuous in $\tau$, the intermediate value theorem shows that for each $\kappa\in [a_{i}/2,1]$, there exists a threshold $\tau(\kappa)\geq 0$ such that the slope equals $a_{i}$. 

    To prove the continuity, let $g(s_{-i};\tau)$ be the functional form inside the integral. Then, consider any sequence $(\tau_{n})_{n\in \mN}$ that converges to a point $\tau$. 
    Set $g_{n}(s_{-i})=g(s_{-i};\tau_{n})$. 
    Then, $g_{n}$ converges pointwise to $g(s_{-i})=g(s_{-i};\tau)$.  
    Moreover, we have $|g_{n}(s_{-i})|\leq 1$ for all $n$, and therefore, we can apply the dominated convergence theorem and conclude that the integral of $g_{n}(s_{-i})$ converges to that of $g(s_{-i})$. 
    This implies that the expression \eqref{eq: Slope} is continuous in $\tau$. 

    Next, consider the intercept \eqref{eq: Intercept}. 
    At $\kappa=a_{i}/2$ and $\tau(\kappa)=0$, the mechanism $x_{i}(s_{i},s_{-i})=\kappa\cdot \mI\{\lr{s_{-i}}\geq \tau(\kappa)\}$ induces the slope $\kappa=a_{i}/2$. 
    Note that if $a_{i}\in [0,2]$ and $a_{i}s_{i}-b_{i}$ is not pointwise below $\max\{0,2s_{i}-1\}$, we must have $a_{i}/2-b_{i}\geq 0$.  
    Therefore, $\kappa=a_{i}/2$ and $\tau=\tau(\kappa)$ implies the intercept that is larger (in absolute value) than $b_{i}$. 
    Consider next $\kappa=1$ and $\tau=\tau(\kappa)$. 
    Then, consider a type $t_{i}\in [0,1]$ such that $\tau=[1-t_{i}]/t_{i}$. 
    Note that $\tau\geq 0$ implies the existence of such a type within the unit interval $[0,1]$. 
    For this type $t_{i}$, $\lr{s_{-i}}\geq \tau$ is equivalent to $\lr{t_{i},s_{-i}}\geq 1$. 
    Therefore, this mechanism induces the efficient allocation to type $t_{i}$, which implies that the induced indirect utility function coincides with $\overline{U}_{i}$ at the point $t_{i}$. 
    Moreover, by the construction of $\tau(\kappa)$, it has slope $a_{i}$. 
    Since $U_{i}$ is pointwise below $\overline{U}_{i}$ by assumption, we conclude that $\kappa=1$ and $\tau=\tau(\kappa)$ induce the intercept smaller (in absolute value) than $b_{i}$. 
    We can likewise show by dominated convergence theorem that the intercept is continuous in $\kappa$ and $\tau(\kappa)$, and therefore, the intermediate value theorem applies, which implies the existence of $\kappa\in [a_{i}/2,1]$ and $\tau=\tau(\kappa)$ such that the mechanism $x_{i}(s_{i},s_{-i})=\kappa\cdot \mI\{\lr{s_{-i}}\geq \tau\}$ induce $U_{i}$ as the induced indirect utility function. 
    This completes the proof. 
\end{proof}

The lemma below characterizes the extreme points of the set $\mathcal{U}_{i}^{*}$. 
For a given set $S$ of a vector space, $x\in S$ is an \textit{extreme point} of the set $S$ if it is not written as any convex combination of two distinct points in $S$. 

\begin{lemma}\label{lem: extreme points}
    Suppose that a function $U_{i}\in \mathcal{U}^{*}_{i}$ is an extreme point of $\mathcal{U}^{*}_{i}$. 
    Then, there exists a countable collection of non-singleton intervals $\mathcal{T}_{i}$ such that 
    \begin{itemize}
        \item [1.] $U_{i}(s_{i})\in \{\underline{U}_{i}(s_{i}),\overline{U}_{i}(s_{i})\}$ for all $s_{i}\notin \bigcup_{T_{i}\in \mathcal{T}_{i}}T_{i}$.
        \item [2.] For each interval $T_{i}\in \mathcal{T}_{i}$, $U_{i}$ is linear on $T_{i}$, lies strictly between $\underline{U}_{i}$ and $\overline{U}_{i}$ in the interior of $T_{i}$, and satisfies at least one of the following properties:
        \begin{itemize}
            \item [(a).] $U_{i}$ coincides with $\overline{U}_{i}$ at one of the end points of $T_{i}$. 
            \item [(b).] For each end point $s_{i}$ of $T_{i}$, either $U_{i}(s_{i})=\underline{U}_{i}(s_{i})$, or there exists an interval $S_{i}\in \mathcal{T}_{i}$ having $s_{i}$ as an end point such that the interval $S_{i}$ has the property (a). 
        \end{itemize} 
    \end{itemize}
\end{lemma} 

\begin{proof}[Proof of Lemma \ref{lem: extreme points}]
    This lemma is an application of \citet{augias2025economics}. 
    Their Theorem 1 characterizes the extreme points of any convex function interval
    \begin{align*}
        \mathcal{U} = \left\{u:X\rightarrow \mR\ \middle|\ u \text{ is convex, } \underline{u}\leq u\leq \overline{u}, \text{ and } \partial u(X)\subset S\right\},
    \end{align*}
    where $X=[0,1]$ and $S=[\underline{s},\overline{s}]$ are closed intervals, $\underline{u}$ and $\overline{u}$ are differentiable convex functions defined over $X$ such that $\partial \underline{u}(X), \partial \overline{u}(X)\subset S$. 
    
    The set $\mathcal{U}_{i}^{*}$ is a convex function interval such that $X=[0,1]$, $S=\mR$, $\underline{u}=0$, and $\overline{u}=\overline{U}_{i}$.
    Note that we do not need to require $S=\mR_{+}$ because $\overline{U}_{i}(0)=0$ and therefore $U_{i}(0)=0$ for all $U_{i}\in \mathcal{U}^{*}_{i}$; since $U_{i}(s_{i})\geq 0$ for all $s_{i}$ must hold, $U_{i}$ must have a non-negative slope at $s_{i}=0$. 
    Since $U_{i}\in \mathcal{U}_{i}^{*}$ is convex and therefore its slope is increasing, monotonicity constraint is redundant. 
    We will show in Lemma \ref{lem: upper bound is differentiable} below that $\overline{U}_{i}$ is differentiable, and hence, Theorem 1 of \citet{augias2025economics} applies.
    
    If we set $S=\mR$, their Theorem 1 amounts to stating that $u\in \mathcal{U}$ is an extreme point of the convex function interval $\mathcal{U}$ if and only if there exists a countable collection $\mathcal{X}=\{X_{n}\}_{n\in \mN}$ of maximal and non-singleton intervals $X_{n}=[a_{n},b_{n}]\subset X$ such that 
    \begin{itemize} 
        \item [1.] For all $x\notin \bigcup_{n\in \mN}X_{n}$, $u(x)\in \{\underline{u}(x),\overline{u}(x)\}$.
        \item [2.] For each $n\in \mN$, $u$ is linear over $X_{n}$, lies strictly between $\underline{u}$ and $\overline{u}$ in the interior of $X_{n}$, and at least one of the following conditions holds:

        \begin{itemize}
            \item [(a).] There exists $y\in \{a_{n},b_{n}\}$ such that, for all $x\in X_{n}$, $u(x)=\overline{u}(y)+s(x-y)$ with $s\in \partial \overline{u}(y)$. 

            \item [(b).] For each $x\in \{a_{n},b_{n}\}$, either there exists $m\in \mN$ such that $b_{m}=a_{n}$ or $b_{n}=a_{m}$ and $m$ satisfies condition (a), or $u(x)=\underline{u}(x)$. 

            \item [(c).] Either $a_{n}=0$, $u(a_{n})\in \{\underline{u}(a_{n}),\overline{u}(a_{n})\}$, and either there exists $m\in \mN$ such that $a_{m}=b_{n}$ and $m$ satisfies the condition (a) or $u(b_{n})=\underline{u}(b_{n})$. 

            Or, symmetrically, $b_{n}=1$, $u(b_{n})\in \{\underline{u}(b_{n}),\overline{u}(b_{n})\}$, and either there exists $m\in \mN$ such that $a_{n}=b_{m}$ and $m$ satisfies the condition (a) or $u(a_{n})=\underline{u}(a_{n})$. 
        \end{itemize}
    \end{itemize}
    In our convex function interval $\mathcal{U}_{i}^{*}$, the lower bound $\underline{U}_{i}$ is piecewise linear. 
    Therefore, it is easy to check that conditions in their theorem are simplified to the conditions we give in the statement. 
\end{proof}

Finally, we prove Theorem \ref{thm: optimal mechanism countable partition}.

\begin{proof}[Proof of Theorem \ref{thm: optimal mechanism countable partition}]
    By Lemmas \ref{lem: feasible indirect utility optimality conditions} and \ref{lem: objective function is linear in indirect utility}, a mechanism for agent $i$ is optimal if and only if it induces an indirect utility function $U_{i}$ that maximizes \eqref{eq: OBJ} subject to $U_{i}\in \mathcal{U}_{i}^{*}$. 
    The set $\mathcal{U}^{*}_{i}$ is a non-empty, convex, and compact set in the topology induced by sup-norm $\|\cdot \|_{\infty}$.\footnote{See Proposition 1 in \citet{augias2025economics}.} 
    
    Note also that the objective function \eqref{eq: OBJ} is a linear functional of indirect utility. 
    Moreover, since both $f$ and $df/ds_{i}$ are bounded, it is bounded by $M\cdot \|U\|_{\infty}$ for some finite constant $M<\infty$, which together with Proposition 2 in \citet{luenberger1997optimization} (Chapter 5.2) implies that \eqref{eq: OBJ} is also continuous in the sup-norm. 
    Therefore, Bauer's maximum principle implies that a solution is an extreme point of $\mathcal{U}_{i}^{*}$. 
    Let $U_{i}$ be any extreme point. 
    Then, Lemma \ref{lem: extreme points} implies the existence of a collection of intervals $\mathcal{T}_{i}$ with the conditions given in Lemma \ref{lem: extreme points}.

    We first prove a stronger statement: for any indirect utility function $U_{i}\in \mathcal{U}_{i}^{*}$, there exists a monotone threshold mechanism which implements it. 
    It will be then obvious by construction and the last statement in Lemma \ref{lem: threshold mechanism induces any linear utility} that the conditions in Theorem \ref{thm: optimal mechanism countable partition} are satisfied if $U_{i}$ is an extreme point. 
    
    For each point $t\in [0,1]$, take any point $a_{x}\in \partial U_{i}$ in the subgradients of $U_{i}$. 
    Since $U_{i}$ is convex, the existence of subgradients is guaranteed. 
    Set $b_{t}=-U_{i}(t)+ta_{t}$.
    Then, define a linear function
    \begin{align*}
        V_{t}(s_{i}) = a_{t}s_{i} - b_{t} = a_{t}\cdot (s_{i}-t) + U_{i}(t). 
    \end{align*}

    We argue that, by $U_{i}\in \mathcal{U}_{i}^{*}$, the linear function $V_{t}$ satisfies the conditions in Lemma \ref{lem: threshold mechanism induces any linear utility}. 
    Since $U_{i}$ is increasing and $2$-Lipschitz continuous, we clearly have $a_{t}\in [0,2]$. 
    Then, convexity of $U_{i}$ implies that $V_{t}$ is pointwise below $U_{i}$, which also shows that 
    \begin{align*}
        0 = U_{i}(0) \geq V_{t}(0) = -b_{t},
    \end{align*}
    hence $b_{t}\geq 0$. 
    Finally, note that $V_{t}(t)=U_{i}(t)\geq \max\{0,2t-1\}$ by the definition of $\mathcal{U}_{i}^{*}$.  
    Therefore, by Lemma \ref{lem: threshold mechanism induces any linear utility}, there exist $\kappa_{t}\in [0,1]$ and $\tau_{t}\geq 0$ such that the mechanism $x_{i}(s_{i},s_{-i})=\kappa_{t}\cdot \mI\{\lr{s_{-i}}\geq \tau_{t}\}$ that induces $V_{t}$. 

    Finally, consider the mechanism $x_{i}$ such that, for each $s_{i}\in [0,1]$, 
    \begin{align*}
        x_{i}(s_{i},s_{-i}) = \kappa_{s_{i}}\cdot \mI\{\lr{s_{-i}}\geq \tau_{s_{i}}\}.
    \end{align*} 
    This mechanism induces utility function $U_{i}$ under truthful report. 
    Since $U_{i}\geq 0$ implies that the mechanism satisfies participation constraint \eqref{eq: P}. 
    Finally, by Lemma \ref{lem: threshold mechanism induces linear utility}, for each point $t_{i}$, the mechanism $\kappa_{t_{i}}\cdot \mI\{\lr{s_{-i}}\geq \tau_{t_{i}}\}$ provides utility $V_{t_{i}}(s_{i})$ to each type $s_{i}$. 
    Therefore, since $U_{i}$ is an upper envelope of the linear functions $\{V_{t}\}_{t\in [0,1]}$ and $V_{s_{i}}(s_{i})=U_{i}(s_{i})$, reporting a true type is optimal for any type, hence \eqref{eq: IC}. 
    
    In summary, for any $U_{i}\in \mathcal{U}_{i}^{*}$, the monotone threshold mechanism that we construct obtains the indirect utility function $U_{i}$. 
    Constructing an associated partition is trivial: 
    \begin{align*}
        \mathcal{S}_{i}=\left\{S_{i}(a)\subset [0,1]\ \middle|\ S_{i}(a)=\{s_{i}\mid a\in \partial U_{i}(s_{i})\}\right\}.
    \end{align*}
    Any two elements of the set $\mathcal{S}_{i}$ are disjoint except for an overlap on a set of measure zero. 
    Moreover, since $U_{i}$ is convex, each set $S_{i}\in \mathcal{S}_{i}$ is an interval. 
    
    If $U_{i}$ is an extreme point of $\mathcal{U}^{*}_{i}$ with an associated collection on intervals $\mathcal{T}_{i}$, then $\mathcal{S}_{i}$ is a union of $\mathcal{T}_{i}\subset \mathcal{S}_{i}$. 
    Moreover, if two intervals $S_{i},T_{i}\in \mathcal{T}_{i}$ are adjacent, Lemma \ref{lem: extreme points} implies that over one of these two intervals, say $S_{i}$, either $U_{i}$ must coincide with $\underline{U}_{i}$ over $S_{i}$ or $U_{i}$ must coincide with $\overline{U}_{i}$ at an end point of $S_{i}$. 
    In the former case, the mechanism over $S_{i}$ is given by $\kappa\cdot \mI\{\lr{s_{-i}}\geq 0\}$, where $\kappa=0$ if $S_{i}\subset [0,1/2]$ and $\kappa=1$ if $S_{i}\subset [1/2,1]$. 
    Note that $U_{i}$ being equal to $\overline{U}_{i}$ and linear on $S_{i}$ implies $S_{i}\subset [0,1/2]$ or $S_{i}\subset [1/2,1]$.  
    In the latter case, the mechanism over $S_{i}$ is given by $\mI\{\lr{s_{-i}}\geq \tau\}$ by construction. 
\end{proof}

\begin{proof}[Proof of Theorem \ref{thm: optimal mechanism}]
    This is a corollary of Theorem \ref{thm: optimal mechanism countable partition}. 
\end{proof}

\section{Proof of Proposition \ref{prop: exclusion at the bottom}} \label{sec: appendixb} 

For this section, assume $\text{supp }\mF = (0,1)$. 
Recall from Lemma \ref{lem: objective function is linear in indirect utility} that the objective function is linear in indirect utility. 
In particular, rewriting \eqref{eq: OBJ}, the objective function is
\begin{align*}
    \int_{s_{i}}g_{i}(s_{i})U_{i}(s_{i})ds_{i}, 
    \text{ where }
    g_{i}(s_{i}) = -3(1-2s_{i})f(s_{i}) - 2s_{i}(1-s_{i})\frac{df(s_{i})}{ds_{i}}.
\end{align*}
Note that $g_{i}(0)<0$ and $g_{i}(1)>0$. 
In particular, since $g_{i}$ is continuous in $s_{i}$, there exists $\epsilon_{i}>0$ such that $g_{i}(s_{i})<0$ for all $s_{i}\in [0,\epsilon_{i}]$. 

\begin{lemma}\label{lem: upper bound is differentiable}
    $\overline{U}_{i}$ is strictly convex and differentiable in the interior $(0,1)$. 
    Moreover, the right derivative of $\overline{U}_{i}$ at $0$ equals $0$.  
\end{lemma} 

\begin{proof}[Proof of Lemma \ref{lem: upper bound is differentiable}]
    For the efficient mechanism, we get 
    \begin{align*}
        X_{i}(s_{i}) &= \int_{s_{-i}} \left\{\prod_{k\neq i}s_{k}+\prod_{k\neq i}(1-s_{k})\right\}\mI\{\lr{s_{i},s_{-i}}\geq 1\} d\mF^{\otimes}(s_{-i}) 
    \end{align*}
    for all $s_{i}\in [0,1]$. 
     
    For any convergent sequence $\{s_{i,n}\}_{n\in \mN}\rightarrow s_{i}$ let $x_{i,n}(s_{-i})=\mI\{\lr{s_{i,n},s_{-i}}\geq 1\}$. 
    Then, $x_{i,n}$ converges pointwise to $x_{i}(s_{-i})=\mI\{\lr{s_{i},s_{-i}}\geq 1\}$. 
    Again, $x_{i,n}$ is uniformly bounded by $1$, and therefore, the dominated convergence theorem implies that 
    \begin{align*}
        \lim_{n\rightarrow \infty}X_{i}(s_{i,n})
        = \int_{s_{-i}} \left\{\prod_{k\neq i}s_{k}+\prod_{k\neq i}(1-s_{k})\right\}\mI\{\lr{s_{i},s_{-i}}\geq 1\} d\mF^{\otimes}(s_{-i})
        = X_{i}(s_{i}).
    \end{align*} 
    Therefore, $X_{i}$ is continuous. 

    Now, by \eqref{eq: EV} of Lemma \ref{lem: feasible mechanism}, we have
    \begin{align*}
        U_{i}(s_{i}) = 2^{n-1}\int_{0}^{s_{i}}X_{i}(t_{i})dt_{i}.
    \end{align*}
    Since $X_{i}$ is continuous, the right-hand side is differentiable at any point $s_{i}\in (0,1)$, which equals $2^{n-1}X_{i}(s_{i})$. 
    Hence, so is the left-hand side $U_{i}(s_{i})$. 
    In particular, $X_{i}(0)=0$ and $X_{i}(1)=2^{2-n}$ under efficient allocation, which implies that the right derivative of $\overline{U}_{i}$ is zero at zero. 
    Finally, it is easy to see that $X_{i}$ is strictly increasing under efficient allocation, which implies that the derivative of $\overline{U}_{i}$ is strictly increasing. 
    Therefore, $\overline{U}_{i}$ is strictly convex. 
\end{proof} 

\begin{proof}[Proof of Proposition \ref{prop: exclusion at the bottom}]
    Let $U_{i}$ be an optimal indirect utility function. 
    Then, by Theorem \ref{thm: optimal mechanism countable partition}, there exists a countable collection of intervals $\mathcal{T}_{i}$ that satisfies the conditions in the statement of Theorem \ref{thm: optimal mechanism countable partition}. 

    First, suppose that $\mathcal{T}_{i}$ contains an interval of the form $[0,\epsilon_{i}]$ for some $\epsilon_{i}$. 
    Then, $U_{i}$ is linear over $[0,\epsilon_{i}]$. 
    Since $\overline{U}_{i}(0)=\underline{U}_{i}(0)=0$ and the right derivatives $\underline{U}_{i}$ and $\overline{U}_{i}$ are also equal to $0$ by Lemma \ref{lem: upper bound is differentiable}, the linear function $U_{i}$ must have slope zero, as otherwise either $U_{i}$ crosses $\underline{U}_{i}$ or exceeds $\overline{U}_{i}$ at some point close to $0$. 
    This implies that $U_{i}(s_{i})=0$ for all $s_{i}\in [0,\epsilon_{i}]$. 
    Moreover, by Theorem \ref{thm: optimal mechanism countable partition}, a monotone threshold mechanism induces $U_{i}$. 
    The only parameter that induces zero utility is $\kappa=0$, i.e., $x_{i}(s_{i},s_{-i})=0$ for all $s_{i}\in [0,\epsilon_{i}]$. 

    Second, suppose that $\mathcal{T}_{i}$ does not include any interval that includes $0$. 
    Then, there exists $\epsilon_{i}>0$ such that $U_{i}(s_{i})=\overline{U}_{i}(s_{i})$ for all $s_{i}\in [0,\epsilon_{i}]$. 
    Take $\epsilon_{i}$ small enough such that $g_{i}(s_{i})<0$ for all $s_{i}\leq \epsilon_{i}$. 
    Note that such an $\epsilon_{i}$ exists as discussed at the top of this section. 
    
    By Lemma \ref{lem: upper bound is differentiable}, $\overline{U}_{i}$ is differentiable at $\epsilon_{i}$. 
    Let $a_{i}>0$ be its derivative and $a_{i}s_{i}-b_{i}$ be the tangent line. 

    Now, consider the following function: 
    \begin{align*}
        V_{i}(s_{i}) 
        = 
        \begin{cases}
            \max\{0,a_{i}s_{i}-b_{i}\} \quad & \text{if } s_{i}\leq \epsilon_{i}, \\
            U_{i}(s_{i}) \quad & \text{if } s_{i}\geq \epsilon_{i}.
        \end{cases}
    \end{align*}
    Note that $U_{i}\in \mathcal{U}_{i}^{*}$ and the construction imply that $V_{i}$ is also increasing, convex, pointwise above $\underline{U}_{i}$, and pointwise below $\overline{U}_{i}$. 
    Hence, $V_{i}\in \mathcal{U}_{i}^{*}$. 
    Moreover, $V_{i}$ is pointwise below $U_{i}$. 
    Since $V_{i}(s_{i})<U_{i}(s_{i})$ only if $s_{i}\leq \epsilon_{i}$ and hence $g_{i}(s_{i})<0$, the alternative function provides a higher objective value than $U_{i}$. 
    This contradicts Lemma \ref{lem: feasible indirect utility optimality conditions}. 
\end{proof}

\section{Proofs of Proposition \ref{prop: BIC-EPIC non-equivalence} and Remark \ref{rem: interdependent value}} \label{sec: appendixd} 

\begin{proof}[Proof of Proposition \ref{prop: BIC-EPIC non-equivalence}]
    First, note that for any profile $s_{i}$ and $s_{-i}$, 
    \begin{align*}
        \mE[\omega\mid s_{i},s_{-i}]
        &= +1\cdot \frac{\prod_{k}f_{+1}(s_{k})}{\prod_{k}f_{+1}(s_{k})+\prod_{k}f_{-1}(s_{k})} -1\cdot \frac{\prod_{k}f_{-1}(s_{k})}{\prod_{k}f_{+1}(s_{k})+\prod_{k}f_{-1}(s_{k})} \\
        &= \frac{\prod_{k}s_{k} - \prod_{k}(1-s_{k})}{\prod_{k}s_{k} + \prod_{k}(1-s_{k})},
    \end{align*}
    where we use $s_{k}=\mP[\omega=+1\mid s_{k}]=f_{+1}(s_{k})/[f_{+1}(s_{k})+f_{-1}(s_{k})]$ in the last equation. 
    Now, we move on to the proof. 
    
    Let $x_{i}$ be any mechanism for agent $i$ that satisfies both \eqref{eq: P} and \eqref{eq: EPIC}. 
    Note that for the lowest type $s_{i}=0$, 
    \begin{align*}
        \mE[\omega\cdot x_{i}(0,s_{-i})\mid s_{i}=0] &= -1\cdot \mE[x_{i}(0,s_{-i})\mid \omega=-1].  
    \end{align*}
    Therefore, \eqref{eq: P} implies that $x_{i}(0,s_{-i})=0$ must hold almost surely. 
    Therefore, for any $s_{i}$ and $s_{-i}$, \eqref{eq: EPIC} implies
    \begin{align*}
        0
        &=\mE[\omega\cdot x_{i}(0,s_{-i})\mid s_{i},s_{-i}] \\
        &\leq \mE[\omega\cdot x_{i}(s_{i},s_{-i})\mid s_{i},s_{-i}]
        = x_{i}(s_{i},s_{-i})\cdot \frac{\prod_{k}s_{k} - \prod_{k}(1-s_{k})}{\prod_{k}s_{k} + \prod_{k}(1-s_{k})}.  
    \end{align*} 
    Therefore, $x_{i}(s_{i},s_{-i})=0$ must hold whenever $\lr{s_{i},s_{-i}}< 1$. 

    Hence, the efficient allocation mechanism $\mI\{\lr{s_{i},s_{-i}}\geq 1\}$ provides a higher probability of allocation for any signal profile. 
    Note also that the efficient mechanism clearly satisfies both \eqref{eq: P} and \eqref{eq: EPIC}.
    Therefore, it is uniquely optimal for the designer up to a measure-zero set. 
\end{proof}

\begin{remark}[Misaligned preference]\label{rem: BIC vs EPIC in welfare}
    As discussed in Subsection \ref{subsec: proof idea}, any optimal mechanism in the original problem must induce an indirect utility function that lies pointwise below $\overline{U}_{i}$, the indirect utility generated by the efficient mechanism.  
    Consequently, optimal outcomes are never Pareto ordered: the designer strictly prefers the optimal mechanism under  interim incentive compatibility, whereas agents prefer the one derived under ex-post incentive compatibility.
    \qed  
\end{remark}

\begin{proof}[Proof of Remark \ref{rem: interdependent value}]
    Fix any agent $i$ and consider any mechanism $x_{i}$. To simplify notation, we write $v_{i}=(s_{i},\epsilon_{i})$ to denote a type profile of each agent. For any type profile $v_{i}$ and report $\hat{v}_{i}$, her interim expected payoff is given by
    \begin{align*}
        & \mE[(\omega+\epsilon_{i})\cdot x_{i}(\hat{v}_{i},v_{-i})\mid s_{i}] \\
        & = s_{i}(1+\epsilon_{i})\cdot \mE[x_{i}(\hat{v}_{i},v_{-i})\mid \omega=+1]-(1-s_{i})(1-\epsilon_{i})\cdot \mE[x_{i}(\hat{v}_{i},v_{-i})\mid \omega=-1] \\
        & = ((1+\epsilon_{i})\cdot s_{i} + (1-\epsilon_{i})\cdot (1-s_{i}))\cdot U_{i}(\hat{v}_{i}; v_{i}), 
    \end{align*}
    where we define
    \begin{align*}
        U_{i}(\hat{v}_{i}; v_{i}) 
        = t_{i}(v_{i})\cdot \mE[x_{i}(\hat{v}_{i},v_{-i})\mid \omega=+1]-(1-t_{i}(v_{i}))\cdot \mE[x_{i}(\hat{v}_{i},v_{-i})\mid \omega=-1]. 
    \end{align*} 
    We claim that $U_{i}(v_{i}; v_{i})$ is measurable with respect to effective types $t_{i}$. In fact, assume that $t_{i}(v_{i})=t_{i}(\hat{v}_{i})$. If $U_{i}(v_{i}; v_{i})<U_{i}(\hat{v}_{i}; \hat{v}_{i})$, then the above expression and $U_{i}(\hat{v}_{i}; \hat{v}_{i})=U_{i}(\hat{v}_{i}; v_{i})$ implies a contradiction to incentive compatibility. 
    
    Let $U_{i}(t_{i})=U_{i}(v_{i}; v_{i})$ be the "normalized" indirect utility of each effective type $t_{i}=t_{i}(v_{i})$. As in the proof of Theorem \ref{thm: optimal mechanism}, we can see that if $x_{i}$ is optimal, then $U_{i}$ must be a convex function sandwiched between $\underline{U}_{i}$ and $\overline{U}_{i}$. 
    Therefore, the same computation as in the proof of Theorem \ref{thm: optimal mechanism} shows that there exists a monotone threshold mechanism which induces $U_{i}$ as the indirect utility function. Note that $\epsilon_{-i}$ is independent of $\omega$ and $s_{-i}$ and therefore $\lr{s_{-i},\epsilon_{-i}}=\lr{s_{-i}}$. 

    It remains to check that the objective function is measurable with respect to the normalized indirect utility functions $U_{i}$. 
    Note that $U_{i}$ is the supremum of affine functions, and the affine function generated by the truthful report has slope 
    \begin{align*}
        \mE[x_{i}(v_{i},v_{-i})\mid \omega=+1] + \mE[x_{i}(v_{i},v_{-i})\mid \omega=1],
    \end{align*}
    where $t_{i}=t_{i}(v_{i})$. Thus, at every point at which $U_{i}$ is differentiable,
    \begin{align*}
        U_{i}'(t_{i}(v_{i}))
        = \mE[x_{i}(v_{i},v_{-i})\mid \omega=+1] + \mE[x_{i}(v_{i},v_{-i})\mid \omega=1].
    \end{align*}
    Combining this equality with the definition of $U_{i}$ gives
    \begin{align*}
        \mE[x_{i}(v_{i},v_{-i})\mid \omega=+1]
        &= U_{i}(t_{i}) + (1-t_{i})U_{i}'(t_{i}),\\
        \mE[x_{i}(v_{i},v_{-i})\mid \omega=+1]
        &= t_{i}U_{i}'(t_{i}) - U_{i}(t_{i}),
    \end{align*}
    where $t_{i}=t_{i}(v_{i})$. Therefore, at almost all effective types, the interim allocation probabilities conditional on either state are uniquely determined by $U_{i}$. Hence, as in the proof of Theorem \ref{thm: optimal mechanism}, we can check that the allocation probability to agent $i$ is a linear functional of normalized indirect utility $U_{i}$, and in particular, any two mechanisms inducing the same indirect utility function imply the same allocation probability. 
\end{proof}

\section{Proof of Theorem \ref{thm: optimal mechanism log-concave}} 

We need to start with one preliminary. 
For any two signed measures $\mu$ and $\nu$ defined over an interval $[a,b]$, we say that $\mu$ dominates $\nu$ in \textit{convex order} if
\begin{align*}
    \int cd\mu \geq \int cd\nu
\end{align*}
for any convex functions $c:[a,b]\rightarrow \mR$. 
Then, for any two right-continuous functions $H$ and $G$ defined over $[a,b]$, we say that $G$ \textit{majorizes} $H$ if 
\begin{align*}
    \int_{a}^{x} H(y)dy \geq \int_{a}^{x} G(y)dy
\end{align*}
for each $x\in [a,b]$, with equality at $x=b$. 
Note that here we do not assume $H$ and $G$ are non-decreasing \citep{kleiner2021extreme}.  

The next lemma extends the equivalence between majorization and convex-order beyond probability distributions. 
We need $H(b)=G(b)$ for equivalence. 

\begin{lemma}\label{lem: majorization equivalence}
    Take any two signed measures $\mu$ and $\nu$ defined over an interval $[a,b]$ and define $H(x)=\mu([a,x])$ and $G(x)=\nu([a,x])$ for each $x\in [a,b]$. 
    Suppose $H(b)=G(b)$. 
    Then, $\mu$ dominates $\nu$ in convex order if and only if $G$ majorizes $H$. 
\end{lemma}

\begin{proof}[Proof of Lemma \ref{lem: majorization equivalence}]
    The proof is a simple extension of the one for probability distributions.\footnote{See, e.g., Theorem 3.A.1. in \citet{shaked2007stochastic} and the explanation therein.} 
    However, we do not find a reference for signed measures and therefore provide a proof for completeness. 

    To see the "only-if" part, note that for each $t\in \mR$, $\max\{0,x-t\}$ is a convex function of $x$. 
    Moreover, for any $x\in [a,b]$, we get 
    \begin{align*}
        \int_{a}^{x} H(t)dt
        &= \int_{a}^{x} \mu([a,t]) dt \\
        &= \int_{a}^{x} \int_{a}^{b}\mI\{s\leq t\}d\mu(s) dt \\
        &= \int_{a}^{b} \int_{a}^{x} \mI\{s\leq t\}dtd\mu(s) 
        = \int_{a}^{b} \max\{0,x-s\}d\mu(s). 
    \end{align*}
    By the assumption that $\mu$ dominates $\nu$ in convex order, the last term is larger under $\nu$ than under $\mu$. 
    Moreover, if we set $x=b$, then $\max\{0,x-s\}=b-s$ is linear in $s$. 
    Then, both $\max\{0,x-s\}$ and $-\max\{0,x-s\}$ are convex in $s$, which implies that the last term is the same under $\nu$ and $\mu$. 
    Hence, $G$ majorizes $H$. 

    To prove the "if" part, we use Theorem 1.6.3 in \citet{niculescu2006convex}, which shows that every convex function $c:[a,b]\rightarrow \mR$ can be expressed as
    \begin{align*}
        c(x) = \alpha + \beta x + \int_{a}^{b} \max\{0,x-t\}d\rho(t), 
    \end{align*}
    for some non-negative Borel measure $\rho$ on $(a,b)$ and $\alpha,\beta\in \mR$. 
    Then, integrating both sides by $\mu$, the equation in the above paragraph implies
    \begin{align*}
        \int cd\mu 
        &= \alpha H(b) + \beta\int_{a}^{b}xd\mu(x) + \int_{a}^{b} \int_{a}^{t} H(x)dx d\rho(t) \\
        &= \alpha H(b) + \beta\left\{bH(b)-\int_{a}^{b}H(x)dx\right\} + \int_{a}^{b} \int_{a}^{t} H(x)dx d\rho(t),
    \end{align*}
    where the second equation uses $x=-\max\{0,b-x\}+b$. 
    Therefore, $H(b)=G(b)$ and the definition of majorization imply that $\mu$ dominates $\nu$ in convex order. 
\end{proof}

Now, we move on to the main analysis. 
Define a signed measure $\mu_{i}$ such that 
\begin{align*}
    \mu_{i}(S_{i}) = \int_{S_{i}}g_{i}(s_{i})ds_{i} - \mI\{\underline{s}\in S_{i}\}  \cdot 2\underline{s}(1-\underline{s})f(\underline{s}) + \mI\{\overline{s}\in S_{i}\} \cdot 2\overline{s}(1-\overline{s})f(\overline{s}),
\end{align*}
for each measurable set $S_{i}\subset [\underline{s},\overline{s}]$. 
Then, from Lemma \ref{lem: objective function is linear in indirect utility}, the objective function is linear in indirect utility with the weight given by this measure. 
For each interval $I_{i}=[a_{i},b_{i}]$, set 
\begin{align*}
    G_{i}|_{I_{i}}(s_{i}) = \mu(\{t_{i}\in I_{i}\mid t_{i}\leq s_{i}\}),
\end{align*}
for each $s_{i}\in [a_{i},b_{i}]$, to be the cumulative function of the signed measure $\mu_{i}$ over the restricted domain $I_{i}$. 

Consider any $U_{i}\in \mathcal{U}_{i}^{*}$ that is piecewise linear, i.e., there exists a finite partition $\mathcal{S}_{i}$ of the signal space $[\underline{s},\overline{s}]$ into finite intervals and $U_{i}$ is linear over each $S_{i}\in \mathcal{S}_{i}$. 
We do not restrict a partition $\mathcal{S}_{i}$ to be maximal, i.e., we allow there to exist two adjacent intervals over which $U_{i}$ have the same functional form. 

The next lemma provides a sufficient condition under which $U_{i}$ is optimal. 
Both the result and the proof of the next lemma are largely inspired by Theorem 1 in \citet{kleiner2022optimal}.\footnote{Theorem 1 in \citet{kleiner2022optimal} is a characterization result. The conditions are not a necessary condition in our case however, because the set $\mathcal{U}_{i}^{*}$ involves more restrictions such as monotonicity and a lower bound relative to the set (convex functions that lie pointwise below a convex function $h$) over which \citet{kleiner2022optimal} optimizes. We extend the sufficiency part of his result to incorporate the lower bound $\underline{U}_{i}$. See also Theorem 2 in \citet{augias2025economics}, which provides a general sufficient condition for the optimality of a given function; for completeness, we include a concise proof below.} 

\begin{lemma}\label{lem: optimality certification}
    Let $U_{i}$ be a piecewise linear function with a partition $\mathcal{S}_{i}$. 
    Suppose that one of the following holds for each $S_{i}\in \mathcal{S}_{i}$.
    \begin{itemize}
        \item $U_{i}(s_{i})=\underline{U}_{i}(s_{i})$ and $\mu_{i}(\{s_{i}\})\leq 0$ for all $s_{i}\in S_{i}$.
        \item $U_{i}(s_{i})=\overline{U}_{i}(s_{i})$ at some $s_{i}\in S_{i}$,  $G_{i}|_{S_{i}}(\max S_{i})\geq 0$, and $G_{i}|_{S_{i}}$ majorizes $\delta_{s_{i}}$, which puts a point mass of mass $G_{i}|_{S_{i}}(\max S_{i})$ at $s_{i}$.  
        \item $\underline{U}_{i}(s_{i})\leq U_{i}(s_{i})<\overline{U}_{i}(s_{i})$ for every $s_{i}\in S_{i}$, $G_{i}|_{S_{i}}(\max S_{i})=0$, and $G_{i}|_{S_{i}}$ majorizes a zero measure. 
    \end{itemize}
    Then, $U_{i}$ is optimal. 
\end{lemma}

\begin{proof}[Proof of Lemma \ref{lem: optimality certification}]
    Suppose that $U_{i}$ is a piecewise linear function with a partition $\mathcal{S}_{i}$. 
    Take any interval $S_{i}\in \mathcal{S}_{i}$. 
    Assume first that $S_{i}=[a_{i},b_{i}]$ satisfies the second condition. 
    Let $s_{i}\in S_{i}$ be a point such that $U_{i}(s_{i})=\overline{U}_{i}(s_{i})$. 
    Then, by the definition of $\delta_{s_{i}}$, we have $\delta_{s_{i}}(\max S_{i})=G_{i}|_{S_{i}}(\max S_{i})$, and therefore, Lemma \ref{lem: majorization equivalence} implies that $\delta_{s_{i}}$ dominates $\mu|_{S_{i}}$ in convex order. 
    Then, for any other feasible indirect utility function $V_{i}\in \mathcal{U}_{i}^{*}$, we have 
    \begin{align*}
        \int_{a_{i}}^{b_{i}} V_{i}dG_{i}|_{S_{i}} 
        &\leq \int_{a_{i}}^{b_{i}} V_{i}d\delta_{s_{i}} \\
        &= V_{i}(s_{i})G_{i}|_{S_{i}}(\max S_{i}) \\
        &\leq \overline{U}_{i}(s_{i})G_{i}|_{S_{i}}(\max S_{i}) \\
        &= U_{i}(s_{i})G_{i}|_{S_{i}}(\max S_{i}) \\
        &= \int_{a_{i}}^{b_{i}} U_{i}d\delta_{s_{i}} 
        \leq \int_{a_{i}}^{b_{i}} U_{i}dG_{i}|_{S_{i}},
    \end{align*}
    where the first inequality follows by convex-order dominance, the second inequality follows because $\overline{U}_{i}$ is the pointwise maximizer in $\mathcal{U}_{i}^{*}$, and the second equality follows by assumption. 
    The last inequality follows with equality, noting that $U_{i}$ is linear and hence $-U_{i}$ is convex over the domain $S_{i}=[a_{i},b_{i}]$.

    Next, suppose that $S_{i}=[a_{i},b_{i}]$ satisfies the last condition instead of the second condition. 
    Then, $G_{i}|_{S_{i}}(\max S_{i})=0$. 
    Hence, we have the same sequence of inequalities as above, where the second inequality holds with equality. 

    Let $\mathcal{T}_{i}\subset \mathcal{S}_{i}$ be the subset of intervals that satisfy the conditions in the first case. 
    Set $I_{i}=\bigcup_{S_{i}\in \mathcal{S}_{i}\setminus\mathcal{T}_{i}}S_{i}$. 
    Then, together with the above discussions, for any $V_{i}\in \mathcal{U}_{i}^{*}$, 
    \begin{align*}
        \int_{\underline{s}}^{\overline{s}}V_{i}dG_{i}
        = \int_{\mathcal{T}_{i}}V_{i}dG_{i} + \int_{I_{i}}V_{i}dG_{i} 
        \leq \int_{\mathcal{T}_{i}}\underline{U}_{i}dG_{i} + \int_{I_{i}}U_{i}dG_{i} 
        = \int_{\underline{s}}^{\overline{s}}U_{i}dG_{i},
    \end{align*}
    which completes the proof. 
\end{proof}

Applying this lemma to a special case, we obtain the following. 

\begin{lemma}\label{lem: three-region partition optimal}
    Suppose that there exist two thresholds $\underline{s}\leq s_{i}^{\texttt{min}}\leq s_{i}^{\texttt{max}}\leq \overline{s}$ such that $g_{i}(s_{i})\leq 0$ for all $s_{i}\leq s_{i}^{\texttt{min}}$, $g_{i}(s_{i})\geq 0$ for all $s_{i}\geq s_{i}^{\texttt{max}}$, and 
    \begin{align*} 
        &\int_{s_{i}^{\texttt{min}}}^{s_{i}^{\texttt{max}}}g_{i}(t_{i})dt_{i} \\
        &- \mI\{s_{i}^{\texttt{min}}\leq \underline{s}\}  \cdot 2\underline{s}(1-\underline{s})f(\underline{s}) + \mI\{s_{i}^{\texttt{max}}\geq \overline{s}\} \cdot 2\overline{s}(1-\overline{s})f(\overline{s}) \geq 0, \label{eq: Sign}\tag{Sign} \\
        &\int_{s_{i}^{\texttt{min}}}^{s_{i}}(s_{i}-t_{i})g_{i}(t_{i})dt_{i} \\
        &- \mI\{s_{i}^{\texttt{min}}\leq \underline{s}\}  \cdot 2(s_{i}-\underline{s})\underline{s}(1-\underline{s})f(\underline{s}) \leq 0 \text{ for all } s_{i}\in [s_{i}^{\texttt{min}},s_{i}^{\texttt{max}}], \label{eq: MPC}\tag{Dominance} 
    \end{align*} 
    with equality at $s_{i}=s_{i}^{\texttt{max}}$. 
    If there exists $U_{i}\in \mathcal{U}_{i}^{*}$ that takes value $\underline{U}_{i}$ over $[\underline{s},s_{i}^{\texttt{min}}]$, is linear over $[s_{i}^{\texttt{min}},s_{i}^{\texttt{max}}]$, and coincides with $\overline{U}_{i}$ over $[s_{i}^{\texttt{max}},\overline{s}]$, the indirect utility function $U_{i}$ is optimal. 
\end{lemma}

\begin{proof}[Proof of Lemma \ref{lem: three-region partition optimal}]
    This follows immediately from Lemma \ref{lem: optimality certification} by taking a specific partition $\mathcal{S}_{i}=\{[\underline{s},1/2], [1/2,s_{i}^{\texttt{min}}],  [s_{i}^{\texttt{min}}, s_{i}^{\texttt{max}}], [s_{i}^{\texttt{max}},\overline{s}]\}$. 
    Here, we use
    \begin{align*}
        \int_{a_{i}}^{s_{i}}G_{i}|_{[a_{i},b_{i}]}(t_{i})dt_{i} 
        &= \int_{a_{i}}^{s_{i}}\int_{a_{i}}^{t_{i}}g_{i}(u_{i})du_{i}dt_{i}\\
        &= \int_{a_{i}}^{s_{i}}\int_{u_{i}}^{s_{i}}g_{i}(u_{i})dt_{i}du_{i}
        = \int_{a_{i}}^{s_{i}} (s_{i}-u_{i})g_{i}(u_{i})du_{i}, 
    \end{align*} 
    for any interval $[a_{i},b_{i}]$ that does not include the end point $\underline{s}$, where the second line changes the order of integrals. 
    A similar computation induces the associated condition when the interval includes the endpoint. 
\end{proof} 

From here, we focus on log-concave density functions. 
We start with one preliminary lemma. 

\begin{lemma}\label{lem: log-concave implies monotonicity}
    Suppose $f$ is log-concave. 
    Then, 
    \begin{align*}
        \frac{1}{2s_{i}(1-s_{i})f(s_{i})}\cdot g_{i}(s_{i})
        = -\frac{3}{2}\cdot \frac{1-2s_{i}}{s_{i}(1-s_{i})} - \frac{df(s_{i})/ds_{i}}{f(s_{i})}
    \end{align*}
    is increasing in $s_{i}\in (0,1)$. 
\end{lemma}

\begin{proof}[Proof of Lemma \ref{lem: log-concave implies monotonicity}] 
    Note that the first term on the right-hand side is decreasing in $s_{i}\in (0,1)$. 
    The result then follows immediately because $f$ being log-concave implies that the second term is decreasing. 
\end{proof}

\begin{lemma}\label{lem: three-region partition optimal log-concave case}
    Suppose $f$ is log-concave. 
    Then, there exist $\underline{s}\leq s_{i}^{\texttt{min}}< s_{i}^{\texttt{max}}\leq \overline{s}$ that satisfy the conditions in Lemma \ref{lem: three-region partition optimal} if and only if 
    \begin{align*}
        \int_{s_{i}^{\texttt{min}}}^{s_{i}^{\texttt{max}}}(s_{i}^{\texttt{max}}-t_{i})g_{i}(t_{i})dt_{i} 
        - \mI\{s_{i}^{\texttt{min}}\leq \underline{s}\}  \cdot 2(s_{i}^{\texttt{max}}-\underline{s})\underline{s}(1-\underline{s})f(\underline{s}) =0.
    \end{align*}
\end{lemma}

\begin{proof}[Proof of Lemma \ref{lem: three-region partition optimal log-concave case}]
    Suppose the equality in the statement holds.
    Let $\nu_{i}$ be the probability measure defined as 
    \begin{align*}
        \nu_{i}(\{t_{i}\}) = \dfrac{2t_{i}(1-t_{i})f(t_{i})}{\int_{s_{i}^{\texttt{min}}}^{s_{i}^{\texttt{max}}}2t_{i}(1-t_{i})f(t_{i})dt_{i}}.
    \end{align*} 
    Since $g_{i}(t_{i})/[2t_{i}(1-t_{i})f(t_{i})]$ is increasing by Lemma \ref{lem: log-concave implies monotonicity} and $s_{i}^{\texttt{max}}-t_{i}$ is decreasing in $t_{i}$, the Harris inequality implies that\footnote{It says that if a function $g$ is decreasing and $h$ is increasing, $\mE[gh]\leq \mE[g]\mE[h]$ for a given probability measure. One can check it by expanding $\mE[(g(x)-g(y))(h(x)-h(y))]\leq 0$ for two iid random variables $x$ and $y$.}  
    \begin{align*}
        0
        & \leq\mI\{s_{i}^{\texttt{min}}\leq \underline{s}\}  \cdot 2(s_{i}^{\texttt{max}}-\underline{s})\underline{s}(1-\underline{s})f(\underline{s}) \\
        &= \int_{s_{i}^{\texttt{min}}}^{s_{i}^{\texttt{max}}}(s_{i}^{\texttt{max}}-t_{i})g_{i}(t_{i})dt_{i}\\
        &= \nu([s_{i}^{\texttt{min}},s_{i}^{\texttt{max}}])\int_{s_{i}^{\texttt{min}}}^{s_{i}^{\texttt{max}}}(s_{i}^{\texttt{max}}-t_{i})\left\{\frac{1}{2t_{i}(1-t_{i})f(t_{i})}\cdot g_{i}(t_{i})\right\}d\nu_{i}(t_{i})\\
        &\leq \nu([s_{i}^{\texttt{min}},s_{i}^{\texttt{max}}])\int_{s_{i}^{\texttt{min}}}^{s_{i}^{\texttt{max}}}(s_{i}^{\texttt{max}}-t_{i})d\nu_{i}(t_{i}) \int_{s_{i}^{\texttt{min}}}^{s_{i}^{\texttt{max}}}\left\{\frac{1}{2t_{i}(1-t_{i})f(t_{i})}\cdot g_{i}(t_{i})\right\}d\nu_{i}(t_{i}) \\
        &\leq \nu([s_{i}^{\texttt{min}},s_{i}^{\texttt{max}}])(s_{i}^{\texttt{max}}-s_{i}^{\texttt{min}})
        \int_{s_{i}^{\texttt{min}}}^{s_{i}^{\texttt{max}}}\left\{\frac{1}{2t_{i}(1-t_{i})f(t_{i})}\cdot g_{i}(t_{i})\right\}d\nu_{i}(t_{i}) \\
        &= (s_{i}^{\texttt{max}}-s_{i}^{\texttt{min}})\int_{s_{i}^{\texttt{min}}}^{s_{i}^{\texttt{max}}}g_{i}(t_{i})dt_{i},
    \end{align*}
    which implies \eqref{eq: Sign} as $\underline{s}\leq s_{i}^{\texttt{min}}$. 

    To see \eqref{eq: MPC}, consider a function
    \begin{align*}
        h(t_{i}) = \frac{(s_{i}^{\texttt{max}}-t_{i})2t_{i}(1-t_{i})f(t_{i})}{\int_{s_{i}^{\texttt{min}}}^{s_{i}^{\texttt{max}}}
        (s_{i}^{\texttt{max}}-t_{i})2t_{i}(1-t_{i})f(t_{i})dt_{i}}.
    \end{align*}
    Note that $h$ is a density function. 
    To simplify notation, let $\Lambda>0$ be the denominator of the function $h$ in the above expression. 
    Then, we have 
    \begin{align*}
        & \int_{s_{i}^{\texttt{min}}}^{s_{i}}(s_{i}-t_{i})g_{i}(t_{i})dt_{i} \\
        &= \int_{s_{i}^{\texttt{min}}}^{s_{i}^{\texttt{max}}} \max\{0, s_{i}-t_{i}\}g_{i}(t_{i})dt_{i} \\
        &= \Lambda \int_{s_{i}^{\texttt{min}}}^{s_{i}^{\texttt{max}}} \frac{\max\{0, s_{i}-t_{i}\}}{s_{i}^{\texttt{max}}-t_{i}}\cdot \frac{g_{i}(t_{i})}{2t_{i}(1-t_{i})f(t_{i})}h(t_{i})dt_{i} \\ 
        &\leq \Lambda \int_{s_{i}^{\texttt{min}}}^{s_{i}^{\texttt{max}}} \frac{\max\{0, s_{i}-t_{i}\}}{s_{i}^{\texttt{max}}-t_{i}} h(t_{i})dt_{i}\cdot \int_{s_{i}^{\texttt{min}}}^{s_{i}^{\texttt{max}}}\frac{g_{i}(t_{i})}{2t_{i}(1-t_{i})f(t_{i})}h(t_{i})dt_{i} \\
        &\leq \Lambda \frac{s_{i}-\underline{s}}{s_{i}^{\texttt{max}}-\underline{s}}\cdot \int_{s_{i}^{\texttt{min}}}^{s_{i}^{\texttt{max}}}\frac{g_{i}(t_{i})}{2t_{i}(1-t_{i})f(t_{i})}h(t_{i})dt_{i} \\ 
        &= \frac{s_{i}-\underline{s}}{s_{i}^{\texttt{max}}-\underline{s}} \cdot \int_{s_{i}^{\texttt{min}}}^{s_{i}^{\texttt{max}}} (s_{i}^{\texttt{max}}-t_{i}) g_{i}(t_{i})dt_{i} \\
        &= \mI\{s_i^{\texttt{min}}\leq \underline{s}\} 2(s_{i}-\underline{s})\underline{s}(1-\underline{s})f(\underline{s}), 
    \end{align*}
    where the first inequality again uses Lemma \ref{lem: log-concave implies monotonicity} and the Harris inequality. 
\end{proof}


\begin{lemma} \label{lem: full support log concave}
Suppose that $f$ is log-concave on $[0,1]$ and has mean $1/2$. 
Then
\[
    \frac{1}{4} f\left(\frac{1}{2}\right) \geq \int_{1/2}^{1}(1-t_{i})f(t_{i})dt_{i}. 
\]
\end{lemma}
\begin{proof}[Proof of Lemma \ref{lem: full support log concave}]
    We prove a stronger inequality 
    \begin{align*}
        \frac{1}{2} f\left(\frac{1}{2}\right) \geq \int_{1/2}^{1}f(t_{i})dt_{i}.
    \end{align*} 
    Indeed, once this inequality holds, we have
    \begin{align*}
        \int_{1/2}^{1}(1-t_{i})f(t_{i})dt_{i}
        \leq \frac{1}{2}\int_{1/2}^{1}f(t_{i})dt_{i}
        \leq \frac{1}{4}f\left(\frac{1}{2}\right).
    \end{align*} 
    Assume that this inequality does not hold. We prove the argument by contradiction. 

    Let $p=(\log f)'(1/2)$. If $p\leq 0$, then we have $f'(1/2)\leq 0$. Since $f$ is log-concave and therefore $f'/f$ is decreasing, we have $f'(s_{i})\leq 0$ for all $s_{i}\geq 1/2$. This implies that $f(1/2)\geq f(s_{i})$ for all $s_{i}\geq 1/2$, hence the desired inequality. Hence, suppose $p>0$. 

    Since $\log f$ is concave, we have $\log f(x_{i}+1/2)-\log f(1/2)\leq px_{i}$ for all $x_{i}$. Hence 
    \begin{align*}
        f\left(\frac{1}{2}+x_{i}\right) \leq f\left(\frac{1}{2}\right)e^{px_{i}}, 
    \end{align*}
    for all $x_{i}\in [-1/2, 1/2]$. We therefore obtain
    \begin{align*}
        \int_{0}^{1/2}x_{i} f\left(\frac{1}{2}-x_{i}\right)dx_{i} \leq f\left(\frac{1}{2}\right) \cdot \int_{0}^{1/2}x_{i} e^{-px_{i}}dx_{i}.
    \end{align*}
    
    Now, define an auxiliary function $h(x_{i})=f(1/2+x_{i})/f(1/2)$ for $x_{i}\in[0,1/2]$. The above inequality implies $0\leq h(x_{i})\leq e^{px_{i}}$, while, by assumption, 
    \begin{align*}
        \int_{0}^{1/2}h(t_{i})dt_{i} = f\left(\frac{1}{2}\right)^{-1} \int_{1/2}^{1}f(t_{i})dt_{i} \geq \frac{1}{2}. 
    \end{align*}
    Now, consider minimizing the integral
    \begin{align*}
        \int_{0}^{1/2}t_{i}h(t_{i})dt_{i}
    \end{align*}
    among all functions $h$ with these two properties. 
    \citet{luenberger1997optimization} necessity theorem implies that a solution must maximize the Lagrangian, which implies that the solution $h^{*}$ places mass as far to the left as possible until $h^{*}$ puts a total mass of $1/2$. Hence
    \begin{align*}
        \int_{0}^{1/2}x_{i} f\left(\frac{1}{2}+x_{i}\right)dx_{i} > f\left(\frac{1}{2}\right) \int_{0}^{c} x_{i} e^{px_{i}}dx_{i}, \quad \text{where} \quad c=\frac{1}{p}\log\left(1+\frac{p}{2}\right).
    \end{align*}

    Finally, let $z=p/2>0$. Since $pc=\log(1+z)$, a straight calculation shows 
    \begin{align*}
        \int_{0}^{c} x_{i}e^{px_{i}}dx_{i} - \int_{0}^{1/2}x_{i}e^{-px_{i}}dx_{i} = \frac{D(z)}{p^{2}},
    \end{align*}
    where $D(z)=(1+z)\log(1+z)-z-1+(1+z)e^{-z}$. Note that $D(0)=0$, and 
    \begin{align*}
        D'(z)
        &=\left\{\log(1+z)+(1+z)\frac{1}{1+z}\right\}-1+e^{-z}+(1+z)(-e^{-z}) \\
        &=\log(1+z)-ze^{-z} \\
        &\geq \int_{0}^{z}\frac{1}{1+t}dt-\frac{z}{1+z} \\
        &> \int_{0}^{z}\frac{1}{1+z}dt-\frac{z}{1+z} \\
        &= 0.
    \end{align*}
    Therefore, $D(z)>0$ for every $z>0$. Then, together with the above bounds, we obtain 
    \begin{align*}
        \int_{0}^{1/2}x_{i} f\left(\frac{1}{2}+x_{i}\right)dx_{i} > \int_{0}^{1/2}x_{i} f\left(\frac{1}{2}-x_{i}\right)dx_{i}.
    \end{align*}
    This is equivalent to $\int_{0}^{1}(t_{i}-1/2)f(t_{i})dt_{i}>0$, which contradicts that $f$ has mean $1/2$. Therefore, the desired inequality follows.
\end{proof}

Finally, we are ready to prove Theorem \ref{thm: optimal mechanism log-concave}.
Note that, as discussed in the main section, for this proof to be valid, we do not need the full-support assumption. 
In fact, it is sufficient that the support is symmetric. 

\begin{proof}[Proof of Theorem \ref{thm: optimal mechanism log-concave}]
    Consider the class of mechanisms given in the statement:
    \begin{align*}
        x_{i}(s_{i},s_{-i})
        = 
        \begin{cases}
            0 \quad & \text{if} \quad s_{i} \leq s^{\texttt{min}}_{i}(\tau) \\
            \mI\{\texttt{LR}(s_{-i})\geq \tau\} \quad & \text{if} \quad s^{\texttt{min}}_{i}(\tau) \leq s_{i} \leq s^{\texttt{max}}_{i}(\tau)\\ 
            \mI\{\texttt{LR}(s_{i},s_{-i})\geq 1\} \quad & \text{if} \quad s^{\texttt{max}}_{i}(\tau)\leq s_{i}, 
        \end{cases}
    \end{align*}     
    Here, for each $\tau$, we define $s^{\texttt{min}}_{i}(\tau)$ and $s^{\texttt{max}}_{i}(\tau)$ such that the condition \eqref{eq: EV} is satisfied. 
    Specifically, we set 
    \begin{align*}
        s^{\texttt{min}}_{i}(\tau) &= \frac{\int_{s_{-i}} \prod_{k\neq i}(1-s_{k}) \mI\{\lr{s_{-i}}\geq \tau\} d\mF^{\otimes}(s_{-i})}{\int_{s_{-i}} \left\{\prod_{k\neq i}s_{k}+\prod_{k\neq i}(1-s_{k})\right\}\mI\{\lr{s_{-i}}\geq \tau\} d\mF^{\otimes}(s_{-i})}\\
        s^{\texttt{max}}_{i}(\tau) &= \frac{1}{1+\tau}.
    \end{align*} 
    We can see that $(s^{\texttt{min}}_{i}(\tau),s^{\texttt{max}}_{i}(\tau))$ are continuous in $\tau$ by the dominated convergence theorem, as in the proof of Lemma \ref{lem: upper bound is differentiable}. 
    By construction, the induced utility function $U_{i}(s_{i};\tau)$ is in the set $\mathcal{U}_{i}^{*}$. 
    Moreover, $U_{i}(s_{i};\tau)$ equals $0$ over $[0,s^{\texttt{min}}_{i}(\tau)]$, is linear over $[s^{\texttt{min}}_{i}(\tau),s^{\texttt{max}}_{i}(\tau)]$, and coincides with $\overline{U}_{i}$ in the remaining interval. 

    Therefore, it remains to show that there exists a threshold $\tau\in [0,1]$ such that 
    \begin{align*}
        s^{\texttt{min}*}_{i}(\tau) = \max\{s^{\texttt{min}}_{i}(\tau), \underline{s}\}
        \quad \text{and} \quad 
        s^{\texttt{max}*}_{i}(\tau) = \min\{s^{\texttt{max}}_{i}(\tau), \overline{s}\}  
    \end{align*}
    satisfy the equality condition in the statement of Lemma \ref{lem: three-region partition optimal log-concave case}. 
    At $\tau=0$, we have $(s^{\texttt{min}}_{i}(\tau),s^{\texttt{max}}_{i}(\tau))=(1/2,1)$. 
    Then, noting that $\underline{s}\leq 1/2$ must hold by Bayes' plausibility $\mE[s_{i}]=1/2$, 
    \begin{align*}
        & \int_{s^{\texttt{min}*}_{i}(\tau)}^{s^{\texttt{max}*}_{i}(\tau)}(s^{\texttt{max}*}_{i}(\tau)-t_{i})g_{i}(t_{i})dt_{i} \\
        &= \int_{1/2}^{\overline{s}}(\overline{s}-t_{i})g_{i}(t_{i})dt_{i} \\
        &= -\int_{1/2}^{\overline{s}}(\overline{s}-t_{i})\left\{3(1-2t_{i})f(t_{i})+2t_{i}(1-t_{i})\frac{df(t_{i})}{dt_{i}}\right\}dt_{i} \\ 
        &= -\int_{1/2}^{\overline{s}}3(\overline{s}-t_{i})(1-2t_{i})f(t_{i})dt_{i}  +\left(\overline{s}-\frac{1}{2}\right)\frac{1}{2}\cdot f\left(\frac{1}{2}\right) \\
        &\quad \quad \quad + 2\int_{1/2}^{\overline{s}}(3t_{i}^{2}-2(1+\overline{s})t_{i}+\overline{s})f(t_{i})dt_{i} \\
        &= \left(\overline{s}-\frac{1}{2}\right)\frac{1}{2}\cdot f\left(\frac{1}{2}\right) - \int_{1/2}^{\overline{s}}(\overline{s}-(2\overline{s}-1)t_{i})f(t_{i})dt_{i}, 
    \end{align*}
    where the second line uses integration by parts. 

    Note that when $\overline{s}=1$, then by Lemma \ref{lem: full support log concave}, the last expression is positive. 
    In general, given that $t_{i}$ follows $\mF$, prepare a new random variable
    \begin{align*}
        z_{i} = \frac{1}{\overline{s}-\underline{s}}\cdot t_{i} - \frac{\underline{s}}{\overline{s}-\underline{s}}.
    \end{align*}
    Then, $z_{i}$ takes values in $[0,1]$. Moreover, using $\overline{s}=1-\underline{s}$, we can check that $z_{i}$ has mean $1/2$. Note that $z_{i}$ has density function $f^{*}(z_{i})= f(\underline{s}+(\overline{s}-\underline{s})z_{i})$, 
    which is clearly log-concave. Therefore, Lemma \ref{lem: full support log concave} applies, hence 
    \begin{align*}
        \frac{1}{4} f^{*}\left(\frac{1}{2}\right) \geq \int_{1/2}^{1}(1-z_{i})f^{*}(z_{i})dz_{i}. 
    \end{align*}
    Changing variables and rewriting it as an expression in $t_{i}$, we can observe that, the integral in the above paragraph is weakly positive. 



    Next, consider the case $\tau=1$, wherein $s^{\texttt{min}}_{i}(\tau)<1/2$ and $s^{\texttt{max}}_{i}(\tau)=1/2$. 
    In this case, note that if $f$ is log-concave and $f'(1/2)\geq 0$, then, $g_{i}(t_{i})\leq 0$ for all $t_{i}\in [0,1/2]$. 
    Therefore, we have
    \begin{align*}
        \int_{s_{i}^{\texttt{min}*}(\tau)}^{s_{i}^{\texttt{max}*}(\tau)}(s_{i}^{\texttt{max}}-t_{i})g_{i}(t_{i})dt_{i}\leq 0.
    \end{align*} 

    Recall from the first paragraph that $(s^{\texttt{min}}_{i}(\tau),s^{\texttt{max}}_{i}(\tau))$ are continuous in $\tau$. 
    Therefore, so are $(s^{\texttt{min}*}_{i}(\tau),s^{\texttt{max}*}_{i}(\tau))$ and the integral we consider. 
    Hence, by the intermediate value theorem, there exists $\tau\in [0,1]$ that satisfies the equality in Lemma \ref{lem: three-region partition optimal log-concave case} under $(s^{\texttt{min}*}_{i}(\tau),s^{\texttt{max}*}_{i}(\tau))$. 
    Hence, Lemma \ref{lem: three-region partition optimal} completes the proof.  
\end{proof}

We conclude this section by characterizing the optimal mechanism under uniform distribution. 

\begin{proof}[Deriving the optimal mechanism in Figure \ref{fig: uniform}]
    Suppose that $n=2$ and $\mF(s_{i})=s_{i}$ for all $s_{i}\in (0,1)$. 
    Then, by setting the two parameters as described, $s_{i}^{\texttt{min}}=3/8$ and $s_{i}^{\texttt{max}}=3/4$, the left-hand side of \eqref{eq: MPC} is written as
    \begin{align*}
        -\int_{3/8}^{s_{i}}\left(s_{i}-t_{i}\right)\cdot \left\{3(1-2t_{i})\right\}dt_{i}
        &= -3 \int_{3/8}^{s_{i}}\left\{2t_{i}^{2}-\left(1+2s_{i}\right)t_{i}+s_{i}\right\}dt_{i} \\
        &= -3 \left[\frac{2}{3}t_{i}^{3} - \left(\frac{1}{2}+s_{i}\right)t_{i}^{2}+ s_{i}t_{i}\right]^{s_{i}}_{3/8} \\
        & = \left(s_{i}-\frac{3}{4}\right)\left(s_{i}-\frac{3}{8}\right)^{2},
    \end{align*}
    for each $s_{i}$. 
    The last expression is negative for $s_{i}\in (3/8,3/4)$ and equals zero at $s_{i}\in \{3/8,3/4\}$. 
    Moreover, the left-hand side of \eqref{eq: Sign} is given by
    \begin{align*}
        -3\int_{3/8}^{3/4}(1-2t_{i})dt_{i} = \frac{9}{64} \geq 0. 
    \end{align*} 
    
    Finally, we need to check that the function $U_{i}$ specified as in the statement of Lemma \ref{lem: three-region partition optimal} is in $\mathcal{U}_{i}^{*}$. 
    Every condition except convexity is trivial. 
    Note that for $n=2$ with uniform distribution, 
    \begin{align*}
        \overline{U}_{1}(s_{1}) 
        = 2\int_{s_{2}} (s_{1}+s_{2}-1)\mI\{s_{2}\geq 1-s_{1}\}ds_{2} 
        = s_{1}^{2},
    \end{align*}
    where the first equation uses the expression in the right-hand side of \eqref{eq: EV}. 
    Therefore, the slope of $\overline{U}_{1}$ at $3/4$ equals $3/2$, which coincides with the slope of a linear function that connects the two points $(3/8,0)$ and $(3/4,9/16)$. 
    Therefore, $U_{i}$ has increasing slopes and thus is convex, and by Lemma \ref{lem: three-region partition optimal}, the optimal mechanism is given as in the figure. 
\end{proof}

\section{Proof of Theorem \ref{prop: large market}}\label{sec: large market proof}

We start with the following preliminary result. 
For clarity, we add a superscript and let $\overline{U}_{i}^{n}(s_{i})=\mE[\omega\cdot \mI\{\lr{s_{i},s_{-i}}\geq 1\}\mid s_{i}]$ denote the agent's first-best payoff when the number of agents in the market is $n$.

\begin{lemma}\label{lem: upper bound converges uniformly}
    The upper bound $\overline{U}_{i}^{n}$ is pointwise increasing in $n$ and converges uniformly to $\overline{U}_{i}^{\infty}$. 
\end{lemma}

\begin{proof}[Proof of Lemma \ref{lem: upper bound converges uniformly}]
    Note that $\overline{U}_{i}^{n}(s_{i})$ is the maximum interim payoff the agent $i$ can obtain among all mechanisms. 
    Moreover, for a market with $n+1$ agents, the agent $i$ can obtain the interim payoff $\overline{U}_{i}^{n}(s_{i})$ under the mechanism given by $x_{i}(s_{i},s_{-i})=\mI\{\lr{s_{-j}}\geq 1\}$ for some $j\neq i$. 
    Hence, $\overline{U}_{i}^{n}(s_{i})\leq \overline{U}_{i}^{n+1}(s_{i})$ for all $n$, proving the first statement. 

    To prove the second statement, note that, for each $n\in \mN$, $\overline{U}_{i}^{n}$ is continuous by Lemma \ref{lem: upper bound is differentiable}. 
    Moreover, for every $s_{i}$, the first-best payoff $\overline{U}_{i}^{n}(s_{i})$ converges to $\overline{U}_{i}^{\infty}(s_{i})$ because in the limit $n\rightarrow\infty$, we have $\lr{s_{i},s_{-i}}\geq 1$ if and only if $\omega=+1$. 
    Therefore, $(\overline{U}_{i}^{n})_{n\in \mN}$ is an increasing sequence of continuous functions that converges pointwise to $\overline{U}_{i}^{\infty}$ and $\overline{U}_{i}^{\infty}$ is also continuous.
    Hence, Dini's theorem shows the second statement. 
\end{proof}

Now, consider the designer's maximization problem in the large market $n\rightarrow \infty$. 
Recall that $V_{\infty}$ is the limit of $V_{n}$, which is the maximum value the designer can obtain from agent $i$ in finite market $n\in\mN$. 

\begin{lemma}\label{lem: limit problem berge}
    Let $U_{i}$ maximize the linear functional characterized in Lemma \ref{lem: objective function is linear in indirect utility} over the set of feasible indirect utility functions 
    \begin{align*}
        \mathcal{U}_{i}^{*\infty} = \left\{U_{i}:[0,1]\rightarrow [0,1]\ \middle|\ U_{i} \text{ is increasing, convex, and } \underline{U}_{i}\leq U_{i}\leq \overline{U}_{i}^{\infty} \right\}.
    \end{align*} 
    Then, $V_{\infty}$ equals the designer's objective value under $U_{i}$. 
\end{lemma}

\begin{proof}[Proof of Lemma \ref{lem: limit problem berge}]
    This is a simple application of Berge's maximum theorem. 
    For each finite $n\in \mN$, we have
    \begin{align*}
        V_{n} = \max_{U_{i}} \int U_{i} d\mu 
        \text{ s.t. } U_{i}\in \mathcal{U}_{i}^{*n},
    \end{align*}
    where $\mu$ is a signed measure defined in the proof of Theorem \ref{thm: optimal mechanism log-concave} and $\mathcal{U}_{i}^{*n}$ is the set that replaces $\overline{U}_{i}^{\infty}$ with $\overline{U}_{i}^{n}$ in the definition of $\mathcal{U}_{i}^{*\infty}$. 
    Endow the space $\mathcal{C}$ of increasing and convex functions $U_{i}:[0,1]\to[0,1]$ with the supremum norm and the associated distance. 
    
    The objective function is a linear functional of indirect utility. 
    Then, since the density $f$ has uniformly bounded derivative by assumption, the linear functional is also bounded in supremum norm, and is therefore continuous. 
    Moreover, Lemma \ref{lem: upper bound converges uniformly} implies that $\mathcal{U}_{i}^{*n}$ converges to $\mathcal{U}_{i}^{*\infty}$ in Hausdorff distance.
    Therefore, Berge's maximum theorem implies that the limit of $V_{n}$ equals the designer's objective value under $U_{i}$. 
    This completes the proof because $V_{n}\rightarrow V_{\infty}$ as $n\rightarrow\infty$ by definition. 
\end{proof}

Then, we characterize the extreme points of the set $\mathcal{U}_{i}^{*\infty}$.

\begin{lemma}\label{lem: extreme point large market}
    If $U_{i}$ is an extreme point of $\mathcal{U}_{i}^{*\infty}$, then, there exist two thresholds $s^{\texttt{min}}_{i}$ and $s^{\texttt{max}}_{i}$ such that $U_{i}$ is linear over $[s^{\texttt{min}}_{i},s^{\texttt{max}}_{i}]$ and coincide with $\underline{U}_{i}(s_{i})=\max\{0,2s_{i}-1\}$ in the remaining region. 
\end{lemma}

\begin{proof}[Proof of Lemma \ref{lem: extreme point large market}]
    Let $U_{i}$ be any extreme point of the set $\mathcal{U}_{i}^{*\infty}$. 
    By Lemma \ref{lem: extreme points}, there exists a countable collection of non-singleton intervals $\mathcal{T}_{i}$ that satisfies the conditions in the statement of Lemma \ref{lem: extreme points}. 
    
    Note that $\overline{U}_{i}^{\infty}(s_{i})=s_{i}$ is linear.
    Suppose $U_{i}(s_{i})$ coincides with $\overline{U}_{i}^{\infty}(s_{i})=s_{i}$ at some point $s_{i}\in (0,1)$. 
    Then, for any small $\epsilon_{i}\in \mR$, convexity of $U_{i}$ requires
    \begin{align*}
        U_{i}(s_{i})
        &\leq \frac{1}{2}U_{i}(s_{i}-\epsilon_{i})+\frac{1}{2}U_{i}(s_{i}+\epsilon_{i}) \\
        &\leq \frac{1}{2}\overline{U}_{i}^{\infty}(s_{i}-\epsilon_{i})+\frac{1}{2}\overline{U}_{i}^{\infty}(s_{i}+\epsilon_{i}) 
        = s_{i}.
    \end{align*}
    By assumption, $U_{i}(s_{i})=s_{i}$, and hence, every inequality holds with equality. 
    Hence, since $U_{i}$ is weakly pointwise below $\overline{U}_{i}^{\infty}$, we must have $U_{i}=\overline{U}_{i}^{\infty}$. 
    Then, we complete the proof by setting $s^{\texttt{min}}_{i}=0$ and $s^{\texttt{max}}_{i}=1$. 
    Hence, assume that $U_{i}(s_{i})<\overline{U}_{i}^{\infty}(s_{i})$ for all $s_{i}\in (0,1)$. 

    Then, the first condition in Lemma \ref{lem: extreme points} implies that we must have $U_{i}(s_{i})=\underline{U}_{i}(s_{i})$ for all $s_{i}\notin \bigcup_{T_{i}\in \mathcal{T}_{i}}T_{i}$. 
    Moreover, condition (a) in the second condition never holds, and therefore, for all group $T_{i}\in \mathcal{T}_{i}$, condition (b) in the second condition holds.
    Hence, for every interval $T_{i}\in \mathcal{T}_{i}$, the function $U_{i}$ is linear over $T_{i}$, lies strictly between $\underline{U}_{i}$ and $\overline{U}_{i}$ in the interior of $T_{i}$, and coincides with $\underline{U}_{i}$ at the end points. 
    The shape of $\underline{U}_{i}$ then implies that $\mathcal{T}_{i}$ must be a singleton having a unique element of a form $[s^{\texttt{min}}_{i},s^{\texttt{max}}_{i}]$. 
    This finishes the proof. 
\end{proof} 

\begin{proof}[Proof of Theorem \ref{prop: large market}] 
    Take any indirect utility function $U_{i}$ that maximizes the objective function given by Lemma \ref{lem: objective function is linear in indirect utility} subject to the constraint $\mathcal{U}_{i}^{*\infty}$. 
    Then, Bauer's maximum principle implies that an extreme point of $\mathcal{U}_{i}^{*\infty}$ is a solution to this optimization problem. 
    Then, Lemma \ref{lem: extreme point large market} implies that $U_{i}$ is linear over $I_{i}=[s^{\texttt{min}}_{i},s^{\texttt{max}}_{i}]$ and coincide with $\underline{U}_{i}(s_{i})=\max\{0,2s_{i}-1\}$ in the remaining region. 
    Moreover, $V_{\infty}$ is the value of $U_{i}$. 

    By assumption, we have $U_{i}(s_{i})=a_{i}s_{i}-b_{i}$ over $[s^{\texttt{min}}_{i},s^{\texttt{max}}_{i}]$ for some $a_{i}\geq 0$ and $b_{i}\geq 0$. 
    If $s^{\texttt{min}}_{i}> 1/2$, then $U_{i}=\underline{U}_{i}$. 
    Hence, assume without loss of generality that $s^{\texttt{min}}_{i}\leq 1/2$.
    Then, $U_{i}\leq \overline{U}_{i}^{\infty}$ and $s^{\texttt{min}}_{i}\leq 1/2$ requires $a_{i}\in [0,2]$. 

    Now, consider a finite market with size $n\in \mN$. 
    Then, we construct a mechanism as follows, depending on one of two cases.  
    First, assume $U_{i}\in \mathcal{U}_{i}^{*n}$. 
    Then, by Lemma \ref{lem: threshold mechanism induces any linear utility}, there exists $\kappa(n)$ and $\tau(n)$ such that the mechanism
    \begin{align*}
        x_{i}(s_{i},s_{-i};n) 
        = 
        \begin{cases}
            0 \quad & \text{if} \quad s_{i} \leq s^{\texttt{min}}_{i} \\
            \kappa(n)\cdot \mI\{\lr{s_{-i}}\geq \tau(n)\} \quad & \text{if} \quad s^{\texttt{min}}_{i} \leq s_{i}\leq s^{\texttt{max}}_{i}, \\
            1 \quad & \text{if} \quad s^{\texttt{max}}_{i}\leq s_{i}, 
        \end{cases}
    \end{align*} 
    implements $U_{i}$. 
    We can check this mechanism is feasible as in the proof of Theorem \ref{thm: optimal mechanism countable partition}. 

    Second, suppose that $U_{i} \notin \mathcal{U}_{i}^{*n}$.
    Since $U_{i}$ lies pointwise above $\underline{U}_{i}$ and coincides with $\underline{U}_{i}$ on $[0,1]\setminus I_{i}$, it follows that the line segment $a_{i}s_{i}-b_{i}$ lies above $\overline{U}_{i}^{n}$ at some point in $I_{i}$.
    Let $b_{i}(n)\geq b_{i}$ be the smallest value such that $a_{i}s_{i}-b_{i}(n)$ lies pointwise below $\overline{U}_{i}^{n}$.
    Then, let $U_{i}^{n}$ be the upper envelope of $\underline{U}_{i}$ and $a_{i}s_{i}-b_{i}(n)$, i.e., $U_{i}^{n}(s_{i})=\max\{\underline{U}_{i}(s_{i}),a_{i}s_{i}-b_{i}(n)\}$. 
    
    Then, Lemma \ref{lem: threshold mechanism induces any linear utility} applies to $U_{i}^{n}$, which shows that a mechanism of a form
    \begin{align*}
        x_{i}(s_{i},s_{-i};n) 
        = 
        \begin{cases}
            0 \quad & \text{if} \quad s_{i} \leq s^{\texttt{min}}_{i}(n) \\
            \mI\{\lr{s_{-i}}\geq \tau(n)\} \quad & \text{if} \quad s^{\texttt{min}}_{i}(n) \leq s_{i}\leq s^{\texttt{max}}_{i}(n), \\
            1 \quad & \text{if} \quad s^{\texttt{max}}_{i}(n)\leq s_{i}, 
        \end{cases}
    \end{align*} 
    implements $U_{i}^{n}$. 
    Here, $s^{\texttt{min}}_{i}(n)$ and $s^{\texttt{max}}_{i}(n)$ are kinks of $U_{i}^{n}$. 
    Again, this mechanism is feasible. 
    By construction, $U_{i}^{n}$ touches $\overline{U}_{i}^{n}$ at some point in $I_{i}$. 

    Here, for each market size $n\in \mN$, let $U_{i}^{n}$ be the indirect utility function induced by the mechanism constructed above. 
    If $U_{i}\in \mathcal{U}_{i}^{*N}$ for some $N$, then $U_{i}\in \mathcal{U}_{i}^{*n}$ for all $n\geq N$ by Lemma \ref{lem: upper bound converges uniformly}, and therefore, $U_{i}^{n}$ trivially converges to $U_{i}$ in supremum norm. 
    If $U_{i}\notin \mathcal{U}_{i}^{*N}$ for all $N$, then, by construction, supremum-norm distance between $U_{i}$ and $U_{i}^{n}$ is $b_{i}(n)-b_{i}$. 
    Since $U_{i}\in \mathcal{U}_{i}^{*\infty}$ and $\overline{U}_{i}^{n}$ converges uniformly to $\overline{U}_{i}^{\infty}$ by Lemma \ref{lem: upper bound converges uniformly}, we must have $b_{i}(n)-b_{i}\rightarrow 0$ as $n\rightarrow \infty$. 
    Therefore, $U_{i}^{n}$ converges to $U_{i}$ in supremum norm as $n\rightarrow \infty$. 

    Finally, note that the objective function is continuous in indirect utility, as discussed in the proof of Lemma \ref{lem: limit problem berge}.
    Therefore, $\mE[x_{i}(s;n)]$ converges, as $n\rightarrow \infty$, to the value induced by $U_{i}$, which equals $V_{\infty}$.
    Since both $\mE[x_{i}(s;n)]$ and $V_{n}$ converge to the same limit $V_{\infty}$, the difference between the two sequences also converges to zero. 
\end{proof}

Finally, we prove our claim in Remark \ref{rem: large-market limit log-concave and symmetric case}. 

\begin{lemma}\label{lem: large-market limit log-concave and symmetric case}
    If $f$ is log-concave and symmetric around the prior $1/2$, then, for the family of mechanisms defined in Theorem \ref{prop: large market}, we have $s_{i}^{\texttt{max}}(n)\rightarrow 1$ as $n\rightarrow \infty$.
\end{lemma} 

\begin{proof}[Proof of Lemma \ref{lem: large-market limit log-concave and symmetric case}]
    Let $U_{i}$ solve the optimization problem in Lemma \ref{lem: limit problem berge}. 
    Then, by Lemme \ref{lem: extreme point large market}, we have
    \begin{align*}
        U_{i}(s_{i})=
        \begin{cases}
            0 & \text{ if } s_{i}\leq s^{\texttt{min}}_{i} \\
            \dfrac{2s^{\texttt{max}}_{i}-1}{s^{\texttt{max}}_{i}-s^{\texttt{min}}_{i}}\cdot (s_{i}-s^{\texttt{min}}_{i}) & \text{ if } s^{\texttt{min}}_{i} \leq s_{i} \leq  s^{\texttt{max}}_{i}\\
            2s_{i}-1 & \text{ if } s^{\texttt{max}}_{i} \leq s_{i}.
        \end{cases}
    \end{align*}
    for some $s_{i}^{\texttt{min}}$ and $s_{i}^{\texttt{max}}$. 
    Then, as we can see from the proof of Theorem \ref{prop: large market}, it is sufficient to show that $s_{i}^{\texttt{max}}=1$. 

    Here, consider the following function:
    \begin{align*}
        W_{i}(s_{i})=
        \begin{cases}
            0 & \text{ if } s_{i}\leq b_{i}/a_{i} \\
            a_{i}s_{i}-b_{i} & \text{ if } b_{i}/a_{i} \leq s_{i},
        \end{cases}        
    \end{align*}
    where $a_{i}=2\cdot (1-U_{i}(1/2))$ and $b_{i}=a_{i}-1$. 
    Intuitively, we construct $W_{i}$ from $U_{i}$ by taking the line segment over the middle region $[s^{\texttt{min}}_{i},s^{\texttt{max}}_{i}]$ and rotating it counterclockwise about $(1/2, U_{i}(1/2))$ until it passes through $(1,1)$. 

    By construction, $W_{i}(s_{i})\geq U_{i}(s_{i})$ for all $s_{i}\geq 1/2$ and $W_{i}(s_{i})\leq U_{i}(s_{i})$ for all $s_{i}\leq 1/2$. 
    Now recall the objective function:
    \begin{align*}
        \mE\left[x_{i}(s)\right]
        &= - \int_{\underline{s}}^{\overline{s}}\left\{3(1-2s_{i})f(s_{i})+2s_{i}(1-s_{i})\frac{df(s_{i})}{ds_{i}}\right\}U_{i}(s_{i})ds_{i} \\
        &\quad \quad \quad \quad +2\overline{s}(1-\overline{s})f(\overline{s})U_{i}(\overline{s})-2\underline{s}(1-\underline{s})f(\underline{s})U_{i}(\underline{s}). 
    \end{align*} 
    If $f$ is log-concave and symmetric around the prior $1/2$, we have $f'(s_{i})\geq 0$ for $s_{i}\leq 1/2$ and $f'(s_{i})\leq 0$ for $s_{i}\geq 1/2$. 
    Therefore, the objective function is strictly increasing in $U_{i}(s_{i})$ for $s_{i}> 1/2$ and strictly decreasing in $U_{i}(s_{i})$ for all $s_{i}< 1/2$. 
    Hence, unless $W_{i}=U_{i}$, $W_{i}$ results in a strictly higher expected payoff to the designer than $U_{i}$. 
    Therefore, we must have $W_{i}=U_{i}$, and in particular, $s_{i}^{\texttt{max}}=1$. 
\end{proof}

\section{Proofs of Proposition \ref{prop: laissez-faire outcome} and Corollary \ref{cor: social learning}}

\begin{proof}[Proof of Proposition \ref{prop: laissez-faire outcome}]
    Recall that $g_{i}(t_{i})=-3(1-2t_{i})f(t_{i})-2t_{i}(1-t_{i})f'(t_{i})$. Since $\log f$ is weakly increasing, we have $g_{i}(t_{i})\leq 0$ for all $t_{i}\leq 1/2$. Moreover, since $\log f$ is weakly concave on $[1/2,1]$, the proof of Lemma \ref{lem: log-concave implies monotonicity} shows that $g_{i}$ satisfies a single-crossing property on $[1/2,1)$. That is, there exists $t^{*}\in [1/2,1)$ such that $g_{i}(t_{i})\leq 0$ if and only if $t_{i}\in [1/2,t^{*}]$. Combining these observations, $g_{i}$ satisfies the single-crossing property on the entire domain $(0,1)$. 

    It is enough to check that every condition in Lemma \ref{lem: three-region partition optimal} holds for some $s_{i}^{\texttt{min}}\geq 1/2$ and $s_{i}^{\texttt{max}}=1$.
    Note that integration by parts gives
    \begin{align*}
        \int_{s_{i}^{\texttt{min}}}^{1}(1-t_{i})g_{i}(t_{i})dt_{i}
        = 2(1-s_{i}^{\texttt{min}})^{2}s_{i}^{\texttt{min}}f(s_{i}^{\texttt{min}})-\int_{s_{i}^{\texttt{min}}}^{1}(1-t_{i})f(t_{i})dt_{i}.
    \end{align*} 
    We take $s_{i}^{\texttt{min}}\geq 1/2$ such that the left-hand side equals zero. Indeed, the right-hand side is negative at $s_{i}^{\texttt{min}}=1/2$ by assumption, and the left-hand side is positive when $s_{i}^{\texttt{min}}\geq t^{*}$, and therefore, the intermediate value theorem guarantees that such a point $s_{i}^{\texttt{min}}$ exists in $[1/2,t^{*}]$.  
    
    First, since $t^{*}\geq 1/2$, it is obvious from the single-crossing property that $g_{i}(t_{i})\leq 0$ for all $t_{i}\leq s_{i}^{\texttt{min}}$ and $g_{i}(t_{i})\geq 0$ for all $t_{i}\geq s_{i}^{\texttt{max}}$. 
    Second, to see the sign condition \eqref{eq: Sign}, note that
    \begin{align*}
        \int_{s_{i}^{\texttt{min}}}^{1}g_{i}(t_{i})dt_{i}
        &= 2s_{i}^{\texttt{min}}(1-s_{i}^{\texttt{min}})f(s_{i}^{\texttt{min}})+\int_{s_{i}^{\texttt{min}}}^{1}(2t_{i}-1)f(t_{i})dt_{i}
        \geq 0, 
    \end{align*}
    where the first equality follows from integration by parts, and the last inequality uses $s_{i}^{\texttt{min}}\geq 1/2$. 
    Finally, note that the construction of $s_{i}^{\texttt{min}}$ implies that \eqref{eq: MPC} holds with equality at the endpoint $s_{i}=1$. Then, the single-crossing property of $g_{i}$ implies that every inequality condition is also satisfied, as seen in the proof of Lemma \ref{lem: three-region partition optimal log-concave case}. 
\end{proof}

\begin{proof}[Proof of Corollary \ref{cor: social learning}] 
    Let $x$ be a social-learning outcome under any network structure $\mathcal{B}$. It is enough to check that the allocation function $x$ satisfies both incentive compatibility and the participation constraint. 

    Consider any agent $i$. Since $x_{i}$ is induced by a Bayesian Nash equilibrium, agent $i$ weakly prefers her equilibrium strategy to any alternative strategy contingent on what she observes. In particular, for any $\hat{s}_{i}$, she can follow the equilibrium strategy prescribed for type $\hat{s}_{i}$ at every possible history. Because the actions observed by agent $i$ are taken by her predecessors, their distribution is unaffected by this deviation. The resulting probability of accepting the good is therefore $x_{i}(\hat{s}_{i},s_{-i})$. Hence,
    \begin{align*}
        \mE[\omega x_{i}(s_{i},s_{-i})\mid s_{i}]
        \geq
        \mE[\omega x_{i}(\hat{s}_{i},s_{-i})\mid s_{i}],
    \end{align*}
    for every $s_{i}$ and $\hat{s}_{i}$, establishing \eqref{eq: IC}. Moreover, agent $i$ can always reject the good and obtain zero. Therefore \eqref{eq: P} must hold. 
\end{proof}

\section{Proof of Theorem \ref{thm: monotone threshold and payment}}\label{sec: monotone threshold and payment proof}

The proof of Theorem \ref{thm: monotone threshold and payment} consists of three main steps. First, we show that an optimal feasible indirect utility function is still found within the same triangular region $\mathcal{U}_{i}^{*}$, which we will also use in the proof of Theorem \ref{thm: large market with payment}. Second, starting from any optimal mechanism, we show that there exists a monotone threshold allocation rule that implements the same indirect utility function. Finally, we show that optimality further implies that the payment function can then be replaced by a monotone threshold payment rule. 

We first note a standard characterization result for feasible mechanisms. Just for notational convenience, let
\begin{align*}
    X^{*}_{i}(s_{i}) 
    &= \mE[x_{i}(\hat{s}_{i},s_{-i})\mid \omega=+1]+\mE[x_{i}(\hat{s}_{i},s_{-i})\mid \omega=-1], \\
    T^{*}_{i}(s_{i}) 
    &= \mE[t_{i}(\hat{s}_{i},s_{-i})\mid \omega=+1] - \mE[t_{i}(\hat{s}_{i},s_{-i})\mid \omega=-1]. 
\end{align*}
Then, we obtain the following.

\begin{lemma}\label{lem: feasible mechanism transfer}
    A mechanism $(x,t)$ is feasible if and only if
    \begin{align*} 
        & X^{*}_{i}(s_{i}) - T^{*}_{i}(s_{i})
        \text{ is increasing in } s_{i}, \label{eq: M-star}\tag{$\text{M}^{*}$} \\
        & U_{i}(s_{i}) 
        = U_{i}(\underline{s})+\int_{\underline{s}}^{s_{i}}(X_{i}^{*}(u_{i})-T_{i}^{*}(u_{i}))du_{i}, \label{eq: EV-star}\tag{$\text{EV}^{*}$} 
    \end{align*}     
    for each agent $i$ and private belief $s_{i}$ in the support of $\mF$. 
\end{lemma}

\begin{proof}[Proof of Lemma \ref{lem: feasible mechanism transfer}]
    Note that the interim expected payoff to the agent with type $s_{i}$ from reporting type $\hat{s}_{i}$ is written as
    \begin{align*}
        U_{i}(\hat{s}_{i};s_{i})
        &= \mE[\omega\cdot x_{i}(\hat{s}_{i},s_{-i})-t_{i}(\hat{s}_{i},s_{-i})\mid s_{i}] \\
        &= s_{i}\cdot \mE[x_{i}(\hat{s}_{i},s_{-i})-t_{i}(\hat{s}_{i},s_{-i})\mid \omega=+1] \\
        & \quad - (1-s_{i})\cdot \mE[x_{i}(\hat{s}_{i},s_{-i})+t_{i}(\hat{s}_{i},s_{-i})\mid \omega=-1] \\
        &= s_{i} \cdot (X_{i}^{*}(\hat{s}_{i})-T^{*}_{i}(\hat{s}_{i})) - \mE[x_{i}(\hat{s}_{i},s_{-i})+t_{i}(\hat{s}_{i},s_{-i})\mid \omega=-1].
    \end{align*} 
    Then, the identical argument as in the proof of Lemma \ref{lem: feasible mechanism}, applied to the above expression, completes the proof.  
\end{proof}

In the proof of Theorem \ref{thm: optimal mechanism}, we show that optimal indirect utility lies in the triangular region $\mathcal{U}_{i}^{*}$. The next result shows that the same claim holds when the payment is nonnegative. 

\begin{lemma}\label{lem: feasible indirect utility transfer}
    Suppose that the payment is nonnegative $t_{i}(s_{i},s_{-i})\geq 0$.
    If $(x_{i},t_{i})$ is optimal and induces $U_{i}$, then it has an extension $\tilde{U}_{i}:[0,1]\rightarrow \mR$ such that  $\tilde{U}_{i}\in\mathcal{U}^{*}_{i}$. 
\end{lemma}
\begin{proof}[Proof of Lemma \ref{lem: feasible indirect utility transfer}]
    As in the proof of Lemma \ref{lem: feasible indirect utility necessary conditions}, consider an extended mechanism defined on the extended type space $[0,1]^n$. In this mechanism, for each agent $i$, we let types in $[0,\underline{s}]$ choose whichever they prefer between $(x_{i}(\underline{s},s_{-i}), t_{i}(\underline{s},s_{-i}))$ and $(0,0)$. Similarly, types in $[\overline{s},1]$ choose whichever they prefer between $(x_{i}(\overline{s},s_{-i}), t_{i}(\overline{s},s_{-i}))$ and $(0,0)$. All other types receive the same allocation as in the original mechanism. 

    By construction, the induced indirect utility function $\tilde{U}_{i}$ is an extension of $U_{i}$. Moreover, as we see in the proof of Lemma \ref{lem: feasible indirect utility necessary conditions}, we can check that the extended mechanism is feasible. Therefore by Lemma \ref{lem: feasible mechanism transfer}, we have \eqref{eq: M-star} and \eqref{eq: EV-star} for $\tilde{U}_{i}$ and $(X_{i}^{*},T_{i}^{*})$ that is induced by the extended mechanism. 

    It remains to show that $\tilde{U}_{i}\in \mathcal{U}_{i}^{*}$.
    Note that the definition of indirect utility and \eqref{eq: P with money} implies $\tilde{U}_{i}(0)=0$. This implies $X_{i}^{*}-T_{i}^{*}$ must be nonnegative almost everywhere; otherwise, \eqref{eq: M-star} implies that this is negative over some positive-measure interval $[0,s_{i}]$, which together with \eqref{eq: EV-star} implies that $\tilde{U}_{i}(s_{i})<0$, violating the participation constraint. Therefore, \eqref{eq: EV-star} implies that $\tilde{U}_{i}$ is increasing and convex over the extended domain $[0,1]$. Since payments are nonnegative, we must clearly have $0\leq \tilde{U}_{i}\leq \overline{U}_{i}$. 
    
    Finally, whenever $\tilde{U}_{i}(s_{i})<\max\{0,2s_{i}-1\}$ at some point $s_{i}$, we can alternatively assign to this type an option of always allocating the good with zero payment. Since the induced direct mechanism is feasible as seen in the proof of Lemma \ref{lem: feasible indirect utility optimality conditions}, and this operation increases the allocation probability, optimality implies that $U_{i}(s_{i})\geq \max\{0,2s_{i}-1\}$. Therefore, $\tilde{U}_{i}\in \mathcal{U}_{i}^{*}$. 
\end{proof}

The next result implies that there exists a monotone threshold allocation rule $x_{i}$ that is optimal with some payment function. Recall that $x_{i}$ is a monotone threshold allocation if there exist a partition $\mathcal{S}_{i}$ of the signal space $[\underline{s},\overline{s}]$ into disjoint intervals such that for each interval $S_{i}\in \mathcal{S}_{i}$, there exist $\kappa_{x}\in[0,1]$ and $\tau_{x}\geq 0$ such that $x_{i}(s) = \kappa_{x} \cdot \mI\{\lr{s_{-i}}\geq \tau_{x}\}$ for all $s_{i}\in S_{i}$ and $s_{-i}\in [\underline{s},\overline{s}]^{n-1}$.

\begin{lemma} \label{lem: threshold allocation transfer}
    Suppose that the payment is nonnegative $t_{i}(s_{i},s_{-i})\geq 0$.
    If there exists some mechanism $(x_{i},t_{i})$ inducing $U_{i}\in\mathcal{U}^{*}_{i}$, then there exists a monotone threshold allocation rule $\hat{x}_{i}$ such that $(\hat{x}_{i},t_{i})$ is feasible, induces $U_{i}$, and provides the same objective value.  
\end{lemma}
\begin{proof}[Proof of Lemma \ref{lem: threshold allocation transfer}]
    Take any mechanism $(x_{i},t_{i})$ inducing $U_{i}\in\mathcal{U}^{*}_{i}$. 
    Note that, for each $s_{i}$, there exist $a_{i}$ and $b_{i}$ such that $U_{i}(s_{i})=a_{i}s_{i}-b_{i}$, where
    \begin{align*}
        a_{i} &= X_{i}^{*}(s_{i})-T_{i}^{*}(s_{i}) \\
        b_{i} & = \mE[x_{i}(s_{i},s_{-i})+t_{i}(s_{i},s_{-i})\mid \omega=-1].
    \end{align*}
    Note that $a_{i}\in[0,2]$, $b_{i}\in[0,1]$, and $a_{i}-1\leq b_{i}\leq a_{i}/2$ by $U_{i}\in\mathcal{U}_{i}^{*}$, as seen in the proof of Lemma \ref{lem: threshold mechanism induces any linear utility}.  

    Fix $s_{i}$. We show that there exists an allocation of the form $\hat{x}_{i}(s_{i},s_{-i})=\kappa_{x}\cdot \mI\{\lr{s_{-i}}\geq \tau_{x}\}$ that induces the same expected utility $U_{i}(s_{i})$. Since $U_{i}$ is convex, replacing the allocation for each type $s_{i}$ with an allocation from this class yields a monotone threshold allocation rule. The resulting allocation is also incentive compatible, again because $U_{i}$ is convex. Since $U_{i}\geq 0$, the participation constraint remains satisfied. 

    Fix any type $s_{i}$ and consider a threshold allocation rule $\hat{x}_{i}(s)=\kappa_{x}\cdot \mI\{\lr{s_{-i}}\geq \tau_{x}\}$. In order for this threshold allocation rule $\hat{x}_{i}$ to induce the same indirect utility function, it is enough to have $\mE[x_{i}(s)\mid \omega]=\mE[\hat{x}_{i}(s)\mid \omega]$ for each state $\omega\in \{-1,+1\}$. Rephrasing, 
    \begin{align*} 
        0&=\int_{s_{-i}}\left\{\prod_{k\neq i}s_{k}\right\}[\kappa_{x} \cdot \mI\{\lr{s_{-i}}\geq \tau_{x}\}-x_{i}(s)]d\mF^{\otimes}(s_{-i}), 
        \label{eq: slope transfer} \tag{H} \\
        0&=\int_{s_{-i}}\left\{\prod_{k\neq i}(1-s_{k})\right\}[\kappa_{x} \cdot \mI\{\lr{s_{-i}}\geq \tau_{x}\}-x_{i}(s)]d\mF^{\otimes}(s_{-i})
        \label{eq: intercept transfer} \tag{L}
    \end{align*}
    for some $\kappa_{x}\in[0,1]$ and $\tau_{x}\geq 0$. By construction, the above equations also imply that the objective value also remains the same. 

    Fix any $\kappa_{x}$ such that 
    \begin{align*} 
    2^{n-1}\int_{s_{-i}}\left\{\prod_{k\neq i}s_{k}\right\} x_{i}(s)d\mF^{\otimes}(s_{-i}) \leq \kappa_{x} \leq 1. 
    \end{align*} 
    If $\tau_{x}=0$, this inequality guarantees that the right-hand side of \eqref{eq: slope transfer} is nonnegative. Similarly, if $\tau_{x}\to \infty$, then the right-hand side of \eqref{eq: slope transfer} converges to a nonpositive number. By dominated convergence theorem, the right-hand side is continuous in $\tau_{x}$, and therefore, the intermediate value theorem ensures the existence of $\tau_{x}(\kappa_{x})\geq 0$ satisfying \eqref{eq: slope transfer}. 

    Next, consider the second condition \eqref{eq: intercept transfer}. 
    We show the existence of $\kappa_{x}$ satisfying \eqref{eq: intercept transfer} together with $\tau_{x}(\kappa_{x})$ which we construct above. Note that, for $\tau_{x}(\kappa_{x})$ to be well-defined, we must find such a $\kappa_{x}$ within the domain given by the above paragraph. 

    First, consider the case
    \begin{align*}
        \kappa_{x} = 2^{n-1}\int_{s_{-i}}\left\{\prod_{k\neq i}s_{k}\right\} x_{i}(s)d\mF^{\otimes}(s_{-i}). 
    \end{align*}
    Then, we have $\tau_{x}(\kappa_{x})=0$. Therefore, the right-hand side of \eqref{eq: intercept transfer} equals
    \begin{align*}
        \int_{s_{-i}}\left\{\prod_{k\neq i}s_{k}-\prod_{k\neq i}(1-s_{k})\right\} x_{i}(s)d\mF^{\otimes}(s_{-i}). 
    \end{align*}
    Now, note that we have $a_{i}/2-b_{i}\geq 0$. Hence,
    \begin{align*}
        &\int_{s_{-i}}\left\{\prod_{k\neq i}s_{k}-\prod_{k\neq i}(1-s_{k})\right\} x_{i}(s)d\mF^{\otimes}(s_{-i}) \\
        &\geq \int_{s_{-i}}\left\{\prod_{k\neq i}s_{k}+\prod_{k\neq i}(1-s_{k})\right\} t_{i}(s)d\mF^{\otimes}(s_{-i}) \\
        &\geq 0,
    \end{align*}
    where the last inequality follows because payment is nonnegative $t_{i}(s_{i},s_{-i})\geq 0$. Thus, the right-hand side of \eqref{eq: intercept transfer} is nonnegative. 

    Second, under $\kappa_{x}=1$ and $\tau_{x}=\tau_{x}(1)$, we show that the right-hand side of \eqref{eq: intercept transfer} is nonpositive. By definition, $\tau_{x}(1)$ satisfies \eqref{eq: slope transfer} at $\kappa_{x}=1$. Equivalently, we have
    \begin{align*}
        0 = & \int_{s_{-i}}\left\{\prod_{k\neq i}s_{k}\right\} \left[\mI\{\lr{s_{-i}}\geq \tau_{x}(1)\}- x_{i}(s)\right]d\mF^{\otimes}(s_{-i}) \\
        & = \int_{s_{-i}} \lr{s_{-i}}\left\{\prod_{k\neq i}(1-s_{k})\right\} \left[\mI\{\lr{s_{-i}}\geq \tau_{x}(1)\}- x_{i}(s)\right]d\mF^{\otimes}(s_{-i}) \\
        & = \int_{s_{-i}} [\lr{s_{-i}}-\tau_{x}(1)]\left\{\prod_{k\neq i}(1-s_{k})\right\} \left[\mI\{\lr{s_{-i}}\geq \tau_{x}(1)\}- x_{i}(s)\right]d\mF^{\otimes}(s_{-i}) \\
        & \qquad + \int_{s_{-i}} \tau_{x}(1)\left\{\prod_{k\neq i}(1-s_{k})\right\} \left[\mI\{\lr{s_{-i}}\geq \tau_{x}(1)\}- x_{i}(s)\right]d\mF^{\otimes}(s_{-i}).
    \end{align*}
    Note that $x_{i}(s)\in[0,1]$ implies that $\mI\{\lr{s_{-i}}\geq \tau_{x}(1)\}- x_{i}(s)\geq 0$ if $\lr{s_{-i}}\geq \tau_{x}(1)$ and $\mI\{\lr{s_{-i}}\geq \tau_{x}(1)\}- x_{i}(s)\leq 0$ if $\lr{s_{-i}}\leq \tau_{x}(1)$.
    Thus, we have
    \[
    \int_{s_{-i}} [\lr{s_{-i}}-\tau_{x}(1)]\left\{\prod_{k\neq i}(1-s_{k})\right\} \left[\mI\{\lr{s_{-i}}\geq \tau_{x}(1)\}- x_{i}(s)\right]d\mF^{\otimes}(s_{-i})\geq 0.
    \]
    Therefore, we obtain 
    \[
    \int_{s_{-i}} \tau_{x}(1)\left\{ \prod_{k\neq i}(1-s_{k})\right\} \left[\mI\{\lr{s_{-i}}\geq \tau_{x}(1)\}- x_{i}(s)\right]d\mF^{\otimes}(s_{-i})\leq 0,
    \]
    which implies the desired inequality as $\tau_{x}(1)\geq 0$.
    By the dominated convergence theorem again, we can show that the right-hand side of \eqref{eq: intercept transfer} is continuous in $\kappa_{x}$ and $\tau_{x}(\kappa_{x})$, and hence, the intermediate value theorem ensures the existence of $\kappa_{x}$ satisfying \eqref{eq: intercept transfer} together with $\tau_{x}(\kappa_{x})$. This completes the proof.
\end{proof}

The next lemma provides a sufficient condition under which there exists a monotone threshold payment rule $t_{i}$ that is optimal with some allocation rule. 

\begin{lemma} \label{lem: sufficient condition optimality of threshold transfer}
    Suppose that the payment is nonnegative and finite $t_{i}(s_{i},s_{-i})\in[0,T]$. If there exists some mechanism $(x_{i},t_{i})$ inducing $U_{i}\in\mathcal{U}_{i}^{*}$ and $t_{i}$ satisfies 
    \[
    \int_{s_{-i}}\left\{\prod_{k\neq i}s_{k}-\prod_{k\neq i}(1-s_{k})\right\}t_{i}(s)d\mF^{\otimes}(s_{-i}) \geq 0,
    \]
    then, there exists a monotone threshold payment rule $\hat{t}_{i}$ such that $(x_{i},\hat{t}_{i})$ is feasible, induces $U_{i}$, and provides the same objective value.
\end{lemma}
\begin{proof}[Proof of Lemma \ref{lem: sufficient condition optimality of threshold transfer}]
    Take any mechanism $(x_{i},t_{i})$ which induces $U_{i}\in\mathcal{U}^{*}_{i}$ and suppose that the payment rule $t_{i}$ satisfies the inequality
    \[
    \int_{s_{-i}}\left\{\prod_{k\neq i}s_{k}-\prod_{k\neq i}(1-s_{k})\right\}t_{i}(s)d\mF^{\otimes}(s_{-i}) \geq 0.
    \]
    Note that, for each $s_{i}$, there exist $a_{i}$ and $b_{i}$ such that $U_{i}(s_{i})=a_{i}s_{i}-b_{i}$, where
    \begin{align*}
        a_{i} &= X_{i}^{*}(s_{i})-T_{i}^{*}(s_{i}) \\
        b_{i} & = \mE[x_{i}(s_{i},s_{-i})+t_{i}(s_{i},s_{-i})\mid \omega=-1].
    \end{align*}
    Note that $U_{i}\in\mathcal{U}_{i}^{*}$ implies $a_{i}\in[0,2]$, $b_{i}\in[0,1]$, and $a_{i}-1\leq b_{i}\leq a_{i}/2$. 

    We show that there exists a payment rule  $\hat{t}_{i}(s_{i},s_{-i})=\kappa_{t}\cdot \mI\{\lr{s_{-i}}\geq \tau_{t}\}$ that induces the same expected utility $U_{i}(s_{i})$. By the same argument as in the proof of Lemma \ref{lem: threshold allocation transfer}, we then obtain a monotone threshold payment rule such that the mechanism $(x,\hat{t})$ remains feasible. Clearly, the allocation probability does not change. 
    
    Fix any type $s_{i}$ and consider a payment rule $\hat{t}_{i}(s)=\kappa_{t}\cdot \mI\{\lr{s_{-i}}\geq \tau_{t}\}$.
    In order for this threshold payment rule $\hat{t}_{i}$ to induce the same indirect utility function, it is enough to have $\mE[t_{i}(s)\mid \omega]=\mE[\hat{t}_{i}(s)\mid \omega]$ for each $\omega\in\{-1,+1\}$. In other words, we must have 
    \begin{align*}
        0 & = \int_{s_{-i}}\left\{\prod_{k\neq i}s_{k}\right\}[\kappa_{t} \cdot \mI\{\lr{s_{-i}}\geq \tau_{t}\}- t_{i}(s)]d\mF^{\otimes}(s_{-i}), \label{eq: slope transfer with transfer} \tag{H'}  \\ 
        0 & = \int_{s_{-i}}\left\{\prod_{k\neq i}(1-s_{k})\right\} [\kappa_{t} \cdot \mI\{\lr{s_{-i}}\geq \tau_{t}\}- t_{i}(s)]d\mF^{\otimes}(s_{-i}), \label{eq: intercept transfer with transfer} \tag{L'}
    \end{align*}
    for some $\kappa_{t}\in[0,T]$ and $\tau_{t}\geq 0$. As in Lemma \ref{lem: threshold allocation transfer}, we show the existence of $\kappa_{t}\in[0,T]$ and $\tau_{t}\geq 0$ satisfying \eqref{eq: slope transfer with transfer} and \eqref{eq: intercept transfer with transfer}.
    
    Take any $\kappa_{t}\in[0,T]$ such that 
    \[
    2^{n-1}\int_{s_{-i}}\left\{\prod_{k\neq i}s_{k}\right\} t_{i}(s)d\mF^{\otimes}(s_{-i}) \leq \kappa_{t} \leq T.
    \]
    If $\tau_{t}=0$, this inequality guarantees that the right-hand side of \eqref{eq: slope transfer with transfer} is nonnegative.
    Similarly, if $\tau_{t}\to \infty$, then the right-hand side of \eqref{eq: slope transfer with transfer} is nonpositive.
    By dominated convergence theorem, the right-hand side is continuous in $\tau_{t}$, and therefore, the intermediate value theorem ensures the existence of $\tau_{t}(\kappa_{t})\geq 0$ satisfying \eqref{eq: slope transfer with transfer}. 

    Next, consider the second condition \eqref{eq: intercept transfer with transfer}. First, consider the case 
    \[
    \kappa_{t}=2^{n-1} \int_{s_{-i}}\left\{\prod_{k\neq i}s_{k}\right\} t_{i}(s)d\mF^{\otimes}(s_{-i}). 
    \]
    Then, we have $\tau_{t}=\tau_{t}(\kappa_{t})=0$.
    Therefore, the right-hand side of \eqref{eq: intercept transfer with transfer} equals 
    \[
    \int_{s_{-i}}\left\{\prod_{k\neq i}s_{k}-\prod_{k\neq i}(1-s_{k})\right\}t_{i}(s)d\mF^{\otimes}(s_{-i}),
    \]
    which is nonnegative by assumption.

    Second, under $\kappa_{t}=T$ and $\tau_{t}=\tau_{t}(T)$, we show that the right-hand side of \eqref{eq: intercept transfer with transfer} is nonpositive.
    By definition, $\tau_{t}(T)$ satisfies \eqref{eq: slope transfer with transfer} at $\kappa_{t}=T$. Equivalently, we have 
    \begin{align*}
        0 = & \int_{s_{-i}}\left\{\prod_{k\neq i}s_{k}\right\} \left[T\cdot \mI\{\lr{s_{-i}}\geq \tau_{t}(T)\}- t_{i}(s)\right]d\mF^{\otimes}(s_{-i}) \\
        & = \int_{s_{-i}} \lr{s_{-i}}\left\{\prod_{k\neq i}(1-s_{k})\right\} \left[T\cdot \mI\{\lr{s_{-i}}\geq \tau_{t}(T)\}- t_{i}(s)\right]d\mF^{\otimes}(s_{-i}) \\
        & = \int_{s_{-i}} [\lr{s_{-i}}-\tau_{t}(T)]\left\{\prod_{k\neq i}(1-s_{k})\right\} \left[T\cdot \mI\{\lr{s_{-i}}\geq \tau_{t}(T)\}- t_{i}(s)\right]d\mF^{\otimes}(s_{-i}) \\
        & \qquad + \int_{s_{-i}} \tau_{t}(T)\left\{\prod_{k\neq i}(1-s_{k})\right\} \left[T \cdot \mI\{\lr{s_{-i}}\geq \tau_{t}(T)\}- t_{i}(s)\right]d\mF^{\otimes}(s_{-i}).
    \end{align*}
    Note that $t_{i}(s_{i},s_{-i})\in[0,T]$ implies that $T\cdot \mI\{\lr{s_{-i}}\geq \tau_{t}(T)\}- t_{i}(s)\geq 0$ if $\lr{s_{-i}}\geq \tau_{t}(T)$ and $T\cdot \mI\{\lr{s_{-i}}\geq \tau_{t}(T)\}- t_{i}(s)\leq 0$ if $\lr{s_{-i}}\leq \tau_{t}(T)$.
    Thus, we have
    \[
    \int_{s_{-i}} [\lr{s_{-i}}-\tau_{t}(T)]\left\{\prod_{k\neq i}(1-s_{k})\right\} \left[T\cdot \mI\{\lr{s_{-i}}\geq \tau_{t}(T)\}- t_{i}(s)\right]d\mF^{\otimes}(s_{-i})\geq 0.
    \]
    Therefore, we obtain 
    \[
    \int_{s_{-i}} \tau_{t}(T)\left\{ \prod_{k\neq i}(1-s_{k})\right\} \left[T\cdot \mI\{\lr{s_{-i}}\geq \tau_{t}(T)\}- t_{i}(s)\right]d\mF^{\otimes}(s_{-i})\leq 0,
    \]
    which implies the desired inequality as $\tau_{t}(T)\geq 0$.
    By the dominated convergence theorem again, we can show that the right-hand side of \eqref{eq: intercept transfer with transfer} is continuous in $\kappa_{t}$ and $\tau_{t}(\kappa_{t})$, and hence, the intermediate value theorem ensures the existence of $\kappa_{t}$ satisfying \eqref{eq: intercept transfer with transfer} together with $\tau_{t}(\kappa_{t})$. 
    This completes the proof.
\end{proof}

The proof of Theorem \ref{thm: monotone threshold and payment} is complete once the inequality condition in Lemma \ref{lem: sufficient condition optimality of threshold transfer} is verified. The next lemma shows that this condition must hold at any optimal mechanism.

\begin{lemma} \label{lem: tight necessary condition of optimal transfer}
Suppose that the payment is nonnegative and finite $t_{i}(s_{i},s_{-i})\in[0,T]$. If a feasible mechanism $(x,t)$ is optimal, then 
\[
\int_{s_{-i}} \left\{\prod_{k\neq i}s_k-\prod_{k\neq i}(1-s_k)\right\}t_i(s)\,d\mF^{\otimes}(s_{-i})\geq 0,
\]
for each $i$ and $s_{i}$. 
\end{lemma}

\begin{proof}[Proof of Lemma \ref{lem: tight necessary condition of optimal transfer}]
    Let $(x,t)$ be an optimal mechanism. Fix any agent $i$. To simplify exposition, we prepare a few functions: 
    \begin{align*}
        H^{-}(s_{-i})&=\prod_{k\neq i}s_{k}-\prod_{k\neq i}(1-s_{k}), \\
        H^{+}(s_{-i})&=\prod_{k\neq i}s_{k}+\prod_{k\neq i}(1-s_{k}), 
    \end{align*}
    for each $s_{-i}$. By seeking a contradiction, we assume that
    \begin{align*}
        \int_{s_{-i}}H^{-}(s_{-i})t_{i}(s)d\mF^{\otimes}(s_{-i}) < 0, \label{eq: assump}\tag{A}
    \end{align*}
    for some type $s_{i}$. 

    We show that \eqref{eq: assump} implies that the mechanism can be perturbed in such a way that the indirect utility function remains unchanged while the allocation probability increases. As in Lemma \ref{lem: threshold allocation transfer}, the convexity of the indirect utility function guarantees the feasibility of the perturbed mechanism. By Lemma \ref{lem: threshold allocation transfer}, we can assume that $x_{i}$ admits a monotone threshold structure. 

    First, we construct a perturbed payment rule. Consider two event sets
    \begin{align*}
        A^{-} &= \{s_{-i}\mid H^{-}(s_{-i})<0 \text{ and } t_{i}(s_{i},s_{-i})>0\}, \\
        A^{+} &= \{s_{-i}\mid H^{-}(s_{-i})>0 \text{ and } t_{i}(s_{i},s_{-i})<T\}.
    \end{align*}
    Note that \eqref{eq: assump} implies that $A^{-}$ has positive measure in $\mF^{\otimes}$, i.e., $\mF^{\otimes}(A^{-})>0$. Suppose that $A^{+}$ is measure zero. Then, $H^{-}(s_{-i})>0$ implies $t_{i}(s_{i},s_{-i})=T$ almost everywhere, and thus $H^{-}(s_{-i})t_{i}(s)\geq H^{-}(s_{-i})T$. Hence, by taking these expectations, we have 
    \begin{align*}
        \int_{s_{-i}}H^{-}(s_{-i})t_{i}(s)d\mF^{\otimes}(s_{-i})\geq T\int_{s_{-i}}H^{-}(s_{-i})d\mF^{\otimes}(s_{-i})=0,
    \end{align*}
    which contradicts to \eqref{eq: assump}. Therefore, $A^{+}$ also has positive measure. 

    Therefore, there exists a sufficiently large integer $m$ such that the sets
    \begin{align*}
        A^{-}_{m} &= \{s_{-i}\mid H^{-}(s_{-i})<0 \text{ and } t_{i}(s_{i},s_{-i})>1/m\}, \\
        A^{+}_{m} &= \{s_{-i}\mid H^{-}(s_{-i})>0 \text{ and } t_{i}(s_{i},s_{-i})<T-1/m\}. 
    \end{align*}
    have positive measures. Now, consider a finite measure
    \begin{align*}
        \nu(B)=\int_B K(s_{-i})\,d\mF^{\otimes}(s_{-i}), \text{ where } K(s_{-i})=\prod_{k\neq i}(1-s_{k}).
    \end{align*}
    Because $K>0$ and $\mF^
    {\otimes}$ has smooth density, $\nu$ is atomless. Hence, there exist measurable subsets $E^{-}\subset A_{m}^{-}$ and $ E^{+}\subset A_{m}^{+}$ such that $\nu(E^{-})=\nu(E^{+})+\eta>0$ for some small number $\eta>0$ which we specify below. 

    Now, for a small positive number $\varepsilon<1/m$, we define a perturbed payment as $\hat{t}_{i}(s_{i},s_{-i})=t_{i}(s_{i},s_{-i})+\varepsilon \psi(s_{-i})$, where $\psi(s_{-i})=\mI\{s_{-i}\in E^{+}\}-\mI\{s_{-i}\in E^{-}\}$. Note that $0\leq \hat{t}_{i}\leq T$ by construction. Then, since $\nu(E^{-})=\nu(E^{+})+\eta>0$, we have
    \begin{align*}
        &\int_{s_{-i}} K(s_{-i})(\hat{t}_{i}(s)-t_{i}(s))d\mF^{\otimes}(s_{-i})=-\eta, \label{eq: intercept t}\tag{Intercept T} \\
        &\int_{s_{-i}} H^{-}(s_{-i})(\hat{t}_{i}(s)-t_{i}(s))d\mF^{\otimes}(s_{-i}) = \varepsilon\Delta, \label{eq: slope t}\tag{Slope T}
    \end{align*}
    where $\Delta=\mE[H^{-}(s_{-i})\psi(s_{-i})]$, with expectation taken with respect to $\mF^{\otimes}$. Note that $\Delta>0$ because $H^{-}(s_{-i})>0$ if $s_{-i}\in E^{+}$ and $H^{-}(s_{-i})<0$ if $s_{-i}\in E^{-}$. 

    Second, we define a perturbed allocation rule. Let the indirect utility be represented by $U_{i}(s_{i})=a_{i}s_{i}-b_{i}$, where we define $a_{i}$ and $b_{i}$ as in the proof of Lemma \ref{lem: threshold allocation transfer}. Recall that $U_{i}\in \mathcal{U}_{i}^{*}$ implies $a_{i}\in [0,2]$ and $a_{i}/2\geq b_{i}$. 

    As we do in the construction of the perturbed transfer, consider two event sets $B^{-}=\{s_{-i}\mid x_{i}(s_{i},s_{-i})>0\}$ and $B^{+}=\{s_{-i}\mid x_{i}(s_{i},s_{-i})<1\}$. If $B^{+}$ has zero measure, then $x_{i}(s)$ for almost all $s_{-i}$, and therefore,
    \begin{align*}
        a_{i} = 2-2^{n-1}\int_{s_{-i}}H^{-}(s_{-i})t_{i}(s)d\mF^{\otimes}>2,
    \end{align*}
    where the strict inequality follows by assumption \eqref{eq: assump}. This is a contradiction, and hence, $B^{+}$ must have a strictly positive measure. Likewise, if $B^{-}$ has zero measure, then $x_{i}(s)=0$ almost surely, hence
    \begin{align*}
        a_{i}/2-b_{i} 
        &= -2^{n-2}\int_{s_{-i}}H^{-}(s_{-i})t_{i}(s)d\mF^{\otimes}-2^{n-1}\int_{s_{-i}}K(s_{-i})t_{i}(s)d\mF^{\otimes}(s_{-i}) \\
        &= -2^{n-2}\int_{s_{-i}}H^{+}(s_{-i})t_{i}(s)d\mF^{\otimes}(s_{-i}) \\
        &<0,
    \end{align*}
    where the last inequality follows because \eqref{eq: assump} implies that $t_{i}$ takes strictly positive values on a positive-measure set. Again, this is a contradiction, and $B^{-}$ has a positive measure. 

    Therefore, there exists a sufficiently large integer $m$ such that the sets
    \begin{align*}
        B^{-}_{m}&=\{s_{-i}\mid x_{i}(s_{i},s_{-i})\geq 1/m\}, \\
        B^{+}_{m}&=\{s_{-i}\mid x_{i}(s_{i},s_{-i})\leq 1-1/m\},
    \end{align*}
    have positive measure. Since $x_{i}$ is a monotone threshold allocation, both $B^{-}_{m}$ and $B^{+}_{m}$ are a connected set with a non-empty set of interior points. 
    
    Now, define $R=H^{+}/K$. 
    Then, we can take interior points $b^{-}\in B^{-}_{m}$ and $b^{+}\in B^{+}_{m}$ such that $\delta=|R(b^{-})-R(b^{+})|/4>0$. Since $R$ is a continuous function, there exist open neighborhoods of these points, $V^{-}\subset B^{-}_{m}$ and $V^{+}\subset B^{+}_{m}$, such that $\sup_{V^{+}}R<\inf_{V^{-}}R$. Then, define
    \begin{align*}
        N^{-} &= \int_{V^{-}} K(s_{-i})d\mF^{\otimes}(s_{-i}), 
        \quad 
        D^{-} = \int_{V^{-}} H^{+}(s_{-i})d\mF^{\otimes}(s_{-i}), \\
        N^{+} &= \int_{V^{+}} K(s_{-i})d\mF^{\otimes}(s_{-i}), 
        \quad 
        D^{+} = \int_{V^{+}} H^{+}(s_{-i})d\mF^{\otimes}(s_{-i}). 
    \end{align*} 
    Note that $R=H^{+}/K$, and therefore, $D^{-}/N^{-}$ is a weighted average of $R$ over $V^{-}$, which implies $D^{-}/N^{-} > \inf_{V^{-}}R$. Likewise, $D^{+}/N^{+} < \sup_{V^{+}}R$. Hence, $\sup_{V^{+}}R<\inf_{V^{-}}R$ implies that we have $N^{-}D^{+}-N^{+}D^{-}<0$. 

    Now, define a perturbed allocation rule $\hat{x}_{i}(s_{i},s_{-i})=x_{i}(s_{i},s_{-i})+\chi(s_{-i})$, where $\chi(s_{-i})=\alpha\mI\{s_{-i}\in V^{+}\}-\beta\mI\{s_{-i}\in V^{-}\}$, where
    \begin{align*}
        \alpha = \frac{\varepsilon \Delta N^{-}-\eta D^{-}}{N^{-}D^{+}-N^{+}D^{-}} 
        \quad \text{ and } 
        \beta = \frac{\varepsilon \Delta N^{+}-\eta D^{+}}{N^{-}D^{+}-N^{+}D^{-}}. 
    \end{align*} 
    Note that if $\eta/\varepsilon\Delta>N^{+}/D^{+}$, then we also have $\eta/\varepsilon\Delta>N^{-}/D^{-}$, and both $\alpha$ and $\beta$ are positive. By definition, $K<H^{+}$, hence $N^{+}/D^{+}<1$. Choose $\eta=\varepsilon\Delta (1+N^{+}/D^{+})/2$. If $\varepsilon$ is sufficiently small, then, by the definitions of $V^{-}$ and $V^{+}$, we have $\hat{x}_{i}(s_{i},s_{-i})\in [0,1]$. Moreover, 
    \begin{align*}
        &\int_{s_{-i}}K(s_{-i})\chi(s_{-i})d\mF^{\otimes}(s_{-i}) 
        = \alpha N^{+} - \beta N^{-}
        = \eta, \label{eq: intercept x}\tag{Intercept X}\\
        &\int_{s_{-i}}H^{+}(s_{-i})\chi(s_{-i})d\mF^{\otimes}(s_{-i}) 
        = \alpha D^{+} - \beta D^{-} 
        = \varepsilon \Delta, \label{eq: slope x}\tag{Slope X}
    \end{align*}
    where the equalities follow by the definitions of $\alpha$ and $\beta$. 
    
    Finally, we verify that the perturbation keeps the same indirect utility function function while increasing the allocation probability. Note that equations \eqref{eq: intercept t} and \eqref{eq: intercept x} ensure that the intercept of the indirect utility $b_{i}$ remains the same. Likewise, equations \eqref{eq: slope t} and \eqref{eq: slope x} imply that $a_{i}$ does not change by the perturbation. Finally, \eqref{eq: slope t} and \eqref{eq: slope x} also show 
    \begin{align*}
        \mE[\chi(s_{-i})\mid s_{i}]
        &= 2^{n-1}\int_{s_{-i}}\left\{s_{i}\cdot \prod_{k\neq i}s_{k} + (1-s_{i})\cdot \prod_{k\neq i}(1-s_{k})\right\} \chi(s_{-i})d\mF^{\otimes}(s_{-i}) \\
        & = 2^{n-1}\int_{s_{-i}}[s_{i}H^{+}(s_{-i})+(1-2s_{i})K(s_{-i})]\chi(s_{-i})d\mF^{\otimes}(s_{-i}) \\
        & = 2^{n-1}[s_{i}\varepsilon\Delta+(1-2s_{i})\eta] \\
        & = 2^{n-1}[s_{i}(\varepsilon\Delta-\eta)+(1-s_{i})\eta] \\
        & >0
    \end{align*} 
    where the inequality follows by $0<\eta<\varepsilon\Delta$. Thus, the perturbation increases allocation probability. Therefore, we can assume without loss of optimality that we have the desired inequality.  
\end{proof}

\begin{proof}[Proof of Theorem \ref{thm: monotone threshold and payment}]
    Let $(x,t)$ be an optimal mechanism. By Lemma \ref{lem: threshold allocation transfer}, we may assume that $x$ is a monotone threshold allocation rule. By Lemmas \ref{lem: sufficient condition optimality of threshold transfer} and \ref{lem: tight necessary condition of optimal transfer}, we may then also replace $t$ with a monotone threshold payment function without loss of optimality. If $x$ and $t$ are monotone threshold rules with distinct partitions, we can take the coarsest refinement of the two partitions, which also consists of monotone intervals. By definition, $(x,t)$ is then a monotone threshold mechanism with respect to the refined partition.
\end{proof}

\section{Proofs of Proposition \ref{prop: positive payment} and Remark \ref{rem: revenue and nonnegative payments}}\label{sec: positive payment proof} 

\begin{lemma}\label{lem: exclusion vs zero utility}
    Suppose that $(x_{i},t_{i})$ satisfies both \eqref{eq: IC with money} and \eqref{eq: P with money} and let $U_{i}$ be the induced indirect utility function. Then, for each $s_{i}$, $U_{i}(s_{i})=0$ if and only if $x_{i}(s_{i},s_{-i})=t_{i}(s_{i},s_{-i})=0$ for almost all $s_{-i}$.  
\end{lemma}

\begin{proof}[Proof of Lemma \ref{lem: exclusion vs zero utility}]
    Define $X_{i}^{*}$ and $T^{*}_{i}$ as specified in Lemma \ref{lem: feasible mechanism transfer}. Take any $s_{i}$ such that $U_{i}(s_{i})=0$. Then, since $U_{i}$ is weakly increasing, $U_{i}(r_{i})=0$ for all $r_{i}\leq 0$. 

    Now, recall that 
    \begin{align*}
        U_{i}(r_{i})
        =r_{i}(X_{i}^{*}(r_{i})-T_{i}^{*}(r_{i}))- \mE[x_{i}(r_{i},s_{-i})+t_{i}(r_{i},s_{-i})\mid \omega=-1].
    \end{align*}
    Now, since \eqref{eq: M-star} implies that $X_{i}^{*}-T^{*}_{i}$ is weakly increasing and $U_{i}$ is also weakly increasing, $X_{i}^{*}-T^{*}_{i}\geq 0$. Therefore, by the envelope formula \eqref{eq: EV-star}, $X_{i}^{*}(r_{i})-T_{i}^{*}(r_{i})=0$ must hold. If $x_{i}(r_{i},s_{-i})>0$ or $t_{i}(r_{i},s_{-i})>0$ for a positive measure of types $s_{-i}$, then the second term of the above expression is strictly positive, hence $U_{i}(r_{i})<0$. This is a contradiction. 
\end{proof}

\begin{proof}[Proof of Proposition \ref{prop: positive payment}]
Let $(x_{i}^{0},t_{i}^{0})$ be a feasible mechanism with nonnegative payment rules, and let $U_{i}^{0}$ be its induced indirect interim utility. 
Suppose that there exists $\epsilon_{i}>0$ such that $\mE[x_{i}^{0}(s_{i},s_{-i})\mid s_{i}]=0$. 
Then, since $U_{i}^{0}$ is increasing, Lemma \ref{lem: exclusion vs zero utility} implies that we have $\mE[x_{i}^{0}(s_{i},s_{-i})\mid s_{i}]=0$ for all $s_{i}\in(0,\epsilon_{i}]$. 
Since payments are nonnegative, \eqref{eq: P with money} implies that $x_{i}^{0}(s_{i},s_{-i})=t_{i}^{0}(s_{i},s_{-i})=0$ almost surely conditional on each such type. Therefore, $U_{i}^{0}(s_{i})=0$ for all $s_{i}\in(0,\epsilon_{i}]$. 
Fix $c\in(0,\epsilon_{i})$.

For each $L>0$, let $E=\{s_{-i}\mid\lr{s_{-i}}\geq L\}$. Note that $\mP[E\mid\omega=-1]/\mP[E\mid\omega=+1]\leq L^{-1}$. Now, consider the following option: on the event $E$, allocate the good with probability $\kappa\in(0,1]$ and charge the payment $\kappa q$, where
\begin{align*}
q=\mE[\omega\mid s_{i}=c,E]=\frac{c\mP[E\mid\omega=+1]-(1-c)\mP[E\mid\omega=-1]}{c\mP[E\mid\omega=+1]+(1-c)\mP[E\mid\omega=-1]}.
\end{align*}
The preceding bound implies that $q\in(0,1)$ and thus $\kappa q\in (0,T)$ for all sufficiently large $L$. 
Note that the interim utility of type $s_{i}$ from this option is given by $V_{i}(s_{i})=A(s_{i}-c)$, where
\begin{align*}
A=\frac{2\kappa\mP[E\mid\omega=+1]\mP[E\mid\omega=-1]}{c\mP[E\mid\omega=+1]+(1-c)\mP[E\mid\omega=-1]}
\leq\frac{2\kappa\mP[E\mid\omega=+1]}{cL}
\leq\frac{2\kappa}{cL}.
\end{align*}
Note that $A\to0$ as $L\to\infty$, keeping $c$ and $\kappa$ fixed.

Now, we construct an alternative mechanism. 
Fix a sufficiently large $L_{*}$. Let $E_{*}$, $q_{*}$, and $A_{*}>0$ be the corresponding parameters as specified above. Then, take another sufficiently large $L$ such that $A/A_{*}\in(0,1)$. Consider the direct mechanism $(x_{i},t_{i})$ that assigns the zero option to reports $s_{i}\leq c$, the new option associated with $E$ to reports $s_{i}\in(c,\epsilon_{i}]$, and, for reports $s_{i}>\epsilon_{i}$, the original option with probability $1-A/A_{*}$ and the fixed option associated with $E_{*}$ with probability $A/A_{*}$.
For every true type $s_{i}$ and report $\hat{s}_{i}$, the expected payoff is
\begin{align*}
    \left(1-\frac{A}{A_{*}}\right)\mE[\omega x_{i}^{0}(\hat{s}_{i},s_{-i})-t_{i}^{0}(\hat{s}_{i},s_{-i})\mid s_{i}] + \mI\{\hat{s}_{i}>c\}V_{i}(s_{i}).
\end{align*}
This uses the fact that every original mechanism gives zero expected payoff to reports in $(0,\epsilon_{i}]$ and $(A/A_{*})A_{*}(s_{i}-c)=V_{i}(s_{i})$. 
Note that the resulting mechanism satisfies both \eqref{eq: P with money} and \eqref{eq: IC with money}, because truthful reporting maximizes both of these two terms. 

It remains to show that the expected allocation probability strictly increases.
The allocation probability from the types $s_{i}\in(c,\epsilon_{i}]$ increases by at least 
\begin{align*}
\kappa c\mP[E\mid\omega=+1](\mF(\epsilon_{i})-\mF(c)).
\end{align*}
Note that types below $c$ remain excluded. For types above $\epsilon_{i}$, the original option is retained with probability $1-A/A_{*}$. Thus, the allocation probability from these types decreases by at most 
\begin{align*}
    \frac{A}{A_{*}}\int_{\epsilon_{i}}^{1}\mE[x_{i}^{0}(s_{i},s_{-i})\mid s_{i}]f(s_{i})ds_{i}
    \leq \frac{A}{A_{*}}\leq \frac{2\kappa\mP[E\mid\omega=+1]}{cLA_{*}}.
\end{align*}
Therefore, the net change in expected allocation is bounded below by
\begin{align*}
\kappa\mP[E\mid\omega=+1]\left[c(\mF(\epsilon_{i})-\mF(c))-\frac{2}{cLA_{*}}\right]>0,
\end{align*}
where the strict inequality holds for all sufficiently large $L$, since $A_{*}>0$ is fixed and $\mF(\epsilon_{i})-\mF(c)>0$. This contradicts the optimality of the original mechanism.  
\end{proof}

\begin{lemma}\label{lem: revenue threshold allocation}
    There exists an optimal mechanism $(x_{i},t_{i})$ such that $x_{i}$ is a monotone threshold allocation rule.
\end{lemma}
\begin{proof}[Proof of Lemma \ref{lem: revenue threshold allocation}]
    Let $(x,t)$ be an optimal mechanism inducing $U_{i}$, and let $\tilde{U}_{i}:[0,1]\rightarrow\mR$ be the extension of $U_{i}$ as in the proof of Lemma \ref{lem: feasible indirect utility transfer}.
    We first show that the restriction $\tilde{U}'_{i}\leq {\vphantom{U}\smash{\overline{U}}}'_{i}$ is without loss of optimality. Define
    \begin{align*}
        \hat{U}_{i}(s_{i})=\int_{0}^{s_{i}}\min\{\tilde{U}_{i}'(u_{i}),{\vphantom{U}\smash{\overline{U}}}'_{i}(u_{i})\}du_{i}.
    \end{align*}
    Then, $\hat{U}_{i}$ is increasing and convex, $0\leq \hat{U}_{i}\leq \tilde{U}_{i}$, and $\hat{U}_{i}'\leq {\vphantom{U}\smash{\overline{U}}}'_{i}$.
    We implement $\hat{U}_{i}$ without decreasing revenue as follows. Fix any type $s_{i}$ and define $a_{i}$ and $b_{i}$ as in the proof of Lemma \ref{lem: threshold allocation transfer}.
    If $\tilde{U}_{i}'(s_{i})\leq {\vphantom{U}\smash{\overline{U}}}'_{i}(s_{i})$, take a convex combination of the original option and the constant allocation and payment $(a_{i}/2,T)$ that induces slope $a_{i}$ and intercept $a_{i}s_{i}-\hat{U}_{i}(s_{i})$.
    This is possible because $b_{i}\leq a_{i}s_{i}-\hat{U}_{i}(s_{i})\leq a_{i}s_{i}\leq a_{i}/2+T$, and does not decrease revenue since \eqref{eq: P with money} implies $\mE[t_{i}(s_{i},s_{-i})\mid s_{i}]\leq s_{i}\leq T$.
    If $\tilde{U}_{i}'(s_{i})>{\vphantom{U}\smash{\overline{U}}}'_{i}(s_{i})$, efficient allocation at $s_{i}$ with constant payment $\overline{U}_{i}(s_{i})-\hat{U}_{i}(s_{i})\in[0,1]$ induces utility $\hat{U}_{i}(s_{i})$ and slope ${\vphantom{U}\smash{\overline{U}}}'_{i}(s_{i})$, and yields revenue at least $\overline{U}_{i}(s_{i})-\tilde{U}_{i}(s_{i})$, an upper bound on the original revenue.
    Since $\hat{U}_{i}$ is convex and nonnegative, the resulting mechanism is feasible and also optimal.
    Thus, without loss of optimality, we may assume $\tilde{U}_{i}'\leq {\vphantom{U}\smash{\overline{U}}}'_{i}$.

    We now show that, without loss of optimality,
    \begin{align*}
        \mE[x_{i}(s_{i},s_{-i})\mid \omega=+1]\geq \mE[x_{i}(s_{i},s_{-i})\mid \omega=-1].
    \end{align*}
    If this inequality fails, the constant allocation $X_{i}^{*}(s_{i})/2$ with payment rule $t_{i}$ induces the same slope $a_{i}$ and an intercept lower than $b_{i}$.
    By the previous argument, it is without loss of optimality to assume that $a_{i}=\tilde{U}_{i}'(s_{i})\leq {\vphantom{U}\smash{\overline{U}}}'_{i}(s_{i})\leq 2$.
    Thus, we have $b_{i}\leq a_{i}s_{i}\leq a_{i}/2+T$, and hence a convex combination with the constant allocation and payment $(a_{i}/2,T)$ induces the original pair of slope $a_{i}$ and intercept $b_{i}$ and strictly increases revenue as $T>\mE[t_{i}(s_{i},s_{-i})\mid s_{i}]$.
    Its allocation is constant and hence satisfies the above inequality.
    Replacing all such options preserves $U_{i}$ and does not decrease revenue.
    By the construction in the proof of Lemma \ref{lem: threshold allocation transfer}, we can therefore replace $x_{i}$ with a monotone threshold allocation rule giving the same expected allocation conditional on each state.
    This preserves $U_{i}$ and the conditional revenue of each option.
\end{proof}

\begin{proof}[Proof of Remark \ref{rem: revenue and nonnegative payments}]
    Take the optimal mechanism constructed in the proof of Lemma \ref{lem: revenue threshold allocation}, whose extended indirect utility satisfies $\tilde{U}_{i}'\leq {\vphantom{U}\smash{\overline{U}}}'_{i}$.
    If $T_{i}^{*}(s_{i})<0$ for some $s_{i}$, we replace the option at $s_{i}$ as follows.
    Choose $\hat{s}_{i}\leq s_{i}$ with ${\vphantom{U}\smash{\overline{U}}}'_{i}(\hat{s}_{i})=a_{i}=\tilde{U}_{i}'(s_{i})$.
    Efficient allocation at $\hat{s}_{i}$ with constant payment $b_{i}-a_{i}\hat{s}_{i}+\overline{U}_{i}(\hat{s}_{i})$ induces the same $(a_{i},b_{i})$. The payment is feasible because
    \begin{align*}
        0\leq b_{i}-a_{i}\hat{s}_{i}+\overline{U}_{i}(\hat{s}_{i})
        \leq \overline{U}_{i}(\hat{s}_{i})+(1-\hat{s}_{i})a_{i}\leq 1\leq T.
    \end{align*}
    The original allocation is not efficient at $\hat{s}_{i}$, since $X_{i}^{*}(s_{i})=a_{i}+T_{i}^{*}(s_{i})<{\vphantom{U}\smash{\overline{U}}}'_{i}(\hat{s}_{i})$.
    Since the efficient allocation is unique almost everywhere, we have
    \begin{align*}
        b_{i}-a_{i}\hat{s}_{i}+\overline{U}_{i}(\hat{s}_{i})
        >\mE[t_{i}(s_{i},s_{-i})\mid s_{i}]+(\hat{s}_{i}-s_{i})T_{i}^{*}(s_{i}) \geq \mE[t_{i}(s_{i},s_{-i})\mid s_{i}],
    \end{align*}
    so this replacement increases conditional revenue.
    Replacing all such options preserves $U_{i}$ and optimality, gives $T_{i}^{*}(s_{i})=0$ at these reports, and retains a monotone threshold allocation rule.
    Thus, without loss of optimality, $T_{i}^{*}(s_{i})\geq 0$, and the same argument as in the proof of Lemma \ref{lem: sufficient condition optimality of threshold transfer} gives a monotone threshold payment rule with the same expected payment conditional on each state.

    Suppose that an optimal mechanism $(x_{i}^{0},t_{i}^{0})$ excludes a positive mass of types.
    As in the proof of Proposition \ref{prop: positive payment}, there is an $\epsilon_{i}>0$ such that both allocation and payment vanish almost surely for every $s_{i}\in(0,\epsilon_{i}]$.
    Use the mechanism $(\hat{x}_{i},\hat{t}_{i})$ constructed in that proof.
    Since $\mE[t_{i}^{0}(s)]\leq T$, the gain from types in $(c,\epsilon_{i}]$ and the loss from higher types satisfy
    \begin{align*}
        \mE[\hat{t}_{i}(s)]-\mE[t_{i}^{0}(s)] 
        &\geq \kappa q c\mP[E\mid\omega=+1](\mF(\epsilon_{i})-\mF(c))-\frac{A}{A_{*}}\mE[t_{i}^{0}(s)] \\
        &\geq \kappa\mP[E\mid\omega=+1]
        \left[q c(\mF(\epsilon_{i})-\mF(c))-\frac{2T}{cLA_{*}}\right] \\
        & > 0
    \end{align*}
    for all sufficiently large $L$, because $q\to 1$ and $\mF(\epsilon_{i})-\mF(c)>0$, a contradiction.
\end{proof}

\section{Proof of Theorem \ref{thm: large market with payment}}\label{sec: large market and payment proof}

We first characterize an optimal mechanism in the limit market, where $n\rightarrow \infty$. Then, we construct a sequence of mechanisms that converges to the limit-market optimal mechanism in the limit. 

As discussed in the main section, the law of large numbers applies in the limit $n\rightarrow \infty$, and the profile $s_{-i}$ reveals the state $\omega$. Therefore, for each agent $i$, mechanisms $(x_{i},t_{i})$ in the limit market can be expressed as a function of private types $s_{i}$ and states $\omega\in \{-1,+1\}$. It admits a decomposition 
\begin{align*}
    x_{i}(s_{i},\omega) & = \kappa_{+1}(s_{i})\cdot \mI\{\omega=+1\} +\kappa_{-1}(s_{i})\cdot \mI\{\omega=-1\} \\
    t_{i}(s_{i},\omega) & = \tau_{+1}(s_{i})\cdot \mI\{\omega=+1\}+\tau_{-1}(s_{i})\cdot \mI\{\omega=-1\}
\end{align*}
where $\kappa_{\omega}(s_{i})\in[0,1]$ and $\tau_{\omega}(s_{i})\geq 0$. We call mechanisms that can directly depend on the state \textit{limit-market mechanisms}. 

We can extend Lemma \ref{lem: limit problem berge} to the current setting. 

\begin{lemma} \label{lem: optimal indirect utility limit payment}
    Suppose that payment is nonnegative, that is, $t_{i}(s_{i},\omega)\geq 0$. Then, any optimal limit-market mechanism induces an indirect utility function $U_{i}$ that has an extension that belongs to the set
    \begin{align*}
        \mathcal{U}_{i}^{*\infty} = \left\{U_{i}:[0,1]\rightarrow\mR\ \middle|\ U_{i} \text{ is increasing, convex, and } \underline{U}_{i}\leq U_{i}\leq \overline{U}_{i}^{\infty} \right\},
    \end{align*}
    where $\overline{U}_{i}^{\infty}(s_{i})=s_{i}$ and $\underline{U}_{i}(s_{i})=\max\{0,2s_{i}-1\}$. 
\end{lemma}
\begin{proof}[Proof of Lemma \ref{lem: optimal indirect utility limit payment}]
    As the payment is nonnegative, the agent obtains the highest payoff by paying nothing and being allocated the good if and only if the state is high. In this case, the interim expected payoff is $U_{i}(s_{i})=\mE[\omega=+1\mid s_{i}]=s_{i}$, which gives the upper bound. To check that $U_{i}$ satisfies the remaining properties, the proofs of Lemmas \ref{lem: feasible mechanism transfer} and \ref{lem: feasible indirect utility transfer} apply verbatim, replacing $s_{-i}$ with $\omega$. 
\end{proof}

The following lemma is a key simplification result for characterizing the optimal limit-market mechanism. 

\begin{lemma}\label{lem: kappa+ and tau- limit payment} 
    Assume $T\geq 1$. Then, there exists an optimal limit-market mechanism $(x_{i},t_{i})$ for agent $i$ such that $\kappa_{+}(s_{i})=1$ and $\tau_{-1}(s_{i})=0$ for each $s_{i}$. 
\end{lemma} 

\begin{proof}[Proof of Lemma \ref{lem: kappa+ and tau- limit payment}]
    Take any feasible limit-market mechanism and suppose that it induces an indirect utility $U_{i}$. By Lemma \ref{lem: optimal indirect utility limit payment}, we have $U_{i}\in \mathcal{U}_{i}^{*}$. 
    
    Note that the expected payoff of type $s_{i}$ from reporting $\hat{s}_{i}$ to a limit-market mechanism is expressed as 
    \begin{align*}
        U_{i}(\hat{s}_{i};s_{i}) 
        = s_{i}\cdot [\kappa_{+1}(\hat{s}_{i})-\tau_{+1}(\hat{s}_{i})] - (1-s_{i})\cdot [\kappa_{-1}(\hat{s}_{i})+\tau_{-1}(\hat{s}_{i})]. 
    \end{align*} 
    Since the indirect utility is convex, it is differentiable almost everywhere. Therefore, the envelope theorem implies that 
    \begin{align*}
        \frac{dU_{i}(s_{i})}{ds_{i}} 
        = \kappa_{+1}(s_{i})-\tau_{+1}(s_{i}) + \kappa_{-1}(s_{i})+\tau_{-1}(s_{i}),
    \end{align*}
    for almost all $s_{i}$. Substituting it, 
    \begin{align*}
        U_{i}(s_{i}) = s_{i}\cdot \frac{dU_{i}(s_{i})}{ds_{i}} - \kappa_{-1}(s_{i})-\tau_{-1}(s_{i}).
    \end{align*}
    Therefore, the objective function conditional on type $s_{i}$ is given by
    \begin{align*}
        & \mE[x_{i}(s_{i},\omega)\mid s_{i}] \\
        & = \int_{\underline{s}}^{\overline{s}}s_{i}\cdot \kappa_{+1}(s_{i})+(1-s_{i})\cdot \kappa_{-1}(s_{i}) d\mF(s_{i}) \\
        & = \int_{\underline{s}}^{\overline{s}}\left\{s_{i}\cdot \kappa_{+1}(s_{i})+(1-s_{i})\cdot \left[-U_{i}(s_{i}) + s_{i}\cdot \frac{dU_{i}(s_{i})}{ds_{i}} -\tau_{-1}(s_{i})\right]\right\}d\mF(s_{i}).
    \end{align*}
    Hence, it is sufficient to show that there exists a limit-market mechanism that induces the same indirect utility $U_{i}(s_{i})$ while $\kappa_{+}(s_{i})=1$ and $\tau_{-1}(s_{i})=0$. As seen in the proof of Lemma \ref{lem: threshold allocation transfer}, if both $U_{i}(s_{i})$ and its slope remains constant, such a limit-market mechanism is also incentive compatible, hence feasible.  

    Let $U_{i}(s_{i})=a_{i}s_{i}-b_{i}$. As discussed in the proof of Theorem \ref{thm: optimal mechanism}, every function in $\mathcal{U}_{i}^{*}$ must be increasing and $2$-Lipschitz, hence $a_{i}\in [0,2]$. Moreover, since $U_{i}$ is convex, $U_{i}(\hat{s}_{i})\geq a_{i}\hat{s}_{i}-b_{i}$ for all $\hat{s}_{i}\in [0,1]$. Then, $U_{i}(\hat{s}_{i})\leq \hat{s}_{i}$ at $\hat{s}_{i}\in \{0,1\}$ implies $b_{i}\geq 0$ and $a_{i}-b_{i}\leq 1$, respectively. Moreover, $U_{i}(s_{i})\geq 0$ and $U_{i}(s_{i})\geq 2s_{i}-1$ imply $a_{i}s_{i}\geq b_{i}$ and $(2-a_{i})s_{i}\leq 1-b_{i}$, and therefore, $a_{i}\geq b_{i}$ and $b_{i}\leq 1$. In summary, we must have $a_{i}\in [0,2]$, $b_{i}\in [0,1]$, and $a_{i}-b_{i}\in [0,1]$. 

    Now, consider a limit-market mechanism at $s_{i}$ such that $\kappa_{+1}(s_{i})=1$, $\kappa_{-1}(s_{i})=b_{i}$, $\tau_{+1}(s_{i})=1-a_{i}+b_{i}$, and $\tau_{-1}(s_{i})=0$. Then, we have
    \begin{align*}
        a_{i} &= \kappa_{+1}(s_{i})-\tau_{+1}(s_{i}) + \kappa_{-1}(s_{i})+\tau_{-1}(s_{i}), \\
        b_{i} &= \kappa_{-1}(s_{i})+\tau_{-1}(s_{i}),
    \end{align*}
    which show that this mechanism induces the same indirect utility and its slope to the agent type $s_{i}$. Moreover, the above discussion ensures that $\kappa_{-1}(s_{i})\in[0,1]$ and $\tau_{+1}(s_{i})\in[0,1]$, i.e., the mechanism is feasible. Note that $\kappa_{+1}(s_{i})=1$ and $\tau_{-1}(s_{i})=0$, which completes the proof. 
\end{proof}

Then, the optimal limit-market mechanism is characterized as follows. It is worth noting that the payment is decreasing in type. 

\begin{lemma}\label{lem: limit-market mechanism payment}
    Assume $T\geq 1$.
    There exists an optimal limit-market mechanism such that, for each agent $i$, there exist two thresholds $s_{i}^{\texttt{min}}$ and $s_{i}^{\texttt{max}}$ such that 
    \begin{align*}
        x_{i}(s_{i},\omega) 
        & = 
        \begin{cases}
            \mI\{\omega=+1\} \quad & \text{if} \quad s_{i} \leq s^{\texttt{min}}_{i} \\
            \mI\{\omega=+1\}+ \kappa_{x}\cdot \mI\{\omega=-1\}\quad & \text{if} \quad s^{\texttt{min}}_{i} \leq s_{i}\leq s^{\texttt{max}}_{i}, \\
            1 \quad & \text{if} \quad s^{\texttt{max}}_{i}\leq s_{i},
        \end{cases} \\
        t_{i}(s_{i},\omega) 
        & = 
        \begin{cases}
            \mI\{\omega=+1\} \quad & \text{if} \quad s_{i} \leq s^{\texttt{min}}_{i} \\
            \kappa_{t}\cdot \mI\{\omega=+1\} \quad & \text{if} \quad s^{\texttt{min}}_{i} \leq s_{i}\leq s^{\texttt{max}}_{i}, \\
            0 \quad & \text{if} \quad s^{\texttt{max}}_{i}\leq s_{i},
        \end{cases}
    \end{align*} 
    for some parameters $\kappa_{x}\in [0,1]$ and $\kappa_{t}\in [0,1]$. 
\end{lemma} 

\begin{proof}[Proof of Lemma \ref{lem: limit-market mechanism payment}]
    Take any optimal limit-market mechanism. Then, by Lemma \ref{lem: kappa+ and tau- limit payment} and the computation in its proof, the objective function is now written as
    \begin{align*}
        \mE[x_{i}(s_{i},\omega)] 
        &= \int_{\underline{s}}^{\overline{s}}\left\{s_{i}+(1-s_{i})\cdot \left[-U_{i}(s_{i}) + s_{i}\cdot \frac{dU_{i}(s_{i})}{ds_{i}}\right]\right\}d\mF(s_{i}).
    \end{align*}
    Therefore, integration by parts implies that the objective function is a linear functional of indirect utility functions $U_{i}$. 

    Now, consider an optimal indirect utility function $U_{i}$ that maximizes the above linear objective over all functions in $\mathcal{U}_{i}^{*}$. As shown in the proof of Theorem \ref{thm: optimal mechanism}, the set $\mathcal{U}_{i}^{*}$ is convex and compact, and if $f$ has bounded derivative, then the objective function is also continuous in indirect utility functions. Therefore, by Bauer's maximum principle, we may take $U_{i}$ to be an extreme point of $\mathcal{U}_{i}^{*}$. It remains to show that $U_{i}$ is implementable by a feasible limit-market mechanism. 

    Hence, Lemma \ref{lem: extreme point large market} implies that there exist two thresholds $s_{i}^{\texttt{min}}$ and $s_{i}^{\texttt{max}}$ such that $U_{i}$ is linear over the middle interval $[s_{i}^{\texttt{min}},s_{i}^{\texttt{max}}]$ and coincides with $\underline{U}_{i}$ over the remaining region. Finally, consider a class of limit-market mechanisms described in the statement of this lemma. Take $a_{i}$ and $b_{i}$ such that the indirect utility function is represented by $U_{i}(s_{i})=a_{i}s_{i}-b_{i}$ over the middle region. Then, define $\kappa_{x}=b_{i}$ and $\kappa_{t}=1-a_{i}+b_{i}$. As discussed in the proof of Lemma \ref{lem: kappa+ and tau- limit payment} above, we must have $\kappa_{x}\in [0,1]$ and $\kappa_{t}\in [0,1]$. Moreover, it follows by construction that the mechanism induces $U_{i}$ as the indirect utility function. Since $U_{i}$ is convex and nonnegative, it is clear that the mechanism is incentive compatible and satisfies participation constraint.  
\end{proof}

Finally, we prove Theorem \ref{thm: large market with payment}. Note that, in finite markets, the problem depends on the shape of the payment rules. Therefore, even with Lemma \ref{lem: limit-market mechanism payment}, it is not sufficient to construct a sequence of finite-market utility functions that converges to the optimal limit-market utility. Instead, the proof constructs a sequence of mechanisms which also satisfy, in the limit, the property stated in Lemma \ref{lem: kappa+ and tau- limit payment}. 

\begin{proof}[Proof of Theorem \ref{thm: large market with payment}]
    Take any optimal limit-market mechanism and the induced indirect utility function $U_{i}$. By Lemma \ref{lem: limit-market mechanism payment}, we can assume that the mechanism has the two-threshold structure given in the statement, for some $s_{i}^{\texttt{min}}$ and $s_{i}^{\texttt{max}}$. Note that $U_{i}$ is linear over the middle interval $[s_{i}^{\texttt{min}},s_{i}^{\texttt{max}}]$ and coincides with $\underline{U}_{i}$ over the remaining region. 

    Let $\delta(n)=1-2/n$ and $\kappa_{t}(n)=\delta(n)\cdot \kappa_{t}$, which converge to $1$ and $\kappa_{t}$ in the limit, respectively. Also, define a sequence $\tau(n)=(1-\kappa_{x}+\kappa_{t}(n))/(1-\kappa_{x}-\kappa_{t}(n))$. 
    
    Then, for each sufficiently large $n$, consider a mechanism
    \begin{align*}
        x_{i}(s_{i},s_{-i};n) 
        & = 
        \begin{cases}
            \mI\{\lr{s_{i},s_{-i}}\geq n-1\} \quad & \text{if} \quad s_{i} \leq s^{\texttt{min}}_{i}(n) \\
            \kappa_{x}+(1-\kappa_{x})\cdot \mI\{\lr{s_{i},s_{-i}}\geq \tau(n)\} \quad & \text{if} \quad s^{\texttt{min}}_{i}(n) \leq s_{i}\leq s^{\texttt{max}}_{i}(n), \\
            1 \quad & \text{if} \quad s^{\texttt{max}}_{i}(n)\leq s_{i},
        \end{cases} \\
        t_{i}(s_{i},s_{-i};n) 
        & = 
        \begin{cases}
            \delta(n)\cdot \mI\{\lr{s_{i},s_{-i}}\geq n-1\} \quad & \text{if} \quad s_{i} \leq s^{\texttt{min}}_{i}(n) \\
            \kappa_{t}(n)\cdot \mI\{\lr{s_{i},s_{-i}}\geq \tau(n)\} \quad & \text{if} \quad s^{\texttt{min}}_{i}(n) \leq s_{i}\leq s^{\texttt{max}}_{i}(n), \\
            0 \quad & \text{if} \quad s^{\texttt{max}}_{i}(n)\leq s_{i}.
        \end{cases}
    \end{align*} 
    Here, we take $s_{i}^{\texttt{min}}(n)\in (0,1)$ and $s_{i}^{\texttt{max}}(n)\in (0,1)$ so that the induced utility function is continuous. In the remaining, we show that this mechanism is incentive compatible and converges to the optimal limit-market mechanism in the limit $n\rightarrow \infty$. 

    To see that the mechanism is incentive compatible, consider first any mechanism of a form $x_{i}(s)=\mI\{\lr{s}\geq \tau\}$ and $t_{i}(s)=\delta x_{i}(s)$ for some $\delta\leq 1$. Then, the interim expected payoff to the agent is given by 
    \begin{align*}
        U_{i}(s_{i})
        =\mE[\omega\cdot x_{i}(s)-t_{i}(s)\mid s_{i}]
        = \mE[(\omega-\delta)\cdot \mI\{\lr{s_{i},s_{-i}}\geq \tau\}\mid s_{i}]. 
    \end{align*}
    Now, note that $\mP[\omega=+1\mid s]=1/[1+\lr{s}^{-1}]$. Therefore, if the state were $\omega-\delta$, then at $\lr{s}=\tau$, the agent's expected payoff from receiving the good, conditional on $s$, is 
    \begin{align*}
        \mE[\omega-\delta\mid \lr{s}=\tau]
        = \frac{\tau}{\tau+1} - \frac{1}{\tau+1} - \delta 
        = \frac{\tau-1}{\tau+1} - \delta. 
    \end{align*} 
    Hence, for each given $\delta$, the expected payoff $U_{i}$ is maximized when $\tau$ is set such that the above expression equals zero, i.e., $\tau=\tau^{*}(\delta)$ with $\tau^{*}(\delta)=[1+\delta]/[1-\delta]$, wherein the good is allocated at type profile $s$ if and only if the agent prefers to receive the good at $s$. Note that this attains the first-best payoff for each agent type $s_{i}$, and therefore, the mechanism is incentive compatible. 

    First, we check that no types have a profitable within-group misreport, for each of the three intervals. 
    Now, note that $\tau^{*}(\delta(n))=n-1$. Hence, by the above paragraph, no bottom-group types $s_{i}\leq s_{i}^{\texttt{min}}(n)$ have an incentive to report $\hat{s}_{i}\leq s_{i}^{\texttt{min}}(n)$. The middle-group types' expected payoffs are given by
    \begin{align*}
        \mE\left[\kappa_{x}+ (1-\kappa_{x})\cdot \left(\omega-\frac{\kappa_{t}(n)}{1-\kappa_{x}}\right)\mI\{\lr{s_{i},s_{-i}}\geq \tau(n)\}\ \middle|\ s_{i}\right].
    \end{align*} 
    Therefore, we have 
    \begin{align*}
        \tau^{*}\left(\frac{\kappa_{t}(n)}{1-\kappa_{x}}\right) 
        = \frac{1-\kappa_{x}+\kappa_{t}(n)}{1-\kappa_{x}-\kappa_{t}(n)} 
        = \tau(n), 
    \end{align*}
    and no middle types have an incentive to misreport within the middle group. It is clear that no top-group types can profitably misreport to another top-group type. 

    Second, we show that the mechanisms converge to the optimal limit-market mechanism. Note that $\lr{s}\geq \tau(n)$ if and only if 
    \begin{align*}
        \frac{1}{n}\cdot \sum_{i} \log \lr{s_{i}} \geq \frac{\log \tau(n)}{n}.
    \end{align*}
    Since private beliefs are conditionally independent, the law of large numbers implies that, conditional on state $\omega$, the left-hand side converges in probability to $\mE[\log \lr{s_{i}}\mid \omega]=\mE[\log s_{i} - \log (1-s_{i})\mid \omega]$. Since the distribution $\mF$ is not degenerate, this is positive when $\omega=+1$ and negative when $\omega=-1$. Therefore, since $\log(n-1)/n\rightarrow 0$ and $\log \tau(n)/n\rightarrow 0$ as $n\rightarrow \infty$, we conclude that both $\mI\{\lr{s}\geq n-1\}$ and $\mI\{\lr{s}\geq \tau(n)\}$ converge in probability to $\mI\{\omega=+1\}$. Therefore, for each region, the mechanisms converge to the limit-market mechanism, which also suggests that $s_{i}^{\texttt{min}}(n)$ and $s_{i}^{\texttt{max}}(n)$ converge to $s_{i}^{\texttt{min}}$ and $s_{i}^{\texttt{max}}$, respectively. It is then obvious to see that the mechanisms converge to the optimal limit-market mechanism. 

    Finally, we need to verify that the mechanism is incentive compatible for each sufficiently large $n$. Let $U^{\texttt{b}}(s_{i};n), U^{\texttt{m}}(s_{i};n), U^{\texttt{t}}(s_{i};n)$ be the utility functions induced by the mechanisms for the bottom group $[0,s_{i}^{\texttt{min}}(n)]$, middle group $[s_{i}^{\texttt{min}}(n),s_{i}^{\texttt{max}}(n)]$, and top group $[s_{i}^{\texttt{max}}(n),1]$, respectively, defined over the entire type space $[0,1]$. Then, in the limit $n\rightarrow \infty$, each of them converges to a linear function, where $U^{\texttt{b}}(s_{i};n)$ has slope $0$, $U^{\texttt{m}}(s_{i};n)$ has a positive slope bounded by $2$, and $U^{\texttt{t}}(s_{i};n)$ has slope $2$. Then, by the constructions of $s_{i}^{\texttt{min}}(n)$ and $s_{i}^{\texttt{max}}(n)$, if $n$ is sufficiently large, we have $U^{\texttt{b}}(s_{i};n)\geq U^{\texttt{m}}(s_{i};n),U^{\texttt{t}}(s_{i};n)$ for all $s_{i}\leq s_{i}^{\texttt{min}}(n)$. Likewise, no type in one group has an incentive to misreport as if her type were in another group. Together with the first argument discussed in an earlier paragraph, this implies that the mechanism is incentive compatible. 
\end{proof}

\section{Proof of Proposition \ref{prop: full surplus extraction}}\label{sec: large market and payment full surplus proof}

\begin{proof}[Proof of Proposition \ref{prop: full surplus extraction}]
    The proof is by construction. Consider a mechanism $(x_{i},t_{i})$ for each agent $i$ such that for each $s$, $x_{i}(s)=1$ and 
    \begin{align*}
        t_{i}(s_{i},s_{-i})=C\cdot \left\{\prod_{j\neq i}s_{j} - \prod_{j\neq i}(1-s_{j})\right\},
    \end{align*}
    where $C\in \mR$ is a constant that is defined below. Note that payments can be negative depending on the profile of private beliefs. 

    Since the mechanism for agent $i$ does not depend on agent $i$'s private belief, it is clearly incentive compatible. Therefore, to show that the above mechanism is feasible, it remains only to verify that the mechanism satisfies the participation constraint. 
    
    To show this, we show that $\mE[t_{i}(s)\mid s_{i}]=2s_{i}-1$ for some constant $C$. For each $s_{i}$, the interim expected payment is computed as
    \begin{align*}
        & \mE[t_{i}(s)\mid s_{i}] \\
        &= s_{i}\int_{s_{-i}}t_{i}(s)\prod_{k\neq i}\frac{f_{+1}(s_{k})}{f(s_{k})}d\mF^{\otimes}(s_{-i})  +  (1-s_{i}) \int_{s_{-i}}t_{i}(s)\prod_{k\neq i}\frac{f_{-1}(s_{k})}{f(s_{k})}d\mF^{\otimes}(s_{-i}) \\
        &= 2^{n-1}C\cdot \left\{ s_{i}\int_{s_{-i}}\prod_{k\neq i}(s_{k})^{2}-\prod_{k\neq i}s_{k}(1-s_{k})d\mF^{\otimes}(s_{-i})\right. \\
        & \qquad \qquad \qquad \qquad \left. + (1-s_{i})\int_{s_{-i}}\prod_{k\neq i}s_{k}(1-s_{k})-\prod_{k\neq i}(1-s_{k})^{2}d\mF^{\otimes}(s_{-i})\right\},
    \end{align*}
    where the first equality uses $f_{+1}(s_{k})/f(s_{k})=2s_{k}$ and $f_{-1}(s_{k})/f(s_{k})=2(1-s_{k})$. Here, note that 
    \begin{align*}
        \int_{s_{k}} (1-s_{k})^{2} d\mF(s_{k}) 
        = \mE[1-2s_{k}] + \int_{s_{k}} (s_{k})^{2} d\mF(s_{k})
        = \int_{s_{k}} (s_{k})^{2} d\mF(s_{k}). 
    \end{align*}
    Therefore, if we set
    \begin{align*}
        C = \left(\frac{1}{2}\right)^{n-1}\left\{\int_{s_{-i}} \prod_{k\neq i}s_{k}(1-s_{k})-\prod_{k\neq i}(1-s_{k})^{2}d\mF^{\otimes}(s_{-i})\right\}^{-1},
    \end{align*}
    then we obtain $\mE[t_{i}(s)\mid s_{i}]=2s_{i}-1$. Note that $C$ does not depend on $s_{i}$ and hence it is a constant. 

    Then, the expected payoff for each type $s_{i}$ is given by
    \begin{align*}
        U_{i}(s_{i}) 
        & = \mE[\omega\cdot x_{i}(s_{i},s_{-i})\mid s_{i}] - \mE[t_{i}(s_{i},s_{-i})\mid s_{i}] \\
        & = \mE[\omega\mid s_{i}] - \mE[t_{i}(s)\mid s_{i}] \\
        & = (2s_{i}-1) - (2s_{i}-1) \\
        & = 0,
    \end{align*}
    and therefore, the participation constraint is satisfied. Finally, by the law of iterated expectations, we have $\mE[t(s)]=\mE[2s_{i}-1]=0$, completing the proof. 
\end{proof}

It is worth noting that the budget-balance property, $\mE[t(s)]=0$, relies on the symmetry of the prior, $\mE[\omega]=0$. In fact, under $x_{i}(s)=1$, the participation constraint \eqref{eq: P with money} requires $\mE[\omega \mid s_{i}] \geq \mE[t_{i}(s)\mid s_{i}]$ for each $i$. Taking expectations on both sides then implies $\mE[\omega] \geq \mE[t_{i}(s)]$. Thus, under a general prior, $\mE[t_{i}(s)]$ must be negative. Therefore, if $\mE[\omega] < 0$, ex ante budget balance may fail to be feasible.


\singlespacing 
\bibliography{Reference}
\addcontentsline{toc}{section}{Reference}

\end{document}